\documentclass[12pt]{article}
\usepackage{amsmath,amssymb,amsthm}
\usepackage{graphics,epsfig}
\usepackage{hyperref}
\usepackage[numbers]{natbib}
\usepackage{color}
\usepackage{graphicx}
\usepackage{caption}
\usepackage{subcaption}
\usepackage{float}
\usepackage{mathrsfs}
\usepackage{booktabs}
\usepackage{multirow}
\usepackage{comment}
\usepackage{setspace}
\usepackage{bbm}
\usepackage{bm}
\usepackage{amsfonts}
\usepackage{xcolor}
\usepackage{pslatex}
\usepackage{setspace}
\usepackage{bbm}
\usepackage{listings}

\def \E {\mathbb{E}}

\def \a {\alpha}
\def \be {\beta}

\newtheorem{definition}{\bf Definition}
\newtheorem{defn}[definition]{\bf Definition}

	\newtheorem{remark}{\bf Remark}
	\newtheorem{rmk}[remark]{\bf Remark}

	\newtheorem{theorem}{\bf Theorem}
	\newtheorem{prop}[theorem]{\bf Proposition}
	\newtheorem{lem}[theorem]{\bf Lemma}
	\newtheorem{cor}[theorem]{\bf Corollary}
	
	\newtheorem{cl}[theorem]{\bf Claim}
    \newtheorem{as}[theorem]{\bf Assumption}

\begin{document}

\title{\bfseries  Bayesian Signaling and Entry Decisions under Uncertain Market Conditions}

 \author{
    Mustapha Nyenye Issah$^{1,2}$ \and Paramahansa Pramanik$^{3,4}$
}

\date{
    \small
    $^{1}$ Department of Mathematics, Texas State University, San Marcos, Texas 78666, United States.\\
    $^{2}$
\texttt{djw215@txstate.edu}\\[0.5em]
    $^{3}$Department of Mathematics and Statistics, University of South Alabama, Mobile, AL 36688, United States.\\
    $^{4}$Corresponding author, \texttt{ppramanik@southalabama.edu}
}

\maketitle

\begin{abstract}
	We develop a continuous-time entry-deterrence game in which market demand evolves according to the Chan-Karolyi-Longstaff-Sanders (CKLS) stochastic differential equation, allowing mean reversion and state-dependent volatility. An incumbent with privately known strength strategically chooses advertising and promotional expenditures to influence a potential entrant's beliefs, while the entrant faces a costly, irreversible entry decision and optimally waits until market conditions justify participation. Within a dynamic Stackelberg setting, Bayesian learning, asymmetric information, stochastic demand, and strategic controls jointly determine entry and signaling behavior. Using a Feynman-type path-integral control formulation, we characterize a Markovian Nash feedback equilibrium for the firms' expenditure strategies. Our contribution is to integrate CKLS demand uncertainty, private information, irreversible entry, Bayesian belief updating, and path-integral feedback control within a unified continuous-time entry-deterrence framework, while providing a computational alternative to direct Hamilton-Jacobi-Bellman (HJB) approach. We illustrate the framework empirically using 2010-2024 revenue data for Enterprise Products Partners and Targa Resources. The resulting trajectories are qualitatively consistent with the model's predictions, exhibiting persistence, recovery after adverse shocks, and distinct responses associated with different competitive positions, while supporting the model's strategic mechanisms under uncertainty.
\end{abstract}


\section{ Introduction.}
In this paper, we formulate a continuous-time model of strategic entry deterrence in which the evolution of market demand is governed by the Chan-Karolyi-Longstaff-Sanders (CKLS) dynamics \citep{chan1992empirical}. The CKLS specification provides a flexible stochastic environment in which market conditions can exhibit mean reversion together with state-dependent volatility, so that both the expected evolution of demand and the magnitude of uncertainty confronting the firms vary with the prevailing state of the market. Against this background, an incumbent firm possesses private information about its underlying strength and chooses advertising and promotional investments strategically, not merely for their contemporaneous effect on market outcomes, but also for the information that these observable actions convey to a potential entrant. Promotional expenditure therefore performs a dual role: it constitutes a costly strategic action for the incumbent while simultaneously serving as a signal from which the entrant attempts to infer the incumbent's privately known type. The potential entrant does not directly observe this private information and must instead form an assessment of the incumbent's strength from prior information and observable market outcomes. In particular, the entrant observes the total output generated by the incumbent, uses this observation together with its prior information to infer the action undertaken by the incumbent, and conditions its subsequent behavior on the resulting assessment. Entry itself requires the payment of a fixed sunk cost, an entry fee that cannot be recovered once the decision to participate has been made which makes the timing of entry economically consequential and gives the entrant an incentive to preserve the option value of waiting. The entrant consequently adopts a Markovian Nash feedback strategy \citep{buckdahn2014value,bensoussan2000stochastic,hamadene1995zero}, under which its decision at each point in time depends on the contemporaneous state of the market and the information available at that state. Rather than entering immediately whenever current demand is favorable, the entrant postpones its irreversible commitment until demand crosses a critical threshold at which the expected long-run profitability of participation is sufficient to compensate for the sunk entry cost and the uncertainty surrounding future market conditions. The incumbent, anticipating that the potential entrant conditions its behavior on observable actions and market outcomes, has an incentive to conceal its true type and to select its promotional policy with the entrant's inference problem explicitly in mind. Thus, current actions affect not only current payoffs but also future beliefs, the perceived profitability of entry, and consequently the timing of competitive participation. The entrant, for its part, continuously evaluates whether the information revealed by the incumbent's behavior, together with the realized evolution of demand, makes entry profitable in the long run. Entry occurs when this assessment favors participation, whereas sufficiently unfavorable demand conditions or a sufficiently strong posterior assessment of the incumbent can sustain continued delay. The resulting economic environment therefore combines irreversible investment, asymmetric information, endogenous signaling, and stochastic market conditions within a single intertemporal strategic problem.

The timing and informational structure of these decisions generate a sequential leader--follower relationship and hence a dynamic Stackelberg game \citep{von2010market} in continuous time: the incumbent acts with knowledge that its observable behavior affects the entrant's subsequent inference and entry decision, while the entrant responds to the evolving market state and to the information contained in the incumbent's actions. The strategic problem is therefore intrinsically dynamic because beliefs, controls, and market demand evolve jointly, and a policy that is optimal at one demand realization need not remain optimal following subsequent shocks. To characterize behavior in this environment, we employ a Feynman-type path-integral formulation \citep{feynman1948space,theodorou2010generalized} of stochastic control \citep{pramanik2020optimization} and construct a Markovian Nash feedback equilibrium in which the firms' investment policies, the entrant's beliefs about the incumbent's private type, and the CKLS demand process co-evolve. The equilibrium formulation makes explicit the feedback between economic fundamentals and strategic information: realizations of demand alter the profitability of entry and the incentives for promotional expenditure; incumbent actions affect the entrant's assessment of market conditions and incumbent strength; and the entrant's evolving assessment, in turn, determines whether continued waiting or irreversible entry is optimal. The path-integral representation provides a probabilistic characterization of this feedback problem and permits the equilibrium controls to be studied without relying exclusively on a direct solution of the corresponding Hamilton--Jacobi--Bellman system. The economic contribution of the framework is therefore to place strategic entry deterrence, privately known incumbent strength, costly signaling through advertising and promotional investment, irreversible entry, state-dependent demand uncertainty, and feedback behavior within a unified continuous-time environment. In this setting, deterrence is not a static outcome generated solely by an incumbent's current action; rather, it emerges endogenously from the interaction among the history of observable behavior, the entrant's inference about the incumbent, the option value of delaying a sunk commitment, and the stochastic evolution of market demand. To assess whether the mechanisms emphasized by the theory generate economically plausible patterns outside the model, we complement the theoretical analysis with an empirical application using data for Enterprise Products Partners and Targa Resources over the period 2010--2024. The observed revenue series are used to represent the evolution of the market environment confronting firms in distinct competitive positions, while the theoretical framework provides the corresponding strategic interpretation of their responses under uncertainty. The empirical analysis yields patterns consistent with the qualitative predictions of the model, providing evidence that the framework can reproduce salient features of strategic behavior in an uncertain market environment and thereby illustrating the practical relevance of the proposed entry-deterrence mechanism.

It is important to emphasize the nature of the present contribution. The path-integral control representation employed in this paper builds upon earlier developments in the corresponding author's work on stochastic control and dynamic strategic interactions \citep{pramanik2020optimization,pramanik2026strategic,pramanik2024motivation,pramanik2024optimization}. Accordingly, the objective of the present study is not to introduce a fundamentally new control methodology. Rather, the contribution lies in extending this framework to a dynamic entry-deterrence environment characterized simultaneously by asymmetric information, irreversible entry, Bayesian learning \citep{kalai1994weak,keller1999optimal,mcclellan2022experimentation}, and state-dependent CKLS demand fluctuations. To our knowledge, these ingredients have not previously been integrated within a unified continuous-time signalling game admitting an explicit path-integral representation of Markovian feedback strategies. The contribution of this paper is therefore threefold. First, we formulate a continuous-time entry-deterrence signalling game under CKLS demand uncertainty in which Bayesian beliefs and market conditions jointly evolve. Second, we characterize Markovian feedback strategies through a path-integral representation that provides a computational alternative to direct HJB implementations in this setting. Third, we illustrate the implications of the theoretical framework using an empirical calibration based on Enterprise Products Partners and Targa Resources. The methodological emphasis is consequently on adapting and integrating existing probabilistic tools to address a novel industrial organization problem rather than on developing a new path-integral control theory.

A key reason for adopting the CKLS specification instead of the standard Geometric Brownian Motion (GBM) \citep{black1973pricing,merton1998applications} or Cox-Ingersoll-Ross (CIR) \citep{cox1985theory} processes is that the strategic incentives of both firms depend critically on how uncertainty varies across market states. Under GBM, the volatility of demand is proportional to the level of demand and there is no mean-reversion, implying that shocks have permanent effects on expected future market conditions. Under CIR, demand exhibits mean-reversion, but the elasticity of volatility with respect to the state variable is fixed. By contrast, the CKLS process allows both mean-reversion and a flexible state-dependent volatility structure through the elasticity parameter. Consequently, the magnitude of uncertainty can increase or decrease with market size at rates that are not restricted by the GBM or CIR specifications.

This distinction is economically important for the entry-deterrence problem considered here. The entrant's decision is based on the expected profitability of market entry, which depends not only on the current demand level but also on the future distribution of demand. Likewise, the incumbent's incentive to undertake costly signaling depends on the probability that current market conditions will persist and on the degree of uncertainty surrounding future demand \citep{pramanik2026stochastic}. Under the CKLS process, periods of high demand may be associated with disproportionately high uncertainty, making entry less attractive despite favorable market conditions. Conversely, when demand is low but volatility is moderate, the option value of waiting may be reduced, altering the entrant's optimal entry threshold. These effects directly influence the intensity and duration of the incumbent's signaling behavior.

Firms routinely determine pricing policies as strategic instruments to shape the competitive landscape, particularly to deter entry or hasten exit of rivals. A prominent illustration is limit pricing, in which an incumbent firm maintains a very low price to make prospective entry appear unprofitable \citep{gryglewicz2023strategic,pramanik2020optimization}. In practice,  such strategic behavior unfolds in environments characterized by intertemporal uncertainty, where incumbents repeatedly adjust prices while confronting volatile demand and noisy market signals \citep{pramanik2023path}. These forward-looking decisions require trading off immediate revenue sacrifices against future strategic advantages, making the incumbent’s willingness to sustain low prices directly responsive to the evolution of demand conditions. Similarly, the entry decision of a prospective entrant hinges on how these fluctuating conditions alter the profitability of participation. The analysis that follows investigates this interplay when market demand is continually perturbed by persistent stochastic shocks \citep{gryglewicz2023strategic}.

 This paper is related to the large industrial organization literature examining how incumbent firms strategically influence the entry decisions of potential competitors. Early contributions emphasized the role of commitment and reputation in shaping market structure over time. In particular, \cite{MilgromRoberts1982} demonstrated that incumbents possessing private information may engage in limit pricing to signal low costs and thereby deter entry. Subsequent work extended these insights to explicitly dynamic settings in which strategic actions unfold intertemporally. \cite{FudenbergTirole1986} developed a theory of exit under incomplete information in which firms manipulate rivals' beliefs through their actions, highlighting the importance of reputational considerations in repeated interaction. Related analyses emphasized that strategic investment, capacity choices, and mobility barriers may alter future competitive conditions by affecting entrants' expectations regarding post-entry profitability \citep{Gilbert1989,OrdoverSaloner1989}. Comprehensive treatments of dynamic oligopoly further formalized these mechanisms using continuous-time and stochastic frameworks, illustrating how commitment, learning, and strategic experimentation jointly determine equilibrium market outcomes \citep{FudenbergTirole1991}. Our framework builds upon this tradition by embedding the entry-deterrence problem within a stochastic differential game in which market conditions evolve endogenously according to a CKLS diffusion process, thereby allowing both the state variable and firms' feedback strategies to evolve continuously under uncertainty. 

The paper is also closely connected to the signalling literature that studies how incumbents use observable actions to convey private information regarding their underlying strength. The seminal contributions of \cite{KrepsWilson1982} and \cite{MilgromRoberts1982} established that strategic behavior may serve a signalling function even when such actions are costly, generating separating or pooling equilibria that influence rivals' beliefs. Within industrial organization, these ideas motivated formal analyses of predatory behavior and aggressive strategic conduct under asymmetric information. For example, \cite{Roberts1986} interpreted predatory pricing as a signalling device designed to discourage entry, while \cite{Saloner1987} examined how incomplete information modifies the incentives for predation and merger activity. More recently, \cite{Bagwell2007} showed that multidimensional signalling environments generate richer patterns of entry deterrence than those implied by one-dimensional models. Relative to this literature, our contribution lies in integrating signalling incentives with a continuous-time stochastic demand environment characterized by state-dependent volatility. In doing so, the model connects Bayesian belief updating, optimal stopping considerations, and Markovian feedback controls within a unified framework, thereby extending the classical signalling approach to settings in which both uncertainty and strategic interaction evolve jointly over time.

We consider a two-player, continuous-time framework where an incumbent has private knowledge of its marginal costs, classifying its type as strong or weak. The potential entrant is uncertain about the incumbent’s cost structure and can choose to enter the market by incurring a one time fixed entry cost. Upon entry, the entrant receives a positive, increasing payoff in the direction with market demand if the incumbent is weak, but earns nothing if the incumbent is strong. A weak incumbent faces a lower payoff after entry but can deter it by engaging in a costly signaling, which can be understood as setting low prices to imitate a strong incumbent. We model market dynamics as a CKLS process since it captures the dynamics of a time series that exhibits mean-reversion and volatility that changes with the level of the variable (i.e., heteroskedasticity) \citep{chan1992empirical,decamps2004investment}. This process was first developed to improve models like the Vasicek or Cox-Ingersoll-Ross (CIR) for state variables by allowing more flexibility in how volatility behaves in response to changes in state variables. It accounts for both the mean-reversion of interest rates and allows the volatility of interest rates to vary with the level of the rate. Moreover, if markets are subject to stochastic demand shocks, the CKLS process can model how these shocks affect state and control variables over time, especially when volatility and mean-reversion are relevant. This is applicable in sectors such as energy markets, commodity pricing, and in sports \citep{pramanik2024motivation}. Finally, in strategic business situations, such as dynamic entry deterrence models, CKLS process can represent market demand or cost structures that evolve stochastically \cite{pramanik2024optimization}. For example, in markets where incumbents and potential entrants react to changing demand conditions, the CKLS process helps in forecasting market conditions under different volatility regimes. We allow this process to include either upward or downward drift, making it applicable to markets which are growing or shrinking.

How does the incumbent signal in such a dynamic environment? How does the entrant strategically time its market entry? Our key finding is that in markets experiencing persistent shocks, the incumbent’s incentive to signal vanishes under certain conditions. The entrant aims to enter when demand is sufficiently high compared to its revised belief that the incumbent is strong. Conversely, when demand drops low enough, the threat of entry diminishes in the short term, leading a weak incumbent to forgo costly signaling and instead reveal its type by increasing prices \citep{carmona2015forward}. 

We propose a Markovian Nash feed-back equilibrium of strategies between the incumbent and entrant. One strategy might be the total amount of expenditure on promotions of commercials. This is a type of closed loop strategy (i.e.,control) since the optimal strategy is determined by the current market size $X\in\mathcal X$ and time $s$. We define firm $i^{th}$ strategy as $u^i(s):=\left\{\tilde{u}^i(s,X)|0\leq s\leq t,X\in\mathcal X\right\}$, for all $i=1,2$. In feed-back strategy the decisions might be revised based on new information embodied in current market size. We propose a Feynman-type path integral control method to determine these strategies. The benefit of this approach is that it bypasses the complex calculation of the value function. However, when the market size is extremely large, solving the Hamiltonian-Jacobi-Bellman (HJB) equation numerically becomes highly challenging. As a result, method like the Feynman-type path integral approach is useful. Additionally, if the market dynamics are nonlinear, finding an optimal equilibrium via the HJB equation is not feasible. In such cases, the Feynman-type path integral method proves effective.

The remainder of the paper is organized as follows. Section~2 formulates the
dynamic entry-deterrence model, introduces the CKLS state dynamics, describes
the Bayesian belief process and firms' payoff structures, and presents the
parameter calibration together with a verification of the discount condition
under the calibrated equilibrium. Section~3 develops the probabilistic
construction of the controlled CKLS process, establishing the associated
semigroup representation, local-time decomposition, and excursion-theoretic
properties. These results provide the probabilistic foundation of the state
process but are not used directly in the subsequent derivation of the
equilibrium controls. Section~4 establishes the existence, uniqueness, and
verification of the Markovian Nash feedback equilibrium by deriving the
equilibrium controls from the firms' optimization problems and the
corresponding first-order optimality conditions. Section~5 discusses the
computational advantages of the Feynman-type path-integral control relative to the HJB equation, particularly for high-dimensional and
nonlinear stochastic control problems. Section~6 presents the numerical
implementation of the equilibrium model and computes the optimal expenditure
trajectories implied by the calibrated feedback strategies. Section~7
conducts the empirical analysis using revenue data for Enterprise Products
Partners and Targa Resources, comparing the model-generated trajectories with
the observed market dynamics and evaluating the empirical performance of the
calibrated model. Finally, Section~8 concludes with a summary of the main
findings and discusses several directions for future research.

\section{ Model Formulation.}
This paper models a dynamic strategic interaction between a leader and a follower firms over continuous-time. Firm 1 (the incumbent) possesses private information about its type, either strong or weak, and seeks to convince Firm 2 (the potential entrant) that it is strong in order to deter entry. The weak incumbent strategically decides when to reveal its true type, while the entrant continuously observes public signals, particularly market demand, and chooses the optimal time to enter. Both of the firms face optimal stopping problems, balancing the trade-off between the timing of their actions and the informational consequences of those decisions.

This model is constructed on a probability space \( (\Omega, \mathcal{F}, \mathcal{P}) \), constructed as a product space of a standard Brownian motion \( (\tilde{\Omega}, \tilde{\mathcal{F}}, \tilde{\mathcal{P}}) \) and a binary type space \( \Theta = \{q, w\} \), where \( q \) represents a strong incumbent and \( w \) as weak. Here \( \Omega \) and \( \mathcal{F} \) represent the sample space and the associated Borel sigma-algebra, respectively. The prior belief about the incumbent being strong is denoted by \( p_0 = \eta(\theta = q) \). Formally, the space is defined as \( \Omega = \tilde{\Omega} \times \Theta \), with \( \mathcal{F} = \tilde{\mathcal{F}} \times 2^\Theta \), and \( \mathcal{P} = \tilde{\mathcal{P}} \times \eta \). We assume the state variable as market size \( X(s) \), represents a publicly observable demand proxy expressed in terms of revenue. It is assumed to evolve according to a CKLS stochastic differential equation (SDE),
\begin{equation}\label{0}
dX(s) = \left[\alpha_1 + \alpha_2 X(s) + \sum_{i=1}^2 \theta_i u^i(s)\right]ds + \alpha_3 [X(s)]^{\alpha_4} dB(s),
\end{equation}
where \( \alpha_j \) and \( \theta_i \) are constants, \( u^i(s) \) represents Firm \( i \)'s advertising expenditure (i.e., a Markovian control), and \( B(s) \) is a standard two-dimensional Brownian motion on \( (\Omega^X, \mathcal{F}^X, \mathcal{P}^X) \). The filtration \( \mathcal{F}_s^X = \sigma\left(X(\nu), Y_0 \mid \nu \in [0, \tau]\right) \) captures the information available up to time \( s \), and \( Y_0 \) is an independent random variable realized at time \( s = 0 \). We further assume the firms update their beliefs \( p(s) \) over time using Bayesian updating \citet{mcclellan2022experimentation}.

 The model assumes that the incumbent's privately observed characteristic belongs to the binary type space
	\[
	\Theta=\{q,w\},
	\]
	where \(q\) and \(w\) denote strong and weak incumbents, respectively. This assumption follows the classical signalling literature in industrial organization, where binary types are frequently adopted to isolate the strategic interaction generated by asymmetric information while preserving analytical tractability \citep{milgrom1982limit}. In particular, the distinction between a strong incumbent, against whom entry is unprofitable, and a weak incumbent, for whom entry may be profitable, captures the economically relevant margin underlying entry deterrence.
	
	The binary specification permits an explicit characterization of posterior beliefs through the log-likelihood ratio
	\[
	Z(s)=\log\left(\frac{p(s)}{1-p(s)}\right),
	\]
	which transforms Bayesian updating into a one-dimensional state variable evolving jointly with the CKLS demand process. This reduction plays a central role in the derivation of the Markovian feedback equilibrium and the path-integral representation developed in Sections 3 and 4. Under a richer type space, the posterior distribution itself becomes infinite-dimensional, substantially complicating both the equilibrium characterization and the associated control problem.

Although analytically convenient, the binary formulation should not be interpreted as essential for the qualitative conclusions of the paper. Suppose instead that the incumbent's hidden characteristic is described by a finite collection of types, $\Theta=\{\theta_1,\ldots,\theta_K\},$ $ K<\infty,$ with prior probabilities
\[
\pi_i(s)=\mathbb{P}\!\left(\theta=\theta_i\,\middle|\,\mathcal{F}_s^X\right),
\qquad i=1,\ldots,K.
\]
The entrant's beliefs would then evolve according to a \((K-1)\)-dimensional filtering process constrained to the probability simplex. More generally, if the incumbent's type were continuously distributed with prior measure \(\mu_0\), Bayesian learning would be governed by a stochastic filtering equation for the posterior measure \(\mu_s\). In either case, the fundamental economic mechanisms emphasized by the present model remain unchanged. Specifically, (i) stronger incumbents have weaker incentives to distort behavior through costly signalling, (ii) entrants optimally trade off current profitability against the option value of waiting under uncertainty, and (iii) state-dependent volatility generated by the CKLS dynamics modifies both signalling incentives and entry thresholds through its effect on the evolution of beliefs and future demand. Consequently, the qualitative implications of the analysis derive from the interaction between asymmetric information, irreversible entry, and stochastic demand fluctuations rather than from the binary nature of the type space itself.

\begin{table}[H]
\centering
\caption{CKLS parameter estimates with standard errors and 95\% confidence intervals.}
\label{tab:cklsparams}
\begin{tabular}{lccccc}
\toprule
\textbf{Parameter} & \textbf{Symbol} & \textbf{Estimate} & \textbf{Std. Error} & \textbf{95\% CI} & \textbf{Status} \\
\midrule
Drift intercept & $\alpha_1$ & 0.368491  & 0.487268 & $(-0.586554,1.323536)$ & Estimated \\
Mean-reversion coefficient & $\alpha_2$ & -0.255700 & 0.052318 & $(-0.358243,-0.153157)$ & Estimated \\
Volatility scale & $\alpha_3$ & 0.913034 & 0.212462 & $(0.496608,1.329460)$ & Estimated \\
Volatility elasticity & $\alpha_4$ & 0.166201 & 0.348171 & $(-0.516214,0.848616)$ & Estimated \\
Control sensitivity (Firm 1) & $\theta_1$ & 0.000050 & --- & --- & Calibrated \\
Control sensitivity (Firm 2) & $\theta_2$ & 0.993310 & --- & --- & Calibrated \\
\bottomrule
\end{tabular}
\end{table}

Table \ref{tab:cklsparams} represents the estimated values of the parameters using the maximum likelihood estimation combined with the Euler Maruyama discretization of the CKLS SDE. Moreover, the parameters reported in Table~\ref{tab:cklsparams} should be interpreted as calibration values chosen to reproduce the broad empirical characteristics of the observed revenue series rather than as structural estimates obtained through formal econometric identification. First, the estimated mean-reversion coefficient, \(\alpha_2=-0.255700\), indicates that deviations in normalized revenue tend to dissipate over time, while the positive volatility scale parameter \(\alpha_3=0.913034\) and elasticity parameter \(\alpha_4=0.166201\) imply state-dependent stochastic fluctuations consistent with the CKLS specification. Second, the estimated sensitivity associated with the entrant's expenditure, \(\theta_2=0.993310\), is substantially larger than that associated with the incumbent, \(\theta_1=0.000050\). As the calibrated value $\theta_1=0.000050$ implies that $C=\theta_1\alpha_2^3s^3
\exp\left\{\alpha_2sx-k\right\}$ is several orders of magnitude smaller than $1$ throughout the calibrated state space, the local normalization \(C^2\approx1\) is not imposed in the empirical analysis.
	Instead, all equilibrium controls reported in the numerical
	section are obtained by directly solving the full nonlinear
	first-order optimality system using the iterative algorithm
	described in Section~4. Within the present calibration exercise, this suggests that variation in the observed revenue trajectories is primarily associated with the entrant-side control channel, whereas the direct contribution of the incumbent's advertising expenditure to the fitted drift component is empirically negligible \citep{anderson2026obesity}. Accordingly, the empirical exercise is intended to illustrate the qualitative implications of the theoretical model rather than to provide a statistically identified estimation of the CKLS process. Although the CKLS specification is capable of structural estimation, the objective of the present empirical exercise is illustrative rather than econometric. The parameter values reported in Table~\ref{tab:cklsparams} are employed to generate simulated trajectories that capture salient features of the observed revenue dynamics. Consequently, the analysis should be interpreted as an empirical validation of the model's qualitative implications rather than a formal statistical test of the theory.

The calibrated control-sensitivity parameters indicate a pronounced asymmetry between the two firms. In particular, \(\theta_1=0.000050\) is close to zero, implying that Firm~1's control has a negligible direct effect on the fitted drift of the revenue process within the present calibration. By contrast, \(\theta_2=0.993310\) assigns a much larger direct sensitivity to Firm~2's control. This asymmetry should not be interpreted as structural evidence that the incumbent's strategic actions are generally irrelevant. Rather, it reflects the reduced-form calibration used in the empirical illustration and indicates that the observed variation in normalized revenue is captured primarily through the Firm~2 control channel \citep{reed2026structural}. The theoretical model continues to allow strategic signalling by the incumbent through the equilibrium feedback problem, but the empirical calibration does not provide strong statistical evidence of a quantitatively large incumbent-control effect.

\subsection{Verification of the conditions of $\gamma$  under the calibrated equilibrium.}
\label{subsec:discount_verification}

The existence of the discounted payoff functionals established in
Condition~\eqref{fin} requires the controlled drift of the state process to
remain uniformly below the discount parameter \citep{pramanik2026optimal}. Since the empirical section
employs calibrated parameter values together with numerically computed
equilibrium controls, it is necessary to verify that this condition continues
to hold along the calibrated equilibrium path. For the equilibrium controls
\(
(u^{1,*}(s),u^{2,*}(s))
\),
define the controlled drift
\begin{equation}\label{cdrift}
\mu^{*}(s)
:=
\alpha_{1}
+\alpha_{2}X(s)
+\theta_{1}u^{1,*}(s)
+\theta_{2}u^{2,*}(s),
\qquad
\text{for all}\ s\in[0,t].
\end{equation}
Condition~\eqref{fin} is equivalent to requiring $\inf_{s\in[0,t]}
\left[
\gamma-\mu^{*}(s)
\right]
>0.$
We define $\Delta_{\gamma}(s)
:=
\gamma-\mu^{*}(s)$ as the \emph{discount margin}. A strictly positive lower
bound for \(\Delta_{\gamma}\) guarantees that the payoff denominators remain
uniformly separated from zero. Throughout the numerical implementation the normalized state and the logistically transformed equilibrium controls satisfy
\(
X(s)\in[0,1],\) \(u^{1,*}(s)\in[0,1],\) and \(u^{2,*}(s)\in[0,1].
\)
The coefficients represented in
Table~\ref{tab:cklsparams} are
\(
\alpha_{1}=0.368491,
\alpha_{2}=-0.255700,
\theta_{1}=0.000050,\) and \(
\theta_{2}=0.993310\), respectively.
	\begin{prop}
		\label{prop:discount_margin}
		Consider the controlled drift evaluated in Equation \eqref{cdrift},
		and let
		\(
		\pi_{N}
		=
		\{0=s_{0}<s_{1}<\cdots<s_{N}=t\}
		\)
		be the discretization grid used in the empirical implementation. Define
		\(
		\widehat{\mu}_{\max}
		:=
		\max_{0\leq j\leq N}\mu^{*}(s_{j})
		\)
		and suppose that
		\(
		\sup_{s\in[0,t]}
		\left|
		\mu^{*}(s)-\mu^{*}\bigl(\eta_{N}(s)\bigr)
		\right|
		\leq
		\varepsilon_{N},
		\) for all \(
		\eta_{N}(s)
		:=
		\max\{s_{j}:s_{j}\leq s\}.
		\)
		If \(
		\widehat{\mu}_{\max}+\varepsilon_{N}<0.8794,
		\)
		then, for \(\gamma=0.8794\),
		\[
		\inf_{0\leq s\leq t}
		\left\{
		\gamma-\mu^{*}(s)
		\right\}
		\geq
		0.8794-\widehat{\mu}_{\max}-\varepsilon_{N}
		>0.
		\]
		Consequently, the payoff denominators remain strictly positive along the
		entire calibrated continuous-time trajectory.
	\end{prop}

\begin{proof}
	Fix the calibrated parameter vector
	\[
	\widehat{\vartheta}
	=
	\bigl(
	\widehat{\alpha}_{1},
	\widehat{\alpha}_{2},
	\widehat{\theta}_{1},
	\widehat{\theta}_{2}
	\bigr)
	=
	\bigl(
	0.368491,-0.255700,0.000050,0.993310
	\bigr),
	\]
	and let
	$\pi_{N}
	=
	\{s_{0},s_{1},\ldots,s_{N}\},$ $
	0=s_{0}<s_{1}<\cdots<s_{N}=t,
	$
	denote the temporal partition used in the numerical implementation. Define
	$
	|\pi_{N}|
	:=
	\max_{0\leq j<N}(s_{j+1}-s_{j})
	$
	for its mesh size. Now, for each grid point \(s_{j}\), define
	\[
	\mu^{*}_{j}
	:=
	\mu^{*}(s_{j})
	=
	\widehat{\alpha}_{1}
	+
	\widehat{\alpha}_{2}X(s_{j})
	+
	\widehat{\theta}_{1}u^{1,*}(s_{j})
	+
	\widehat{\theta}_{2}u^{2,*}(s_{j}).
	\]
	Then associated discount margin is $\Delta_{j}
	:=
	\gamma-\mu^{*}_{j},$ $
	j=0,\ldots,N.$ More precisely, let \((\Omega,\mathcal{F},
	(\mathcal{F}_{s})_{0\leq s\leq t},\mathbb{P})\) be the filtered probability
	space on which the controlled CKLS state process and the equilibrium controls
	are defined \citep{pramanik2026quantum}. We work on an event \(\Omega_{0}\in\mathcal{F}\) of probability
	one on which the simulated state values \(X(s_{j})\) and the recursively
	computed feedback controls \(u^{1,*}(s_{j})\) and \(u^{2,*}(s_{j})\) are
	finite for every \(j=0,\ldots,N\). Such an event exists under the standing
	well-posedness and admissibility assumptions imposed on the state equation
	and the feedback controls. Fix \(\omega\in\Omega_{0}\). All quantities below
	are then finite real numbers, and the proof reduces to a deterministic
	finite-dimensional argument. To avoid an unnecessarily cumbersome notation,
	the dependence on \(\omega\) is suppressed.
	
	Since the index set
	$\{0,1,\ldots,N\}
	$
	is finite, the collection
	$
	\{\mu^{*}_{0},\mu^{*}_{1},\ldots,\mu^{*}_{N}\}
	$
	admits a maximum. Accordingly,
	$
	\widehat{\mu}_{\max}
	:=
	\max_{0\leq j\leq N}\mu^{*}_{j}
	$
	is well defined and is attained at at least one index
	\(j^{*}\in\{0,\ldots,N\}\). Thus,
	$
	\mu^{*}_{j^{*}}
	=
	\widehat{\mu}_{\max}.
	$
	By the assumption of the proposition,
	$
	\widehat{\mu}_{\max}<0.8794.
	$
As the discount parameter is fixed at
	$
	\gamma=0.8794,
	$
	the preceding inequality is equivalent to
	$
	\gamma-\widehat{\mu}_{\max}>0.
	$
We now identify the minimum discount margin over the numerical grid. Since
	the mapping
	$
	x\mapsto\gamma-x
	$
	is strictly decreasing on \(\mathbb{R}\), maximization of the drift is
	equivalent to minimization of the corresponding discount margin. Therefore,
	\[
	\begin{aligned}
		\min_{0\leq j\leq N}\Delta_{j}
		&=
		\min_{0\leq j\leq N}
		\left\{
		\gamma-\mu^{*}_{j}
		\right\}=
		\gamma-
		\max_{0\leq j\leq N}\mu^{*}_{j}=
		\gamma-\widehat{\mu}_{\max}.
	\end{aligned}
	\]
	For completeness, this identity can be verified directly. Since
	$
	\mu^{*}_{j}\leq\widehat{\mu}_{\max}
	$
	for every \(j\), subtraction from \(\gamma\) gives
	$
	\gamma-\mu^{*}_{j}
	\geq
	\gamma-\widehat{\mu}_{\max}.
	$
	Consequently,
	$
	\min_{0\leq j\leq N}
	\left\{
	\gamma-\mu^{*}_{j}
	\right\}
	\geq
	\gamma-\widehat{\mu}_{\max}.
	$
	On the other hand, at an index \(j^{*}\) for which
	\(\mu^{*}_{j^{*}}=\widehat{\mu}_{\max}\),
	\[
	\gamma-\mu^{*}_{j^{*}}
	=
	\gamma-\widehat{\mu}_{\max}.
	\]
	Hence, equality holds and
	$
	\min_{0\leq j\leq N}
	\left\{
	\gamma-\mu^{*}_{j}
	\right\}
	=
	\gamma-\widehat{\mu}_{\max}.$	Define
	$
	\widehat{\delta}_{N}
	:=
	\gamma-\widehat{\mu}_{\max}.
	$
	The assumed strict inequality
	$
	\widehat{\mu}_{\max}<\gamma
	$
	implies
	$
	\widehat{\delta}_{N}>0.
	$
	It immediately follows that
	\[
	\Delta_{j}
	=
	\gamma-\mu^{*}_{j}
	\geq
	\widehat{\delta}_{N}>0,
	\qquad
	j=0,\ldots,N.
	\]
	Equivalently,
	\[
	0<
	\widehat{\delta}_{N}
	\leq
	\gamma-
	\left[
	\widehat{\alpha}_{1}
	+
	\widehat{\alpha}_{2}X(s_{j})
	+
	\widehat{\theta}_{1}u^{1,*}(s_{j})
	+
	\widehat{\theta}_{2}u^{2,*}(s_{j})
	\right]
	\]
	for every point of the calibration grid. Substituting the calibrated
	coefficients yields
	\[
	\begin{aligned}
		0<
		\widehat{\delta}_{N}
		\leq
		0.8794
		-
		\Bigl[
		&0.368491
		-0.255700X(s_{j})+0.000050u^{1,*}(s_{j})
		+0.993310u^{2,*}(s_{j})
		\Bigr],
	\end{aligned}
	\]
	for all \(j=0,\ldots,N\). This establishes more than \emph{pointwise positivity} at each separately inspected
	date. The same deterministic constant
	\(\widehat{\delta}_{N}\) is a common lower bound for every denominator on the
	entire numerical grid. In particular,
	\[
	\inf_{0\leq j\leq N}
	\left\{
	\gamma-\mu^{*}(s_{j})
	\right\}
	=
	\widehat{\delta}_{N}>0.
	\]
	Hence, the numerical payoff denominators are uniformly separated from zero.
	The latter statement also yields a uniform reciprocal bound. Namely,
	\[
	0<
	\frac{1}{\gamma-\mu^{*}(s_{j})}
	\leq
	\frac{1}{\widehat{\delta}_{N}},
	\qquad
	j=0,\ldots,N.
	\]
	More generally, for every \(p>0\),
	\[
	0<
	\frac{1}{
		\left[
		\gamma-\mu^{*}(s_{j})
		\right]^{p}}
	\leq
	\widehat{\delta}_{N}^{-p}.
	\]
	Thus any numerical payoff term containing a finite positive power of the
	reciprocal discount margin is finite at every grid point. In particular, if
	the terminal or continuation component of firm \(i\)'s payoff contains a
	factor of the form
	\[
	\frac{G^{i}
		\left[
		s_{j},X(s_{j}),u^{1,*}(s_{j}),u^{2,*}(s_{j})
		\right]}
	{\gamma-\mu^{*}(s_{j})},
	\]
	then,
	\[
	\left|
	\frac{G^{i}
		\left[
		s_{j},X(s_{j}),u^{1,*}(s_{j}),u^{2,*}(s_{j})
		\right]}
	{\gamma-\mu^{*}(s_{j})}
	\right|
	\leq
	\frac{
		\left|
		G^{i}
		\left[
		s_{j},X(s_{j}),u^{1,*}(s_{j}),u^{2,*}(s_{j})
		\right]
		\right|}
	{\widehat{\delta}_{N}}.
	\]
	Consequently, whenever the numerator is finite on the calibrated grid, the
	corresponding payoff contribution is finite as well. The same conclusion applies to a quadrature approximation of a running
	payoff. To see this, let
	$
	0=w_{-1},$ with $ w_{j}\geq0,$
	denote the numerical integration weights associated with the partition
	\(\pi_{N}\), and suppose the discretized payoff of firm \(i\) contains the
	term
	\[
	J^{i}_{N}
	=
	\sum_{j=0}^{N}
	w_{j}
	e^{-\gamma s_{j}}
	\frac{
		F^{i}
		\bigl[
		s_{j},X(s_{j}),u^{1,*}(s_{j}),u^{2,*}(s_{j})
		\bigr]
	}{
		\gamma-\mu^{*}(s_{j})
	}.
	\]
	Then
	\[
	\begin{aligned}
		|J^{i}_{N}|
		&\leq
		\sum_{j=0}^{N}
		w_{j}e^{-\gamma s_{j}}
		\frac{
			\left|
			F^{i}
			\bigl[
			s_{j},X(s_{j}),u^{1,*}(s_{j}),u^{2,*}(s_{j})
			\bigr]
			\right|
		}{
			\gamma-\mu^{*}(s_{j})
		}\\
		&\leq
		\frac{1}{\widehat{\delta}_{N}}
		\sum_{j=0}^{N}
		w_{j}e^{-\gamma s_{j}}
		\left|
		F^{i}
		\bigl[
		s_{j},X(s_{j}),u^{1,*}(s_{j}),u^{2,*}(s_{j})
		\bigr]
		\right|.
	\end{aligned}
	\]
	The sum on the right-hand side is finite because it contains only finitely
	many finite terms. Hence the discretized payoff is well defined. It remains to clarify the probabilistic status of the argument. Since the
	proof was carried out for an arbitrary \(\omega\in\Omega_{0}\), the conclusion
	holds pathwise on \(\Omega_{0}\). Therefore,
	\[
	\mathbb{P}
	\left[
	\min_{0\leq j\leq N}
	\left\{
	\gamma-\mu^{*}\left(s_{j}\right)
	\right\}
	>0
	\right]
	=1,
	\]
	so that
	$
	\widehat{\mu}_{\max}<\gamma
	$
	has been verified for every simulated path under consideration. If the
	empirical implementation is based on one realized calibrated trajectory,
	then the proposition is a deterministic ex post verification for that
	trajectory. If it is based on \(M\) simulated trajectories, indexed by
	\(m=1,\ldots,M\), the same argument applies after defining
	\[
	\widehat{\mu}_{\max}^{(M,N)}
	:=
	\max_{\substack{1\leq m\leq M\\0\leq j\leq N}}
	\mu^{*,m}(s_{j}).
	\]
	The condition
	$
	\widehat{\mu}_{\max}^{(M,N)}<\gamma
	$
 implies the common bound
	\[
	\gamma-\mu^{*,m}(s_{j})
	\geq
	\gamma-\widehat{\mu}_{\max}^{(M,N)}
	>0
	\]
	for all simulated paths \(m\) and all grid indices \(j\). The proposition
	establishes strict positivity at every point of the numerical grid
	\(\pi_{N}\). It does not, by itself, imply that
	$
	\gamma-\mu^{*}(s)>0
	$
	at every continuous time \(s\in[0,t]\), because a continuous-time drift could,
	in principle, attain a larger value between two adjacent grid points. A
	continuous-time conclusion requires an additional control of the
	within-cell discretization error. Such a conclusion follows, for example, if
	there exists a deterministic constant \(\varepsilon_{N}\geq0\) such that
	\[
	\sup_{s\in[0,t]}
	\left|
	\mu^{*}(s)-\mu^{*}\bigl(\eta_{N}(s)\bigr)
	\right|
	\leq
	\varepsilon_{N},
	\]
	where
	$
	\eta_{N}(s)
	:=
	\max\{s_{j}:s_{j}\leq s\}
	$
	is the left-endpoint projection onto the grid. Under this additional bound,
	\[
	\begin{aligned}
		\mu^{*}(s)
		&\leq
		\mu^{*}\bigl(\eta_{N}(s)\bigr)
		+
		\left|
		\mu^{*}(s)-\mu^{*}\bigl(\eta_{N}(s)\bigr)
		\right|\leq
		\widehat{\mu}_{\max}+\varepsilon_{N},
	\end{aligned}
	\]
	and thus
	\[
	\gamma-\mu^{*}(s)
	\geq
	\gamma-\widehat{\mu}_{\max}-\varepsilon_{N}
	=
	\widehat{\delta}_{N}-\varepsilon_{N}.
	\]
	Therefore, if
	$
	\varepsilon_{N}<\widehat{\delta}_{N},
	$
	then
	\[
	\inf_{0\leq s\leq t}
	\left\{
	\gamma-\mu^{*}(s)
	\right\}
	\geq
	\widehat{\delta}_{N}-\varepsilon_{N}
	>0.
	\]
	This strengthened estimate converts the discrete verification into a
	continuous-time pathwise verification. Under the assumptions stated in the proposition, however, the required
	gridwise conclusion follows immediately
	\[
	\min_{0\leq j\leq N}
	\left\{
	\gamma-\mu^{*}(s_{j})
	\right\}
	=
	0.8794-\widehat{\mu}_{\max}
	>0.
	\]
	Consequently, every payoff denominator evaluated in the calibrated numerical implementation is strictly positive, and the corresponding discretized payoff expressions are well defined. this completes the proof.
\end{proof}

\begin{remark}
	\label{rem:discount_margin}
	
	Proposition~\ref{prop:discount_margin} is a verification result for the
	calibrated equilibrium trajectory and should not be interpreted as a uniform
	statement over the full admissible state-control region. In particular, the
	choice
	$
	\gamma=0.8794
	$
	is justified only if the controlled drift computed from the calibrated state
	and equilibrium controls satisfies
	\[
	\widehat{\mu}_{\max}
	=
	\max_{0\leq j\leq N}
	\left\{
	\alpha_{1}
	+\alpha_{2}X(s_{j})
	+\theta_{1}u^{1,*}(s_{j})
	+\theta_{2}u^{2,*}(s_{j})
	\right\}
	<0.8794.
	\]
	The numerical proximity of \(\gamma\) to \(1\) has no mathematical relevance
	for the finiteness condition. What is relevant is the strictly positive
	separation
	$
	\widehat{\delta}_{N}
	:=
	0.8794-\widehat{\mu}_{\max}>0.
	$
	This quantity measures the smallest discount margin observed on the numerical
	grid and provides the deterministic bound
	\[
	\frac{1}{\gamma-\mu^{*}(s_{j})}
	\leq
	\frac{1}{\widehat{\delta}_{N}},
	\qquad
	j=0,\ldots,N.
	\]
	Thus, a small positive value of \(\widehat{\delta}_{N}\) is sufficient for
	well-definedness of the discretized payoff representation, although it may
	also indicate numerical sensitivity because the reciprocal bound becomes
	large as \(\widehat{\delta}_{N}\downarrow0\). The distinction between gridwise and continuous-time verification is
	important. The condition
	$
	\max_{0\leq j\leq N}\mu^{*}(s_{j})<\gamma
	$
	establishes positivity only at the dates contained in the partition
	\(\pi_{N}\). It does not exclude the possibility that the continuous-time
	drift exceeds \(\gamma\) between adjacent grid points. A continuous-time
	conclusion follows if the within-cell approximation error can be controlled.
	For example, if
	\[
	\sup_{s\in[0,t]}
	\left|
	\mu^{*}(s)-\mu^{*}\bigl(\eta_{N}(s)\bigr)
	\right|
	\leq\varepsilon_{N},
	\]
	where
	$
	\eta_{N}(s)
	=
	\max\{s_{j}:s_{j}\leq s\},
	$
	then
	$
	\inf_{0\leq s\leq t}
	\left\{
	\gamma-\mu^{*}(s)
	\right\}
	\geq
	\widehat{\delta}_{N}-\varepsilon_{N}.
	$
	Consequently, the stronger condition
	$
	\varepsilon_{N}<\widehat{\delta}_{N}
	$
	ensures strict positivity of the discount margin along the entire
	continuous-time trajectory. Proposition~\ref{prop:discount_margin} is also conditional on the reported calibration and the
	computed equilibrium controls. It does not imply that
	\(\gamma=0.8794\) satisfies the discount condition for arbitrary admissible
	values of \(X\), \(u^{1}\), and \(u^{2}\). Indeed, over the unrestricted
	normalized domain
	$
	X,u^{1},u^{2}\in[0,1],
	$
	the deterministic upper bound obtained from the calibrated coefficients is
	$
	\alpha_{1}+\theta_{1}+\theta_{2}
	=
	1.361851,
	$
	which exceeds \(0.8794\). Hence, a uniform state controlled domain result cannot be deduced for this value of \(\gamma\). The role of
	Proposition~\ref{prop:discount_margin} is therefore narrower; it verifies
	that the denominator remains positive along the particular calibrated path,
	or collection of simulated paths, used in the empirical implementation.
	For numerical robustness, the empirical analysis should report not only
	\(\gamma\) and \(\widehat{\mu}_{\max}\), but also the resulting margin
	$
	\widehat{\delta}_{N}
	=
	\gamma-\widehat{\mu}_{\max}.
	$
	It is also advisable to repeat the calculation on progressively finer
	partitions. Stability of \(\widehat{\mu}_{\max}\) and
	\(\widehat{\delta}_{N}\) under grid refinement provides evidence that the
	verified positivity is not an artifact of a coarse time discretization.
\end{remark}

Accordingly, the empirical results should not be interpreted as providing statistical evidence that incumbent signaling expenditures exert a quantitatively important effect on realized revenues. Rather, the calibration is intended to demonstrate that the theoretical CKLS signalling framework is capable of reproducing salient qualitative features of the observed revenue dynamics while accommodating asymmetric strategic sensitivities. The strategic mechanisms developed in Sections~3 and~4 remain theoretical implications of the stochastic differential game and should therefore be distinguished from the reduced-form empirical calibration reported here. Establishing statistically identified signalling effects would require richer firm-level expenditure data and a structural estimation procedure designed specifically for inference on the control parameters.

\begin{figure}[H]
    \centering
    \includegraphics[width=0.9\textwidth]{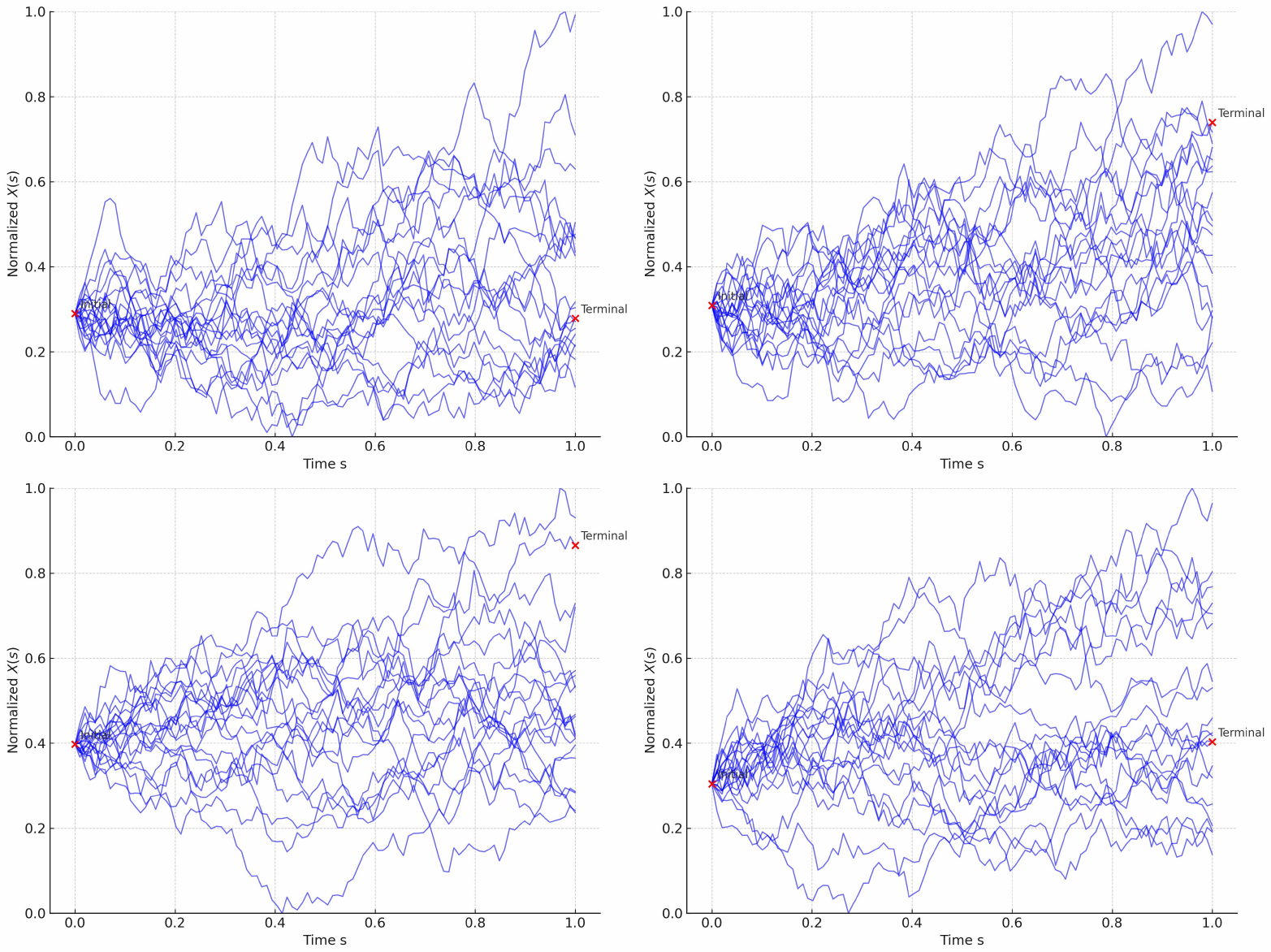}
    \caption{Normalized simulations of the CKLS SDE for four advertising expenditure settings.}
    \label{fig:ckls-trajectories}
\end{figure}
Figure \ref{fig:ckls-trajectories} presents simulated trajectories of a controlled CKLS stochastic differential equation under four fixed pairs of advertising strategies \((u^1, u^2)\). Each panel displays 20 sample paths of the state variable \(X(s)\), normalized to the interval \([0, 1]\), over the time horizon \(s \in [0,1]\). The red dots highlight initial and terminal values, revealing how different control strategies influence the evolution and dispersion of market dynamics.

\begin{figure}[H]
    \centering
    \includegraphics[width=0.95\textwidth]{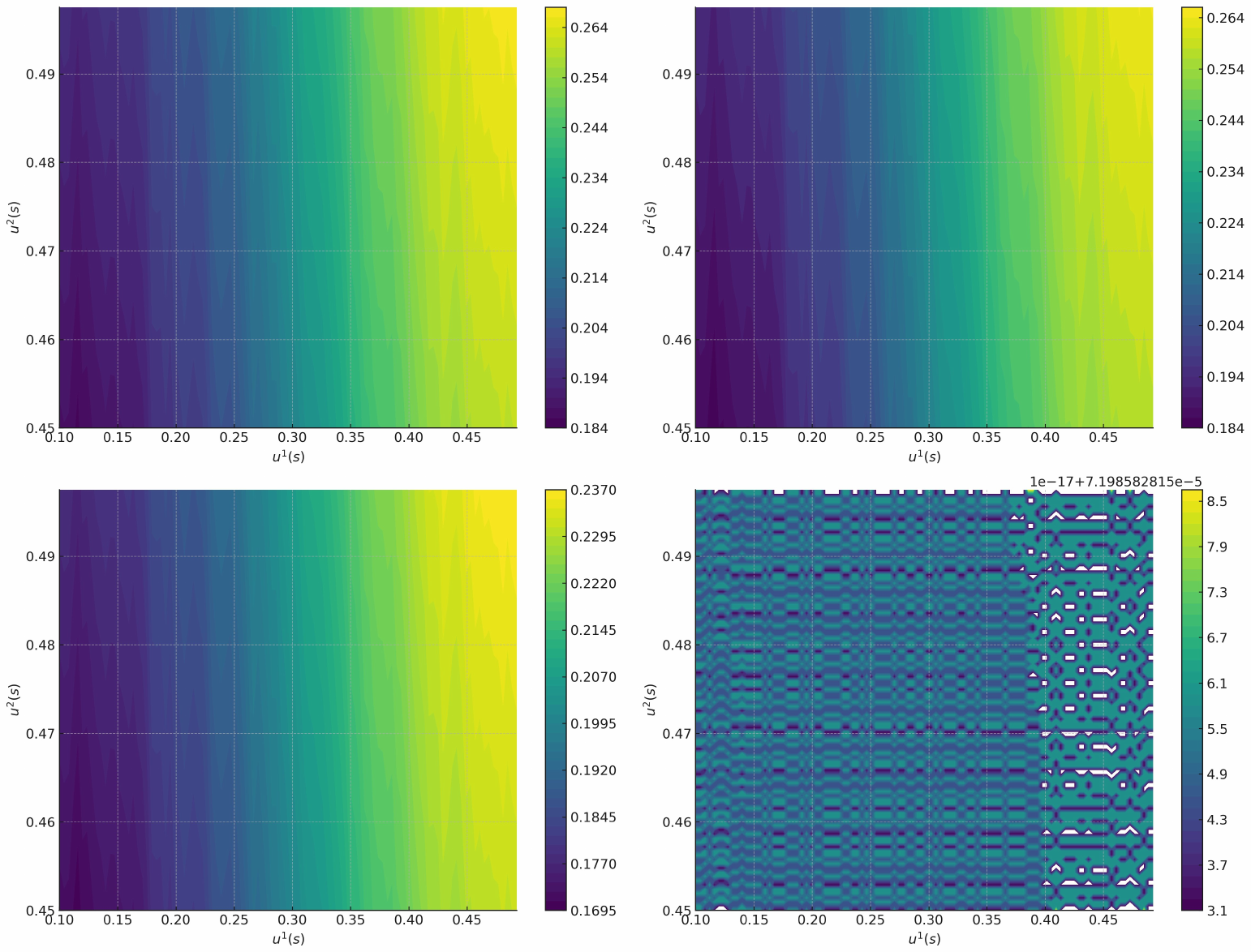}
    \caption{Contour plots illustrating how drift, diffusion, log-amplitude, and gradient fields vary with $u^i$.}
    \label{fig:ckls-contours}
\end{figure}
Figure \ref{fig:ckls-contours} displays four contour plots illustrating various transformations of the CKLS stochastic differential equation under different combinations of firm control variables \(u^1(s)\) and \(u^2(s)\). The top-left and top-right panels depict the drift and raw diffusion terms, respectively, while the bottom-left shows a logarithmic transformation of the diffusion, and the bottom-right presents its numerical gradient. These visualizations help to characterize how different control strategies shape the stochastic behavior of the market state.

The weak incumbent chooses a disclosure time \( \rho \), at which it voluntarily reveals its type. The entrant, on the other hand, selects a market entry time \( \tilde{\tau} \), incurring an irreversible entry cost \( F \). Entry leads to full revelation of the incumbent’s type and initiates payoffs: the entrant earns \( D_E^w X(s) \), and the weak incumbent earns \( D_I^w X(s) \). Prior to entry, the weak incumbent earns \( Q X(s) \) until disclosure, after which it receives \( M^w X(s) \), while the entrant earns nothing until entry occurs. For simplicity, assume entry is never profitable against a strong incumbent, but may be desirable against a weak one when demand \( X(s) \) is sufficiently high, i.e., \( D_E^w > 0 \). Moreover, signaling is costly for the weak incumbent, as post-disclosure payoffs are lower than pre-disclosure (\( M^w < Q \)), yet signaling remains preferable to immediate entry (\( Q > D_E^w \)). Both firms are risk-neutral and discount future payoffs at a constant rate \( \gamma \). To ensure payoffs are finite, consider
\begin{equation}\label{fin}
  \alpha_1 + \alpha_2 X(s) + \sum_{i=1}^2 \theta_i u^i(s) \leq \gamma.  
\end{equation}

Condition \ref{fin} captures the strategic interaction between signaling strength, the evolution of market demand, and the timing of entry decisions under uncertainty and asymmetric information. The beliefs are updated using the Bayes' rule. Define the log-likelihood ratio as \( Z(s) := \ln\left[\frac{p(s)}{1 - p(s)}\right] \), following \cite{mcclellan2022experimentation}, where p(s) is the entrant's posterior belief that the incumbent is the strong type at time s. 

 The payoff functionals are well defined only if the denominator in the terminal payoff terms is strictly positive. Hence, Condition~\eqref{fin} must be verified for the equilibrium feedback controls rather than imposed only as a formal restriction. Define the controlled drift along the equilibrium path by
	\[
	\mu^\ast(s)
	=
	\alpha_1+\alpha_2X(s)+\sum_{i=1}^{2}\theta_i u^{i,\ast}(s).
	\]
	We require
	\[
	\mu^\ast(s)<\gamma,
	\qquad s\in[0,t].
	\]
	Equivalently, there must exist a constant \(\varepsilon>0\) such that
	\[
	\gamma-\mu^\ast(s)\geq \varepsilon,
	\qquad s\in[0,t].
	\]
	
	\begin{prop}[Finiteness of the payoff denominators]
		\label{prop:finiteness_equilibrium}
		Suppose \(X(s)\in[\underline x,\bar x]\subset(0,\infty)\) along the calibrated state path and suppose the equilibrium controls satisfy
		\(\
		u^{i,\ast}(s)\in[\underline u_i,\bar u_i],\) for all \( i=1,2.
		\)
		If
		\[
		\alpha_1+\alpha_2 x+\theta_1u^1+\theta_2u^2
		\leq
		\gamma-\varepsilon
		\]
		for all
		\[
		\left(x,u^1,u^2\right)\in
		[\underline x,\bar x]\times[\underline u_1,\bar u_1]\times[\underline u_2,\bar u_2],
		\]
		for some \(\varepsilon>0\), then the payoff expressions \(J^1\) and \(J^2\) are finite along the equilibrium path.
		
	\end{prop}
	
	\begin{proof}
		Let
		\[
		D^\ast(s)
		=
		\gamma-
		\left[
		\alpha_1+\alpha_2X(s)+\theta_1u^{1,\ast}(s)+\theta_2u^{2,\ast}(s)
		\right].
		\]
		By the stated assumption,
		\[
		D^\ast(s)\geq \varepsilon>0
		\]
		for every \(s\in[0,t]\). Hence the terminal payoff denominators appearing in \(J^1\) and \(J^2\) are uniformly bounded away from zero:
		\[
		\frac{1}{D^\ast(s)}\leq \frac{1}{\varepsilon}.
		\]
		Therefore,
		\[
		\left|
		\frac{D_I^WX(s)}
		{\gamma-\left[\alpha_1+\alpha_2X(s)+\sum_{i=1}^{2}\theta_i u^{i,\ast}(s)\right]}
		\right|
		\leq
		\frac{|D_I^W|}{\varepsilon}|X(s)|
		\]
		and
		\[
		\left|
		\frac{\left[D_E^W-(u^{2,\ast}(t))^2\right]X(s)}
		{\gamma-\left[\alpha_1+\alpha_2X(s)+\sum_{i=1}^{2}\theta_i u^{i,\ast}(s)\right]}
		\right|
		\leq
		\frac{|D_E^W|+\bar u_2^2}{\varepsilon}|X(s)|.
		\]
		Since the controlled CKLS process has finite first moment on the finite horizon \([0,t]\), the right-hand sides are integrable. The running payoff terms are also finite because \(u^{i,\ast}\) is bounded and \(X(s)\) has finite first moment. Consequently, \(J^1\) and \(J^2\) are finite. This proves the proposition.
	\end{proof}

 \begin{prop}[Bayesian updating and log-likelihood ratio]
\label{prop:bayesian_update}
Let the entrant observe the public demand process \(X=\{X(s):s\geq 0\}\) generated by
\begin{equation*}
    dX(s)=\mu_j(s,X(s),u(s))\,ds+\alpha_3 X(s)^{\alpha_4}\,dB(s),
    \qquad j\in\{S,W\},
\end{equation*}
where \(j=S\) denotes a strong incumbent and \(j=W\) denotes a weak incumbent. Consider $\sigma(X(s))=\alpha_3 X(s)^{\alpha_4}$,
and suppose that the entrant observes the filtration $\mathcal F_s^X:=\sigma\{X(r):0\leq r\leq s\}.$
Let \(p(s):=\mathbb P(j=S\mid \mathcal F_s^X)\) be the entrant's posterior belief that the incumbent is strong, and define the log-likelihood ratio
\[
    Z(s):=\log\left[\frac{p(s)}{1-p(s)}\right].
\]
Assume that the two types differ only through the drift of the observed demand process and that both induce probability measures that are mutually absolutely continuous on \(\mathcal F_s^X\). Then
\begin{equation*}
    Z(s)
    =
    Z_0
    +
    \int_0^s
    \frac{\mu_S(r,X(r),u(r))-\mu_W(r,X(r),u(r))}
    {\alpha_3^2 X(r)^{2\alpha_4}}
    \,dX(r)
    -
    \frac{1}{2}
    \int_0^s
    \frac{\mu_S^2(r,X(r),u(r))-\mu_W^2(r,X(r),u(r))}
    {\alpha_3^2 X(r)^{2\alpha_4}}
    \,dr,
\end{equation*}
where
\[
    Z_0=\log\left[\frac{p(0)}{1-p(0)}\right].
\]
Consequently, Bayesian updating is completely characterized by the likelihood ratio generated by the observed demand path.
\end{prop}

\begin{proof}
The proof proceeds by deriving the likelihood ratio of the observed demand path under the two possible incumbent types. Let \(\mathbb P_S\) and \(\mathbb P_W\) denote the probability laws of the process \(X\) on \(\mathcal F_s^X\) when the incumbent is strong and weak, respectively. Conditional on type \(j\in\{S,W\}\), the entrant observes
\[
    dX(s)=\mu_j(s,X(s),u(s))\,ds+\sigma(X(s))\,dB_j(s),
\]
where $\sigma(X(s))=\alpha_3X(s)^{\alpha_4}.$
Since, the diffusion coefficient is common across the two types, the only statistical distinction between the two models is the drift. This is precisely the environment in which the likelihood ratio between \(\mathbb P_S\) and \(\mathbb P_W\) is obtained by Girsanov's theorem.

Fix \(s>0\). Under the weak-type measure \(\mathbb P_W\), the observed process satisfies
\[
    dX(r)=\mu_W(r,X(r),u(r))\,dr+\sigma(X(r))\,dB_W(r),
    \qquad 0\leq r\leq s.
\]
Define
\[
    \lambda(r)
    :=
    \frac{\mu_S(r,X(r),u(r))-\mu_W(r,X(r),u(r))}
    {\sigma(X(r))}.
\]
Under the usual Novikov condition,
\[
    \mathbb E_W\left[
    \exp\left\{
    \frac{1}{2}\int_0^s \lambda^2(r)\,dr
    \right\}
    \right]<\infty,
\]
the stochastic exponential
\[
    \Lambda_s
    :=
    \exp\left\{
    \int_0^s \lambda(r)\,dB_W(r)
    -
    \frac{1}{2}\int_0^s \lambda^2(r)\,dr
    \right\}
\]
is a martingale. Girsanov's theorem then implies that \(\Lambda_s\) is the Radon-Nikodym derivative of the strong-type law with respect to the weak-type law on \(\mathcal F_s^X\), namely
\[
    \Lambda_s
    =
    \frac{d\mathbb P_S}{d\mathbb P_W}\bigg|_{\mathcal F_s^X}.
\]
Therefore, by Bayes' rule,
\[
    p(s)
    =
    \mathbb P(j=S\mid \mathcal F_s^X)
    =
    \frac{p(0)\Lambda_s}{p(0)\Lambda_s+1-p(0)}.
\]
It follows that
\[
    \frac{p(s)}{1-p(s)}
    =
    \frac{p(0)}{1-p(0)}\Lambda_s.
\]
Taking logarithms yields
\[
    Z(s)
    =
    Z_0+\log \Lambda_s,
\]
where
\[
    Z_0=\log\left[\frac{p(0)}{1-p(0)}\right].
\]

It remains to express \(\log\Lambda_s\) in terms of the observed process \(X\). Since under \(\mathbb P_W\),
\[
    dB_W(r)
    =
    \frac{dX(r)-\mu_W(r,X(r),u(r))\,dr}
    {\sigma(X(r))},
\]
we obtain
\begin{align}
    \log \Lambda_s
    &=
    \int_0^s
    \frac{\mu_S(r,X(r),u(r))-\mu_W(r,X(r),u(r))}
    {\sigma(X(r))}
    dW_W(r)
    -
    \frac{1}{2}
    \int_0^s
    \left[
    \frac{\mu_S(r,X(r),u(r))-\mu_W(r,X(r),u(r))}
    {\sigma(X(r))}
    \right]^2
    dr \nonumber \\
    &=
    \int_0^s
    \frac{\mu_S(r,X(r),u(r))-\mu_W(r,X(r),u(r))}
    {\sigma^2(X(r))}
    \,dX(r) \nonumber \\
    &\quad
    -
    \int_0^s
    \frac{\left[\mu_S(r,X(r),u(r))-\mu_W(r,X(r),u(r))\right]\mu_W(r,X(r),u(r))}
    {\sigma^2(X(r))}
    \,dr \nonumber \\
    &\quad
    -
    \frac{1}{2}
    \int_0^s
    \frac{\left[\mu_S(r,X(r),u(r))-\mu_W(r,X(r),u(r))\right]^2}
    {\sigma^2(X(r))}
    \,dr.
\end{align}
Combining the two deterministic integral terms yields
\[
    \left[\mu_S-\mu_W\right]\mu_W
    +
    \frac{1}{2}\left[\mu_S-\mu_W\right]^2
    =
    \frac{1}{2}\left(\mu_S^2-\mu_W^2\right).
\]
Hence,
\[
    \log \Lambda_s
    =
    \int_0^s
    \frac{\mu_S(r,X(r),u(r))-\mu_W(r,X(r),u(r))}
    {\sigma^2(X(r))}
    \,dX(r)
    -
    \frac{1}{2}
    \int_0^s
    \frac{\mu_S^2(r,X(r),u(r))-\mu_W^2(r,X(r),u(r))}
    {\sigma^2(X(r))}
    \,dr.
\]
Substituting \(\sigma^2(X(r))=\alpha_3^2 [X(r)]^{2\alpha_4}\) gives
\[
    Z(s)
    =
    Z_0
    +
    \int_0^s
    \frac{\mu_S(r,X(r),u(r))-\mu_W(r,X(r),u(r))}
    {\alpha_3^2 [X(r)]^{2\alpha_4}}
    \,dX(r)
    -
    \frac{1}{2}
    \int_0^s
    \frac{\mu_S^2(r,X(r),u(r))-\mu_W^2(r,X(r),u(r))}
    {\alpha_3^2 X(r)^{2\alpha_4}}
    \,dr.
\]
This completes the proof.
\end{proof}

\begin{cor}[Reduced-form belief updating used in the paper]
\label{cor:reduced_form_update}
Suppose that, over a short observation interval, the entrant approximates the drift advantage of the strong type over the weak type by
\[
    \mu_S(s,X(s),u(s))-\mu_W(s,X(s),u(s))
    =
    2\left[
    \alpha_1+\alpha_2 X(s)+\sum_{i=1}^2 \theta_i u^i(s)
    \right],
\]
and treats the local volatility as $\sigma^2(X(s))=\alpha_3^2 [X(s)]^{2\alpha_4}.$
Then the local log-likelihood ratio can be written as
\begin{equation}
\label{eq:reduced_log_likelihood}
    Z(s)
    =
    Z_0
    +
    \frac{
    2\left[
    \alpha_1+\alpha_2 X(s)+\sum_{i=1}^2 \theta_i u^i(s)
    \right]
    }
    {\alpha_3^2 [X(s)]^{2\alpha_4}}
    \left[X(s)-x_0\right],
\end{equation}
where \(X(0)=x_0\).
\end{cor}

\begin{proof}
The result follows from Proposition~\ref{prop:bayesian_update} under a local Gaussian approximation to the demand process. Over a sufficiently short interval, the entrant treats the drift difference and the diffusion coefficient as locally constant at their current values. Hence,
\[
    dX(r)\approx X(s)-x_0,
\]
and
\[
    \mu_S(r,X(r),u(r))-\mu_W(r,X(r),u(r))
    \approx
    2\left[
    \alpha_1+\alpha_2 X(s)+\sum_{i=1}^2 \theta_i u^i(s)
    \right].
\]
Substituting these expressions into the stochastic-integral component of the likelihood ratio in Proposition~\ref{prop:bayesian_update} gives
\[
    Z(s)
    =
    Z_0
    +
    \frac{
    2\left[
    \alpha_1+\alpha_2 X(s)+\sum_{i=1}^2 \theta_i u^i(s)
    \right]
    }
    {\alpha_3^2 [X(s)]^{2\alpha_4}}
    \left[X(s)-x_0\right],
\]
which is the reduced-form updating equation used in the model.
\end{proof}

Therefore, the belief dynamics are given by
\[
Z(s) = Z_0 + \frac{2}{\alpha_3^2 [X(s)]^{2\alpha_4}} \left[\alpha_1 + \alpha_2 X(s) + \sum_{i=1}^2 \theta_i u^i(s)\right] [X(s) - x_0],
\]
where \( Z_0 := \ln\left[\frac{p(0)}{1 - p(0)}\right] \) is the initial belief.
The expected payoff of the incumbent is 
\begin{multline}\label{2}
J^1\left(\rho,t,x_0,z_0,u^1\right)=\E_{0,z_0}\biggr\{\int_0^{\rho\wedge t}e^{-\gamma s}[Q-u^1(s)]X(s)ds+\mathbbm{1}(\rho\leq t)\int_\rho^te^{-\gamma s}[M^W-\bar u^3]X(s)ds\\+e^{-\gamma t}\frac{D_I^WX(t)}{\gamma-\left[\a_1+\a_2 X(t)+\sum_{i=1}^2\theta_i u^i(t)\right]}\biggr\},
\end{multline}

\begin{figure}[htbp]
    \centering
    \begin{subfigure}[t]{0.48\textwidth}
        \centering
        \includegraphics[width=\linewidth]{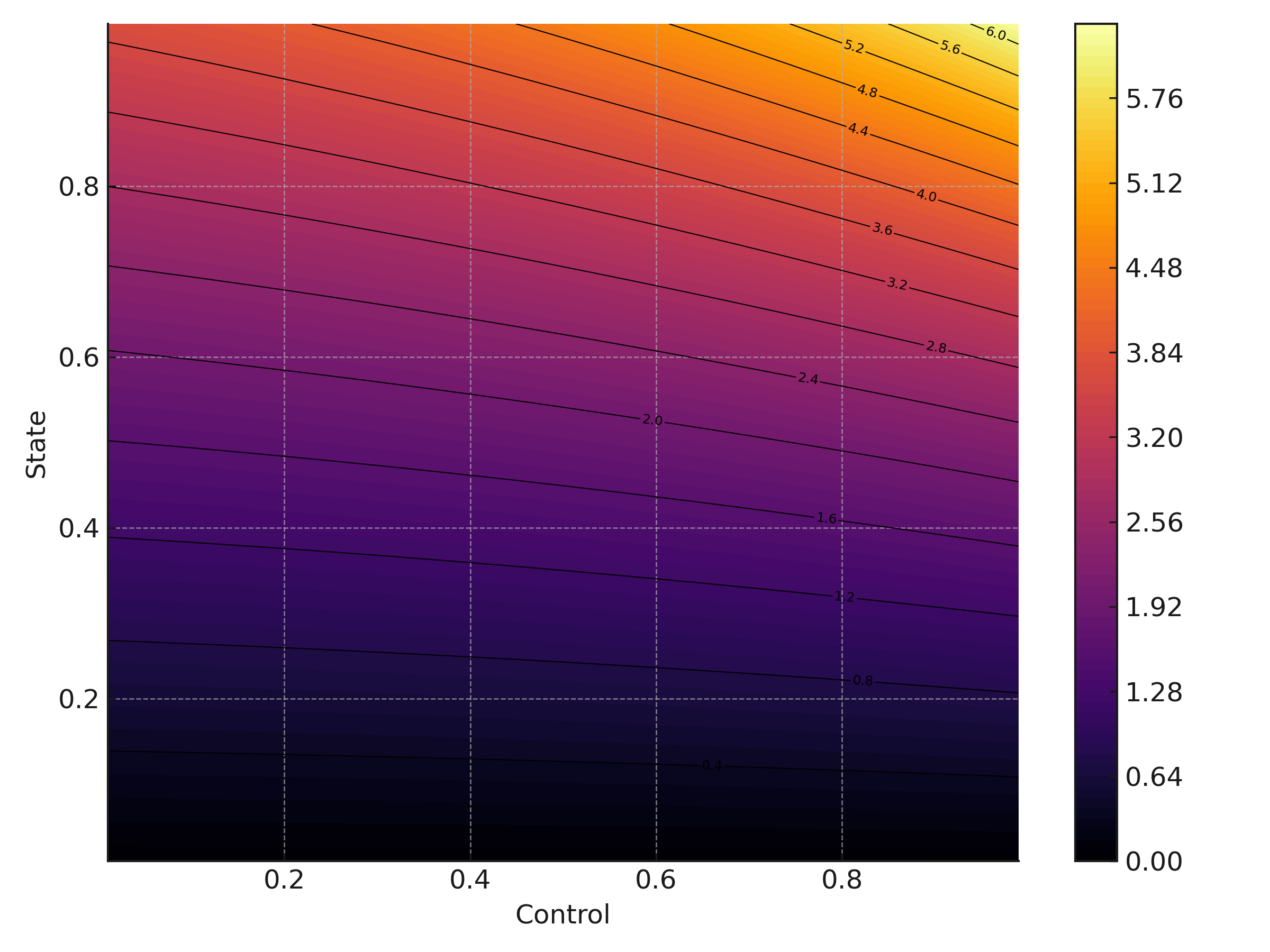} 
        \caption{Contour plot of $J^1$.}
    \end{subfigure}
    \hfill
    \begin{subfigure}[t]{0.48\textwidth}
        \centering
        \includegraphics[width=\linewidth]{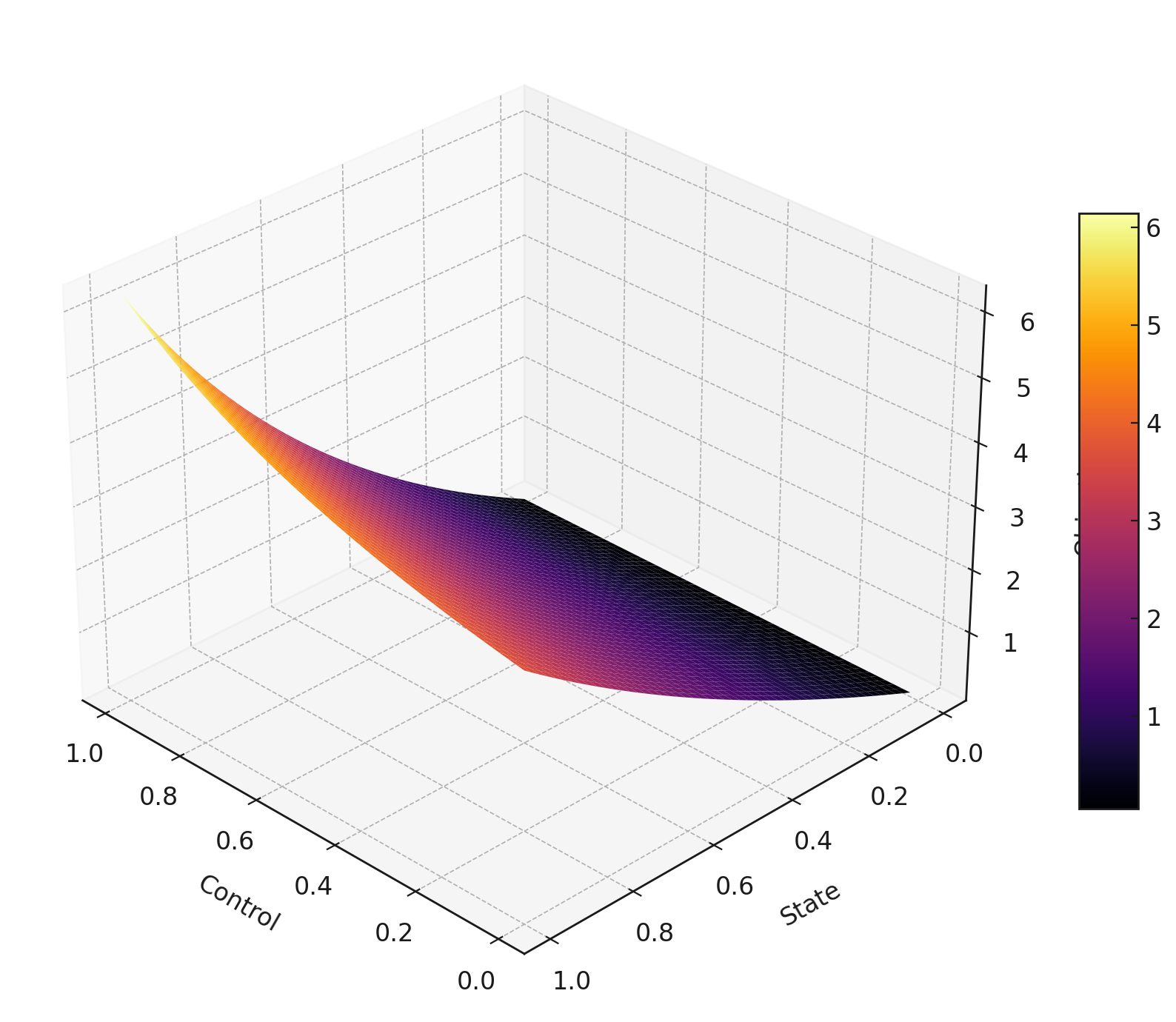} 
        \caption{3D surface of $J^1$.}
    \end{subfigure}
    \caption{Visualizations of the expected payoff function for Firm 1 under varying control $u^1$ and state $X(s)$.}
    \label{fig:leader}
\end{figure}
The two panels of Figure \ref{fig:leader} jointly illustrate the expected payoff function \( J^1\left(\rho,t,x_0,z_0,u^1\right) \) of Firm 1. The left panel presents a level-set contour plot that captures how combinations of control intensity and market state generate iso-payoff curves, thereby identifying regions of strategic complementarity and diminishing returns. The right panel offers a three-dimensional perspective on the same function, showing the curvature and slope of the payoff landscape across the domain. Both visualizations confirm the nonlinearity of firm incentives, with the presence of interior optima and clear interactions between the control and state variables.

Given that the incumbent firm attains its optimal advertising expenditure, the expected payoff of the entrant is 
\begin{equation}\label{3}
J^2\left(t,x_0,z_0,u^1\right)=\E_{0,z_0}\left\{e^{-\gamma t}\left[\mathbbm{1}(\theta=w)\frac{\left[D_E^W-[u^2(t)]^2\right]X(t)}{\gamma-\left[\a_1+\a_2 X(t)+\sum_{i=1}^2\theta_i u^i(t)\right]}-F\right]\right\}.
\end{equation}

\begin{figure}[H]
    \centering
    \begin{subfigure}[t]{0.48\textwidth}
        \centering
        \includegraphics[width=\linewidth]{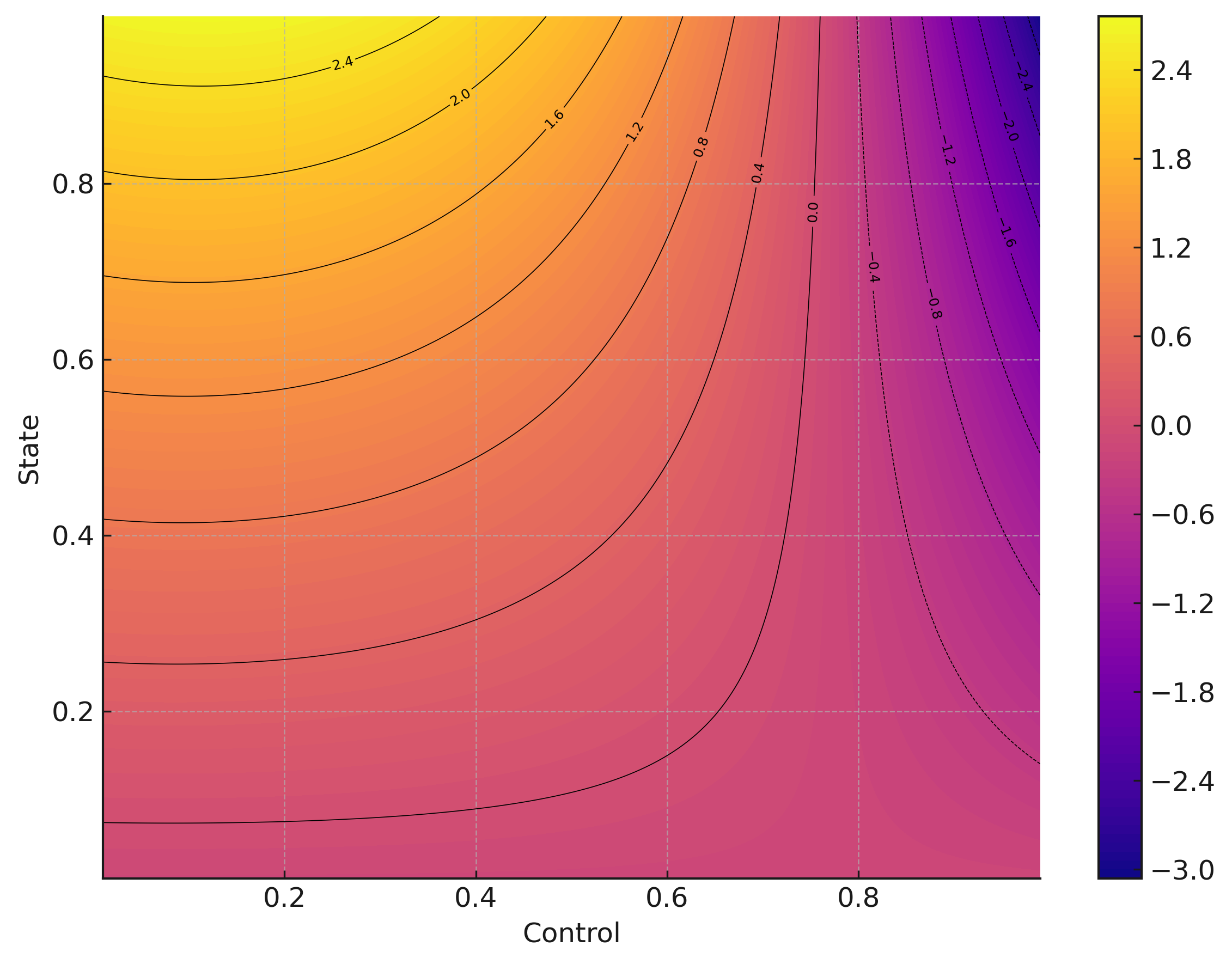} 
        \caption{Contour plot of $J^2$.}
    \end{subfigure}
    \hfill
    \begin{subfigure}[t]{0.48\textwidth}
        \centering
        \includegraphics[width=\linewidth]{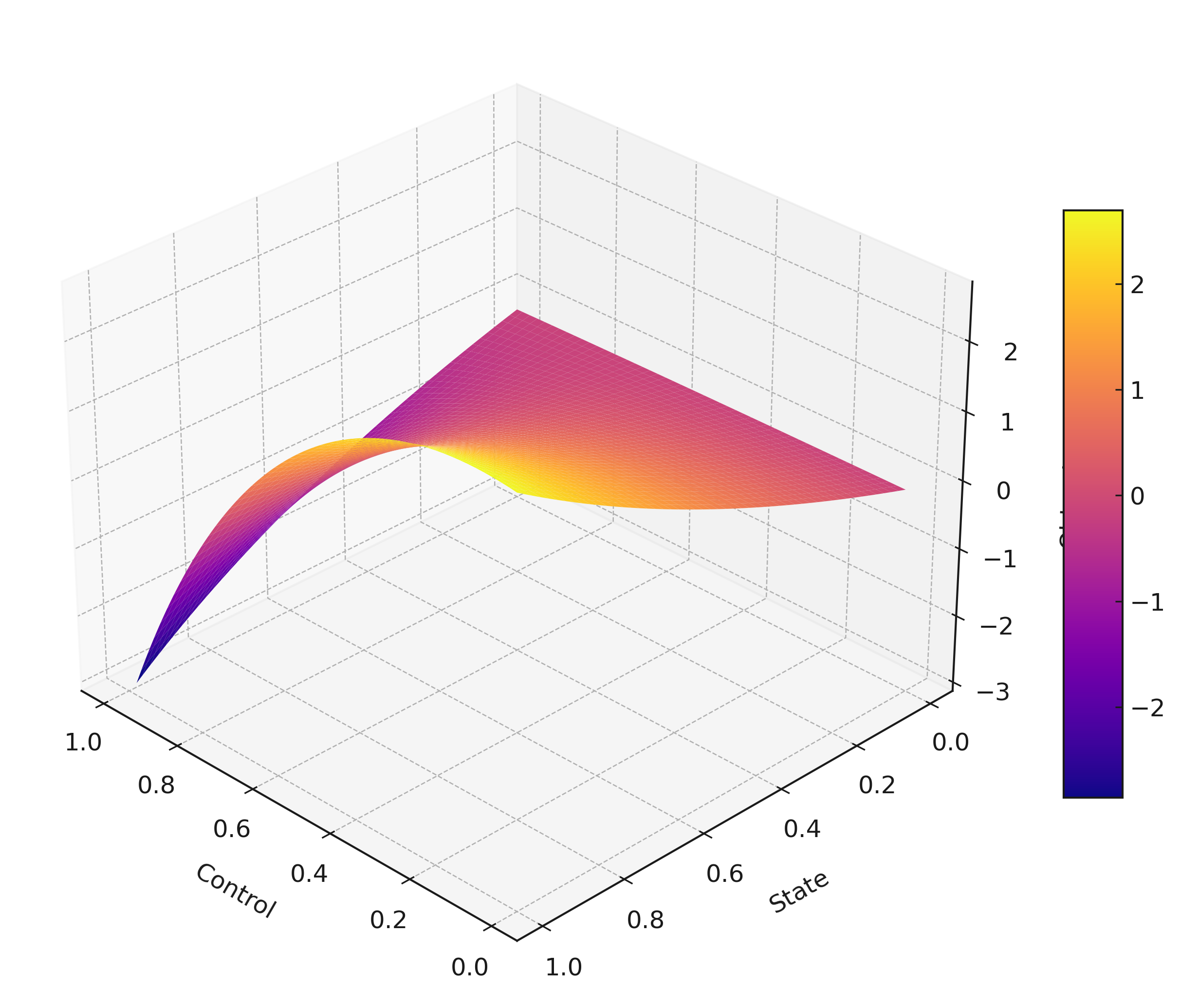} 
        \caption{3D surface plot of $J^2$.}
    \end{subfigure}
    \caption{Visualizations showing the entrant's projected payoff across various scenarios.}
    \label{fig:entrant}
\end{figure}

The two panels of Figure \ref{fig:entrant} illustrate the structure of the expected payoff function \( J^2\left(t,x_0,z_0,u^1\right) \) for the entrant firm in a dynamic stochastic environment, where its behavior is shaped by a continuous-time control problem under uncertainty. The entrant’s payoff depends on both the market share and the firm’s chosen level of expenditure in commercials, reflecting the tension between the marginal benefit of market penetration and the quadratic cost of control. The left panel presents a contour plot depicting iso-payoff curves, highlighting the nonlinear interaction between state and control variables. Regions of high payoff appear as topologically enclosed areas, suggesting interior optima where further increases in control effort yield diminishing or even negative returns due to rising marginal costs. The right panel shows a three-dimensional surface plot of the same function, providing a geometric visualization of payoff curvature over the joint domain. This surface reveals how expected profitability is concave in key regions and steeply declines in others, particularly when aggressive control strategies are used in less favorable market conditions.

\section {Probabilistic Construction.} \label{prob}
Since $B(s)$ in Equation \eqref{0} is a Brownian motion in $\mathbb R^2$, then its normalization a can be written as $\E\{B(s)B(\nu)\}=\min(s,\nu)$, such that $B(0)=0$. In Equation \eqref{0}
\[
\alpha_1 + \alpha_2 X(s) + \sum_{i=1}^2 \theta_i u^i(s)
\]
is a $L^0$-Lipschitz drift. The above process has a generator
\[
\mathfrak G=\frac{1}{2}\frac{d^2}{dx^2}+\left[\alpha_1 + \alpha_2 x + \sum_{i=1}^2 \theta_i u^i\right],
\]
with the density 
\[
\Psi_0(x)=\mathcal C\cdot\exp\left\{2\int_0^x \left[\alpha_1 + \alpha_2 y + \sum_{i=1}^2 \theta_i u^i\right]dy\right\},
\]
where $\mathcal C$ is a positive finite constant. Define the quantum Hamiltonian as
\[
\mathcal H=-\frac{1}{2}\frac{d^2}{dx^2}+\mathfrak T(x),\ \text{where}\ \mathfrak T(x)=\frac{1}{2}\left\{\left[\alpha_1 + \alpha_2 x + \sum_{i=1}^2 \theta_i u^i\right]+\left[\alpha_1 + \alpha_2 + \sum_{i=1}^2 \theta_i u^i\right]\right\}.
\]
Following \cite{truman1992elementary} we use 
\[
\exp\left\{-\int_0^x \left[\alpha_1 + \alpha_2 y + \sum_{i=1}^2 \theta_i u^i\right]dy\right\}\left(\frac{d}{dx}\right)\exp\left\{\int_0^x \left[\alpha_1 + \alpha_2 y + \sum_{i=1}^2 \theta_i u^i\right]dy\right\}=\frac{d}{dx}+\alpha_1 + \alpha_2 x + \sum_{i=1}^2 \theta_i u^i,
\]
to yield the formal Schr\"odinger operator
\[
\mathfrak G:=-\Psi_0^{-1/2}\mathcal H \Psi_0^{1/2},
\]
for $\mathcal H$.   
\begin{as}\label{a0}
(a). The invariant density $\Psi_0\in L^1(\mathbb R^2)$ and it is bounded by first few derivatives.\\
(b). The function $\mathcal T$ is $C^1$ and bounded from below, with its first order derivative is bounded and continuous.\\
(c). The first order derivative of $\mathcal T$ is bounded at the finite number of points where the function $\mathcal T$ indeed have discontinuity.
\end{as}

\begin{prop}
  [Invariant density and Feynman-Kac representation]
\label{prop:fk_claim_proof}
Under Assumption~\ref{a0}, the ground-state transformation associated with
\(\Psi_0\) is well defined on \(L^2(\Psi_0 dx)\). Moreover, the operator
\(\mathcal H\) is bounded from below and generates a strongly continuous
semigroup. For every bounded measurable test function \(f\), the semigroup
admits the Feynman-Kac representation
\[
\left(e^{-t\mathcal H}f\right)(x)
=
\mathbb E_x
\left[
\exp\left\{-\int_0^t \mathfrak T(B_r)\,dr\right\}
f(B_t)
\right],
\]
where \(B\) denotes Brownian motion started from \(x\). Consequently, the
formal path-integral representation used in the control problem is justified.  
\end{prop}

\begin{proof}
Let
\[
A f(x)
=
\frac{1}{2}f''(x)+\mu(x)f'(x),
\]
where \(\mu\) denotes the drift of the controlled diffusion. The density
\(\Psi_0\) is assumed to be integrable and sufficiently regular. Hence,
$\pi(dx):=\Psi_0(x)\,dx$
defines a finite measure after normalization. Since \(\Psi_0>0\), the weighted
space \(L^2(\pi)\) is well defined.

We first verify that the ground-state transformation is meaningful. Define $Uf:=\Psi_0^{1/2}f.$ Since \(\Psi_0\in L^1\) and is positive, \(U\) maps \(L^2(\pi)\) isometrically
into \(L^2(dx)\). Indeed,
\[
\|Uf\|_{L^2(dx)}^2
=
\int_{\mathbb R}|f(x)|^2\Psi_0(x)\,dx
=
\|f\|_{L^2(\pi)}^2.
\]
Thus \(U\) is a unitary ground-state transform from \(L^2(\pi)\) onto its image
in \(L^2(dx)\). The boundedness assumptions on the first few derivatives of
\(\Psi_0\) ensure that this transformation preserves the local Sobolev domain
needed to interpret the second-order differential operator distributionally and
then by closure.

Next consider the Schr\"odinger-type operator
\[
\mathcal H
=
-\frac{1}{2}\frac{d^2}{dx^2}+\mathfrak T(x).
\]
By Assumption~\ref{a0}(b), the potential \(\mathfrak T\) is bounded from below.
Hence there exists \(m\in\mathbb R\) such that
\[
\mathfrak T(x)\geq m
\qquad\text{for all }x.
\]
For \(f\in C_c^\infty(\mathbb R)\),
\[
\langle f,\mathcal H f\rangle_{L^2(dx)}
=
\frac{1}{2}\int_{\mathbb R}|f'(x)|^2\,dx
+
\int_{\mathbb R}\mathfrak T(x)|f(x)|^2\,dx.
\]
Therefore
\[
\langle f,\mathcal H f\rangle_{L^2(dx)}
\geq
m\|f\|_{L^2(dx)}^2.
\]
Thus \(\mathcal H\) is semibounded on \(C_c^\infty(\mathbb R)\). The associated
quadratic form
\[
\mathcal E_{\mathcal H}(f,f)
=
\frac{1}{2}\int_{\mathbb R}|f'(x)|^2\,dx
+
\int_{\mathbb R}\mathfrak T(x)|f(x)|^2\,dx
\]
is closable because the negative part of \(\mathfrak T\) is bounded. Its closure
defines a lower semibounded self-adjoint operator, again denoted by
\(\mathcal H\), through the Friedrichs extension. Hence \(-\mathcal H\)
generates a strongly continuous semigroup $\{e^{-t\mathcal H}:t\geq0\}
$ on \(L^2(dx)\).

It remains to justify the Feynman-Kac representation. Since
\(\mathfrak T\) is bounded from below, write $\mathfrak T=\mathfrak T_+-\mathfrak T_-,$ where \(\mathfrak T_-\) is bounded. Hence, for every finite \(t\), $\exp\left\{-\int_0^t \mathfrak T(B_r)\,dr\right\}$ is integrable. The regularity assumptions on \(\mathfrak T\), together with the
boundedness condition on its derivative except possibly at finitely many points,
ensure that the potential is locally admissible for the Schr\"odinger semigroup.
The exceptional points do not affect Brownian occupation integrals, since
Brownian motion spends zero Lebesgue time at any fixed point almost surely.
Consequently, $\int_0^t \mathfrak T(B_r)\,dr $ is well defined almost surely.
For \(f\in C_b(\mathbb R)\), define
\[
v(t,x)
=
\mathbb E_x
\left[
\exp\left\{-\int_0^t \mathfrak T(B_r)\,dr\right\}
f(B_t)
\right].
\]
Standard localization and It\^o arguments imply that \(v\) solves the parabolic
problem
\[
\frac{\partial v}{\partial t}
=
\frac{1}{2}\frac{\partial^2 v}{\partial x^2}
-
\mathfrak T(x)v,
\qquad
v(0,x)=f(x),
\]
in the weak sense. Since the closed quadratic form above generates a unique
semibounded self-adjoint realization of \(\mathcal H\), this weak solution is
identified with $v(t,\cdot)=e^{-t\mathcal H}f.$ Therefore,
\[
(e^{-t\mathcal H}f)(x)
=
\mathbb E_x
\left[
\exp\left\{-\int_0^t \mathfrak T(B_r)\,dr\right\}
f(B_t)
\right].
\]

Finally, applying the unitary transform \(U\) gives
\[
\mathfrak G
=
-U^{-1}\mathcal H U,
\]
so the diffusion semigroup associated with \(\mathfrak G\) is obtained by
conjugating the Schr\"odinger semigroup. Hence
\[
e^{t\mathfrak G}
=
U^{-1}e^{-t\mathcal H}U.
\]
This establishes that the invariant-density construction, the Hamiltonian
representation, and the Feynman--Kac quantization are not merely formal. They
follow from Assumption~\ref{a0} through the ground-state transformation and the
Friedrichs realization of the Schr\"odinger operator. The claim is proved.
\end{proof}

\begin{cl}
 (a) of Assumption \ref{a0}  considers about unboundedness $(-\infty,\infty)$, and (b) and (c) make sure that the semigroup $\exp(t\mathcal H)$ can be represented by a Feynman-Kac quantization.
\end{cl}

Now we are going to consider the importance of the quantum Hamiltonian $\mathcal H$ in this construction. Let the map $s\mapsto X(s)$ be almost surely (a.s.) continuous such that $\left\{s>0;\ X(s)\gtrless\be\right\}$ is a.s in $\mathbb R^2$, which allows it to be decomposed into a number of countable intervals called  excursion intervals \citep{truman1992elementary}.

\begin{figure}[htbp]
    \centering
    \begin{subfigure}[t]{0.48\textwidth}
    \includegraphics[height=5.5cm, keepaspectratio]{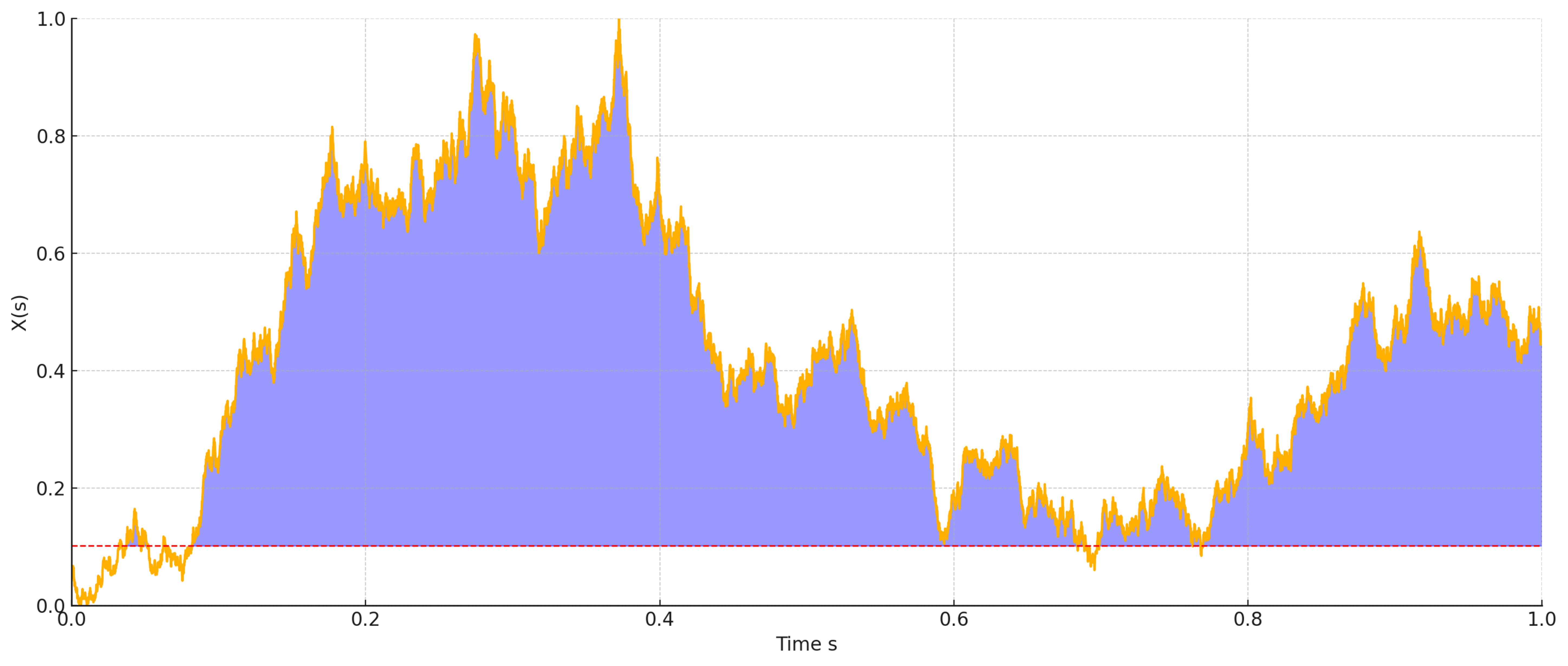}
        \caption{One dimensional CKLS excursion intervals.}
    \end{subfigure}
    \hfill
    \begin{subfigure}[t]{0.48\textwidth}
        \includegraphics[height=5.5cm, keepaspectratio]{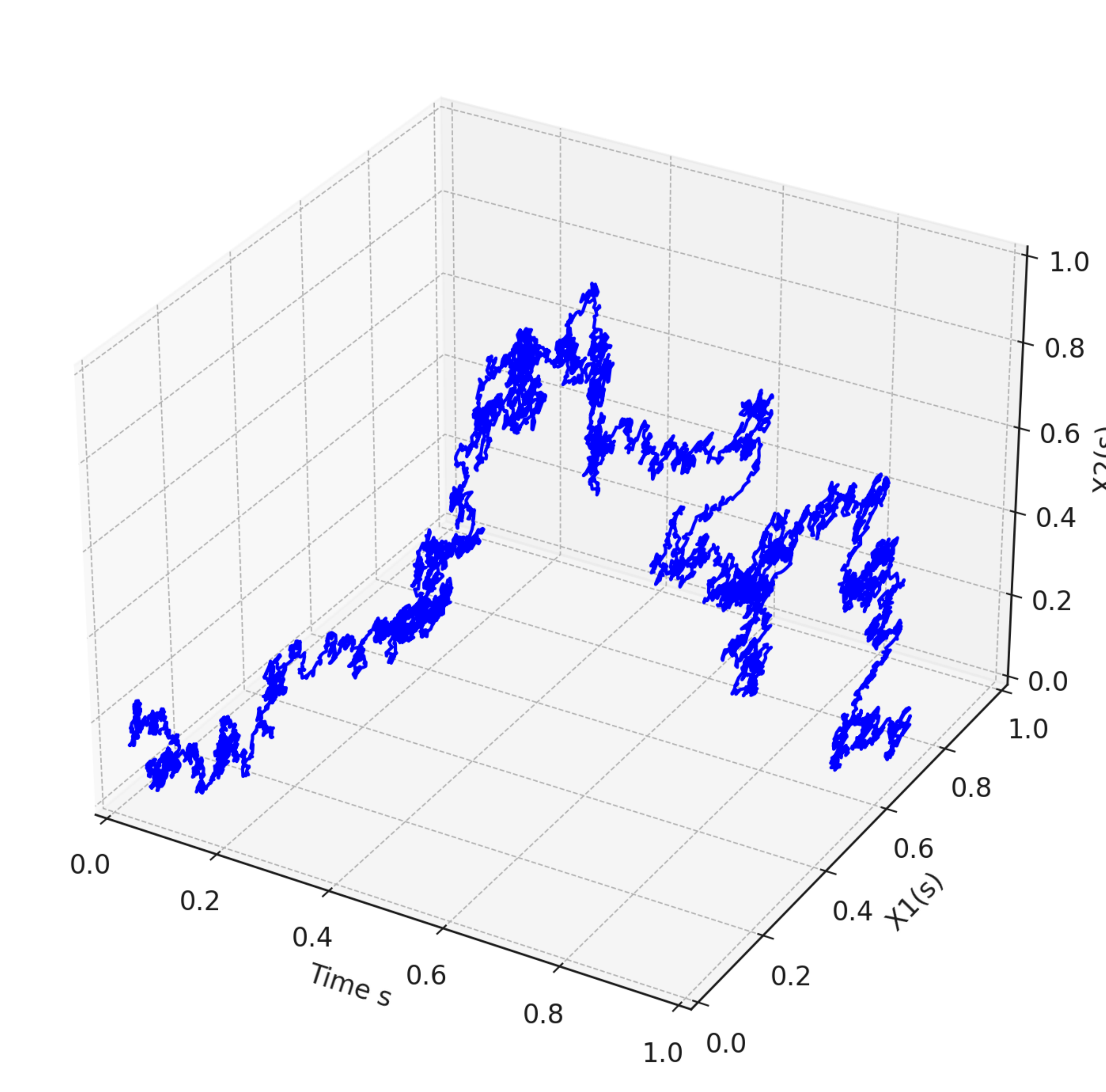}
        \caption{3D extension.}
    \end{subfigure}
    \caption{Comparison of CKLS SDE in one dimensional and $\mathbb R^2$ extension on 3D visualization.}
    \label{fig:ckls}
\end{figure}
The two images in Figure \ref{fig:ckls} illustrate the behavior of the CKLS SDE. The left panel shows the one-dimensional trajectory \( X(s) \) over time, with excursion intervals shaded in blue where the process exceeds a threshold \( \beta \), highlighting its dynamic fluctuations. The right figure extends this to a two-dimensional setting in \( \mathbb{R}^2 \), displaying a 3D trajectory of two coupled CKLS processes \( X_1(s) \) and \( X_2(s) \), offering a visual representation of their joint stochastic evolution over time.

Define
\[
\mathfrak L_\pm(t):=\text{Leb}\{s\in[0,t]:X(s)\gtrless\be\}
\]
and set the local time at $\be$ such that 
\begin{equation}
 \mathfrak L^\be(t):=\lim_{j\downarrow 0} \frac{1}{j} \text{Leb}\left\{ s\in[0,t]:X(s)\in\left[\be-\frac{j}{2},\be+\frac{j}{2}\right]\right\}, 
\end{equation}
where the inverse of $\mathfrak L^\be(t)$ is $\zeta^\be(t):=\inf\left\{s>0;\mathfrak L^\be(s)>t\right\}$; $\zeta^\be(t)=\tau$ being a stopping time with $X(\zeta^\be(t))=\be$.

\begin{lem}\citep{truman1992elementary}
 For $\alpha_1 + \alpha_2 x + \sum_{i=1}^2 \theta_i u^i\equiv 0$ for each $\kappa$, $t>0$ 
 \[
 \E_\be\bigg\{\exp\{-\kappa\zeta^\be(t)\}\bigg\}=\exp\left\{-t\int_0^\infty\left[1-\exp\{-\kappa s\}\right]dv_0(s)\right\},
 \]
 so that
 \[
 v_0[s,\infty)=\sqrt{\frac{2/\pi}{s}},\ \ \text{for each positive s.}
 \]
 Equating the power of $\kappa$ with the number of excursions of duration at least $s$ up to $\mathfrak L^\be=t$,
 \[
 P[e_s(t)=N]=\exp\bigg\{-t v_0[s,\infty)\bigg\}\cdot \frac{[t v_0[s,\infty)]^N}{N!},\ N=0,1,2,\dots,
 \]
 where $e_s(t)$ is the number of excursions.
\end{lem}
\begin{proof}
We consider a reflected Brownian motion \( X(t) \) on \( \mathbb{R}_+ \), starting at zero, governed by the generator \( \frac{1}{2} \frac{d^2}{dx^2} \). In this setting, the drift vanishes, i.e., 
\[
\alpha_1 + \alpha_2 x + \sum_{i=1}^2 \theta_i u^i \equiv 0,
\]
and the process evolves purely through its diffusion component. Let \( \zeta^\beta(t) \) denote the local time accumulated at the origin by time \( t \), and let \( e_s(t) \) be the number of excursions of duration at least \( s \) observed up to local time \( t \), where \( \mathfrak{L}^\beta = t \).

By It\^o excursion theory, the set of excursions away from zero forms a Poisson point process on the space of càdlàg paths with respect to a $\sigma$-finite measure known as the Itô excursion measure, denoted \( n \). The collection of excursion durations then induces a Poisson random measure \( N(ds) \) on \( (0,\infty) \) with intensity measure \( t \cdot v_0(ds) \), where \( v_0[s, \infty) := n(\text{duration} \geq s) \) characterizes the tail of the excursion duration distribution.
Using the Laplace functional of a Poisson random measure, we have for any measurable \( f \colon (0,\infty) \to [0,\infty) \),
\[
\E\left[\exp\left(-\int_0^\infty f(s) N(ds)\right)\right] = \exp\left(-t \int_0^\infty \left(1 - e^{-f(s)}\right) v_0(ds)\right).
\]
Choosing \( f(s) = \kappa s \), the integral becomes \( \kappa \sum s_i \), where \( s_i \) are excursion lengths. Since the total duration of excursions up to local time \( t \) is exactly \( \zeta^\beta(t) \), we get
\[
\E\left[\exp\left(-\kappa \zeta^\beta(t)\right)\right] = \exp\left(-t \int_0^\infty \left(1 - e^{-\kappa s} \right) v_0(ds)\right),
\]
which proves the first assertion.

To evaluate \( v_0[s,\infty) \), recall from the theory of Brownian excursions that the tail of the excursion duration distribution under the Itô measure satisfies
\[
v_0[s,\infty) = n(\text{excursion length} \geq s) = \sqrt{\frac{2/\pi}{s}}.
\]

Finally, since the number of excursions of length at least \( s \) is a Poisson random variable with rate \( t v_0[s,\infty) \), we obtain
\[
\mathbb{P}(e_s(t) = N) = \frac{[t v_0[s,\infty)]^N}{N!} \cdot \exp\left(-t v_0[s,\infty)\right), \quad N = 0, 1, 2, \dots,
\]
which completes the proof.
\end{proof}

\begin{lem}\label{lem:controlled_excursions}
Let \( X(s) \) be the solution to the controlled CKLS SDE
defined in Equation \eqref{0}. Suppose the process starts at \( X(0) = \beta > 0 \). Let \( \zeta^\beta(t) \) denote the local time accumulated at level \( \beta \) by time \( t \), and define \( e_s(t) \) as the number of excursions of duration at least \( s \) from level \( \beta \), observed up to local time \( \zeta^\beta(t) = t \). Then the Laplace transform of the local time satisfies
\[
\E\left[\exp\left(-\kappa \zeta^\beta(t)\right)\right] = \exp\left\{ -t \int_0^\infty \left(1 - e^{-\kappa s} \right) dv^\beta(s) \right\},
\]
where the excursion intensity measure is given by
\[
v^\beta[s, \infty) = \sqrt{\frac{2/\pi}{s}} \cdot \exp\left\{ -2 \int_0^\beta \frac{\alpha_1 + \alpha_2 y + \sum_{i=1}^2 \theta_i u^i}{\alpha_3^2 y^{2\alpha_4}} dy \right\}.
\]
Moreover, the number of such excursions \( e_s(t) \) is a Poisson random variable with rate \( t v^\beta[s, \infty) \), and its distribution is given by
\[
\mathbb{P}[e_s(t) = N] = \exp\left\{-t v^\beta[s,\infty)\right\} \cdot \frac{[t v^\beta[s,\infty)]^N}{N!}, \quad N = 0,1,2,\dots.
\]
\end{lem}

\begin{proof}
Consider the CKLS SDE defined in Equation \eqref{0}, and assume \( X(0) = \beta > 0 \) and the controls \( u^i(s) \) are deterministic and bounded. Define the drift and diffusion components are
$b(x) := \alpha_1 + \alpha_2 x + \sum_{i=1}^2 \theta_i u^i$,
and $ \sigma(x) := \alpha_3 x^{\alpha_4}$, respectively. We are interested in the statistical properties of excursions of the process \( X(s) \) from level \( \beta \). Let \( \zeta^\beta(t) \) denote the local time at level \( \beta \), and define \( e_s(t) \) to be the number of excursions from \( \beta \) of duration at least \( s \), accumulated by local time \( \zeta^\beta(t) = t \). In this case Girsanov's theorem helps to eliminate the drift. Define the Radon-Nikodym derivative as
\[
\frac{d\mathbb{Q}}{d\mathbb{P}} \bigg|_{\mathcal{F}_t} =
\exp\left\{-\int_0^t \frac{b(X(s))}{\sigma(X(s))} dB(s)
- \frac{1}{2} \int_0^t \left(\frac{b(X(s))}{\sigma(X(s))} \right)^2 ds \right\},
\]
which transforms the process \( X(s) \) under \( \mathbb{Q} \) into a diffusion with zero drift and same diffusion coefficient \( \sigma(x) \). Under \( \mathbb{Q} \), the excursions of \( X(s) \) from level \( \beta \) form a Poisson point process in local time, and the Laplace transform of the local time is given by:
\[
\E^\mathbb{Q} \left\{ \exp\left( -\kappa \zeta^\beta(t) \right) \right\}
= \exp\left\{ -t \int_0^\infty \left(1 - e^{-\kappa s} \right) v_0(ds) \right\},
\]
where \( v_0[s,\infty) = \sqrt{ \frac{2/\pi}{s} } \) is the standard Brownian excursion tail measure.In order to revert to \( \mathbb{P} \), we apply the measure change:
\[
\E^\mathbb{P} \left\{ \exp\left( -\kappa \zeta^\beta(t) \right) \right\}
= \E^\mathbb{Q} \left[ \exp\left\{ -\kappa \zeta^\beta(t) \right\}
\cdot \frac{d\mathbb{P}}{d\mathbb{Q}} \right].
\]

Here the intensity of excursions is modified by a tilt determined by the exponential martingale in the Girsanov transformation. This modifies the excursion intensity measure \( v^\beta \) such that
\[
v^\beta(ds) = v_0(ds) \cdot \exp\left\{ -2 \int_0^\beta \frac{b(y)}{\sigma^2(y)} dy \right\}.
\]

Substituting the expressions for \( b(y) \) and \( \sigma(y) \) yields
\[
\frac{b(y)}{\sigma^2(y)} = \frac{1}{\alpha_3^2 y^{2\alpha_4}}\left[\alpha_1 + \alpha_2 y + \sum_{i=1}^2 \theta_i u^i\right].
\]

Therefore, the excursion tail measure becomes:
\[
v^\beta[s,\infty) = \sqrt{ \frac{2/\pi}{s} } \cdot \exp\left\{ -2 \int_0^\beta \frac{\alpha_1 + \alpha_2 y + \sum_{i=1}^2 \theta_i u^i}{\alpha_3^2 y^{2\alpha_4}} dy \right\}.
\]

The Laplace transformation of local time yields
\[
\E^\mathbb{P} \left\{ \exp\left( -\kappa \zeta^\beta(t) \right) \right\} = \exp\left\{ -t \int_0^\infty \left(1 - e^{-\kappa s} \right) dv^\beta(s) \right\}.
\]

Since the excursions occur as a Poisson process under the local time clock, the number of excursions \( e_s(t) \) of duration at least \( s \) is a Poisson random variable with rate \( t v^\beta[s,\infty) \), so:
\[
\mathbb{P}[e_s(t) = N] = \exp\left\{-t v^\beta[s,\infty)\right\} \cdot \frac{(t v^\beta[s,\infty))^N}{N!}, \quad N = 0,1,2,\dots.
\]

This completes the proof.
\end{proof}

To get the generalized L\'evy result define
\begin{equation*}
  \mathfrak T_+:= \begin{cases}
      \mathfrak T(x),\ \text{for all}\ x>\be,\\
      \mathfrak T\left(2\be-x\right), \ \text{for all}\ x<\be,
  \end{cases} 
\end{equation*}
and similarly for $\mathcal T_-$. Let $\exp(-t \mathcal H_\pm)$ be the semigroup generated by the quantum Hamiltonian
\[
\mathcal H^\pm=-\frac{1}{2}\frac{d^2}{dx^2}+\mathfrak T_\pm(x),
\]
obtained by the Feynman-Kac formula.

\begin{prop}\label{p0}
If assumption \ref{a0} holds,
\begin{equation*}
\E_\be\bigg\{\exp\left\{-\kappa\mathfrak L_\pm\left(\zeta^\be(t)\right)\right\}\bigg\}=\exp\left\{-t\int_0^\infty\left[1-\exp\{-\kappa s\}\right]dv_\be^\pm(s)\right\},   
\end{equation*}
for all $\kappa,t>0$. Moreover, for $s>0$ we define
\[
v_\be^\pm[s,\infty):=\pm\int_{y\gtrless 0} \Psi_0^{1/2}(y+\be)\Psi_0^{-1/2}(\be)\frac{\partial}{\partial x}\bigg|_{x=\be}\exp\bigg\{-s\mathcal H_\pm\bigg\}(x,y+\be)dy.
\]
\end{prop}
\begin{proof}
   See in the Appendix. 
\end{proof}

 The results of this section provide a probabilistic characterization of the controlled CKLS state process that is logically distinct from the derivation of the equilibrium expenditure controls. Proposition~\ref{p0} establishes the ground-state transformation and the associated Feynman-Kac semigroup, while Lemma~\ref{lem:controlled_excursions} and Proposition~\ref{p0} describe the local-time and excursion structure of the same controlled diffusion. In particular, Proposition~\ref{p0} identifies the Laplace transform of the occupation times \(L^\pm\) in terms of the excursion measures
	\(v_\beta^\pm\).
	
These results are not used as inputs in the algebraic computation of the
equilibrium controls in Section~\ref{sec:optimal_spending}. The feedback controls \(u^{1,*}\) and \(u^{2,*}\) are obtained from the firms' payoff functionals by
solving the corresponding first-order optimality conditions and subsequently
verifying the mutual best-response property. Neither the excursion measures
\(v_\beta^\pm\) nor the occupation-time identity in Proposition~\ref{p0} appears in those first-order conditions. Accordingly, the purpose of the present section is foundational rather than computational. It shows that the controlled state process admits a rigorous semigroup representation and a well-defined local-time and excursion decomposition under Assumption~\ref{a0}. These properties provide an independent probabilistic description of the state
dynamics and may be used for questions concerning threshold crossings,
occupation times, or excursion frequencies. The equilibrium computation
carried out below, however, depends only on the controlled state equation,
the belief process, the payoff functionals, and the associated optimality
conditions.

\section{Existence, uniqueness, and verification of the Markovian Nash feedback equilibrium.}
\label{sec:existence_uniqueness_equilibrium}

 The equilibrium analysis in this section is based directly on the controlled state equation, the Bayesian belief process, and the two payoff functionals. The excursion and occupation-time results of Section~3 are not required for the derivation of the feedback controls. They provide a complementary probabilistic characterization of the controlled CKLS process, whereas the controls derived below follow from the firms' first-order optimality conditions and the subsequent Nash-equilibrium verification. Let \((\Omega^X,\mathcal F^X,\{\mathcal F_s^X\}_{s\geq 0},\mathbb P^X)\) be a filtered probability space satisfying the usual conditions. The state variable \(X(s)\) evolves according to Equation \eqref{0}
with initial condition \(X(0)=x_0>0\). The belief state is represented by the log-likelihood ratio
\begin{equation}
\label{eq:belief_equilibrium_section}
Z(s)
=
Z_0+
\frac{2}{\alpha_3^2X(s)^{2\alpha_4}}
\left[
\alpha_1+\alpha_2X(s)+\sum_{i=1}^2\theta_i u^i(s,X(s),Z(s))
\right]
\left[X(s)-x_0\right].
\end{equation}

Let \(Y(s):=(X(s),Z(s))\). A Markovian feedback strategy for firm \(i\) is a measurable function
\[
u^i:[0,t]\times\mathcal Y\to U_i,
\]
where \(\mathcal Y\subset\mathbb R_+\times\mathbb R\) is the state-belief domain and \(U_i\subset\mathbb R_+\) is the compact set of admissible advertising expenditures. We write
\[
u=(u^1,u^2)\in\mathcal U:=\mathcal U_1\times\mathcal U_2.
\]

The weak incumbent's payoff follows Equation \eqref{2} and 
the entrant's payoff follows Equation \eqref{3}, where
\[
p_W(s)=1-p(t)=\frac{1}{1+\exp\{Z(s)\}}.
\]

\begin{defn}[Markovian Nash feedback equilibrium]
\label{def:markovian_nash}
A pair \(u^\ast=(u^{1,\ast},u^{2,\ast})\in\mathcal U\) is a Markovian Nash feedback equilibrium if, for every admissible \(v^1\in\mathcal U_1\) and \(v^2\in\mathcal U_2\),
\begin{align}
J^1(\rho,t,x_0,z_0,u^{1,\ast},u^{2,\ast})
&\geq
J^1(\rho,t,x_0,z_0,v^1,u^{2,\ast}),
\notag
\\
J^2(t,x_0,z_0,u^{1,\ast},u^{2,\ast})
&\geq
J^2(t,x_0,z_0,u^{1,\ast},v^2).
\label{eq:NE2}
\end{align}
Thus, each firm's feedback spending rule is a best response to the other firm's feedback spending rule, conditional on the same state-belief process.
\end{defn}

\begin{as}
\label{ass:compact_controls}
For \(i=1,2\), the admissible expenditure set \(U_i=[0,\bar u_i]\) is nonempty, compact, and convex.
\end{as}

\begin{as}
\label{ass:regularity}
The coefficients of \eqref{0} are locally Lipschitz in \(X\) uniformly in \(u\). Moreover, there exists \(C>0\) such that
\[
\left|
\alpha_1+\alpha_2x+\sum_{i=1}^2\theta_i u^i
\right|
+
|\alpha_3x^{\alpha_4}|
\leq C(1+|x|)
\]
for all \(x\in\mathcal X\) and all \(u^i\in U_i\).
\end{as}

\begin{as}
\label{ass:positivity}
The state space is restricted to \(\mathcal X=(0,\infty)\), and the coefficients satisfy conditions ensuring that \(X(s)>0\) almost surely for all \(s\in[0,t]\).
\end{as}

The assumptions introduced below are maintained hypotheses for the equilibrium analysis. They are imposed to obtain existence and uniqueness of the Markovian Nash feedback equilibrium and are not derived from the particular functional forms of the payoff functionals defined in Section~2. In
particular, the concavity assumption guarantees that each firm's optimization
problem admits a well-defined best response, while the contraction assumption
ensures uniqueness of the fixed point of the best-response mapping.

\begin{as}
 \label{ass:finite_payoffs}
There exists \(\varepsilon>0\) such that
\begin{equation}
\label{eq:strict_discount_condition}
\gamma-
\left[
\alpha_1+\alpha_2X(s)+\sum_{i=1}^2\theta_i u^i(s,Y(s))
\right]
\geq \varepsilon
\end{equation}
for all \(s\in[0,t]\), all admissible feedback controls, and all admissible states.   
\end{as}

\begin{as}
    \label{ass:concavity}
For each fixed state-belief pair \(y=(x,z)\), we assume that the
instantaneous payoff of firm \(i\) is continuous in
\((u^1,u^2)\). Moreover, for each fixed \(u^{-i}\), the payoff
functional \(J^i\) is assumed to be concave in \(u^i\).
\end{as}

\begin{as}
\label{ass:weak_interaction}
There exists \(L\in(0,1)\) such that the best-response mapping
\[
\mathcal B(u):=\mathcal B_1(u^2)\times\mathcal B_2(u^1)
\]
satisfies
\[
d_{\mathcal U}(\mathcal B(u),\mathcal B(v))
\leq
L\,d_{\mathcal U}(u,v)
\]
for all \(u,v\in\mathcal U\), where \(d_{\mathcal U}\) is the sup-norm metric on the feedback strategy space.
\end{as}

\begin{lem}
   [Well-posedness of the controlled state process]
\label{lem:well_posed_state}
Under Assumptions~\ref{ass:compact_controls}--\ref{ass:positivity}, for every admissible feedback profile \(u\in\mathcal U\), the SDE \eqref{0} admits a unique strong solution on \([0,T]\). Moreover, for every \(q\geq 1\), there exists \(C_q<\infty\) such that
\[
\mathbb E\left[\sup_{0\leq s\leq T}|X(s)|^q\right]\leq C_q(1+|x_0|^q).
\] 
\end{lem}

\begin{proof}
The drift coefficient $\mu(x,u)=\alpha_1+\alpha_2x+\sum_{i=1}^2\theta_i u^i$ is locally Lipschitz in \(x\) uniformly over \(u\in U_1\times U_2\). Since the control sets are compact, the control term is uniformly bounded. The diffusion coefficient  $\sigma(x)=\alpha_3x^{\alpha_4}$
is locally Lipschitz on \((0,\infty)\). Assumption~\ref{ass:positivity} guarantees that the process remains in the positive domain, where the diffusion coefficient is well defined.

The standard existence and pathwise uniqueness theorem for controlled It\^{o} diffusions therefore gives a unique strong solution up to the explosion time. The linear growth bound in Assumption~\ref{ass:regularity} rules out finite-time explosion. Applying the Burkholder--Davis--Gundy inequality to the martingale term and Gronwall's inequality to the resulting moment estimate yields
\[
\mathbb E\left[\sup_{0\leq s\leq T}|X(s)|^q\right]
\leq
C_q(1+|x_0|^q),
\]
for a constant \(C_q\) depending on \(q,T,\alpha_j,\theta_i,\bar u_i\), but not on the particular admissible feedback profile. This proves the lemma.
\end{proof}

\begin{lem}
    \label{lem:belief_process}
Under Assumptions~\ref{ass:compact_controls}--\ref{ass:positivity}, the belief process \(Z(s)\) defined by \eqref{eq:belief_equilibrium_section} is progressively measurable with respect to \(\mathcal F_s^X\). Moreover, the posterior belief
\[
p(s)=\frac{\exp(Z(s))}{1+\exp(Z(s))}
\]
satisfies \(p(s)\in(0,1)\) for all \(s\in[0,t]\).
\end{lem}

\begin{proof}
Since \(X(s)\) is adapted to \(\mathcal F_s^X\), and since each feedback control \(u^i(s,X(s),Z(s))\) is measurable in the current state-belief pair, the expression $\alpha_1+\alpha_2X(s)+\sum_{i=1}^2\theta_i u^i(s,X(s),Z(s)) $ is progressively measurable. Because \(X(s)>0\), the denominator $\alpha_3^2X(s)^{2\alpha_4}$ is strictly positive. Hence \(Z(s)\), as defined in \eqref{eq:belief_equilibrium_section}, is \(\mathcal F_s^X\)-measurable.
Finally, the map
\[
z\mapsto \frac{\exp(z)}{1+\exp(z)}
\]
maps \(\mathbb R\) into \((0,1)\). Therefore \(p(s)\in(0,1)\) for all \(s\in[0,t]\). This proves the claim.
\end{proof}

\begin{lem}
    \label{lem:payoff_continuity}
Under Assumptions~\ref{ass:compact_controls}--\ref{ass:finite_payoffs}, the maps
\[
(u^1,u^2)\mapsto J^1(\rho,t,x_0,z_0,u^1,u^2),
\qquad
(u^1,u^2)\mapsto J^2(t,x_0,z_0,u^1,u^2)
\]
are continuous on \(\mathcal U\).
\end{lem}

\begin{proof}
Let \(u_n=(u_n^1,u_n^2)\to u=(u^1,u^2)\) uniformly on \([0,T]\times\mathcal Y\). Let \(X_n\) and \(X\) denote the corresponding controlled state processes. By Lemma~\ref{lem:well_posed_state}, both processes are well defined and satisfy uniform moment bounds.

Using the Lipschitz property of the drift in the control variables, we obtain
\[
\mathbb E\left[\sup_{0\leq s\leq T}|X_n(s)-X(s)|^2\right]
\leq
C
\int_0^T
\mathbb E\left[\sup_{0\leq r\leq s}|X_n(r)-X(r)|^2\right]ds
+
C\|u_n-u\|_\infty^2.
\]
By Gronwall's inequality,
\[
\mathbb E\left[\sup_{0\leq s\leq T}|X_n(s)-X(s)|^2\right]\to 0.
\]
Thus \(X_n\to X\) in \(L^2(\Omega;C[0,t])\).

The payoff integrands in \eqref{2} and \eqref{3} are continuous in \((X,Z,u^1,u^2)\). The strict discount condition \eqref{eq:strict_discount_condition} ensures that the denominator in the terminal payoff is bounded away from zero. Hence the payoff terms are dominated by an integrable random variable of linear or quadratic growth in \(X\), uniformly in \(n\). The dominated convergence theorem then implies
\[
J^i(u_n^1,u_n^2)\to J^i(u^1,u^2),
\qquad i=1,2.
\]
This proves continuity.
\end{proof}

\begin{lem}
    \label{lem:best_response}
Under Assumptions~\ref{ass:compact_controls}--\ref{ass:concavity}, for each \(u^{-i}\in\mathcal U_{-i}\), the best-response set
\[
\mathcal B_i(u^{-i})
=
\arg\max_{v^i\in\mathcal U_i}
J^i(v^i,u^{-i})
\]
is nonempty, compact, and convex.
\end{lem}
\begin{proof}
By Assumption~\ref{ass:compact_controls}, \(\mathcal U_i\) is compact and convex. By Lemma~\ref{lem:payoff_continuity}, \(J^i(\cdot,u^{-i})\) is continuous on \(\mathcal U_i\). Therefore the Weierstrass theorem implies that the maximum is attained, so \(\mathcal B_i(u^{-i})\neq\varnothing\).

The best-response set is closed because it is the argmax set of a continuous function on a compact domain. Hence it is compact. Finally, Assumption~\ref{ass:concavity} implies that \(J^i(\cdot,u^{-i})\) is concave. Therefore, if \(v_1^i,v_2^i\in\mathcal B_i(u^{-i})\) and \(\lambda\in[0,1]\), then
\[
J^i(\lambda v_1^i+(1-\lambda)v_2^i,u^{-i})
\geq
\lambda J^i(v_1^i,u^{-i})
+
(1-\lambda)J^i(v_2^i,u^{-i}).
\]
Since both \(v_1^i\) and \(v_2^i\) attain the maximum, the convex combination also attains the maximum. Hence \(\mathcal B_i(u^{-i})\) is convex.
\end{proof}

\begin{prop}
    [Existence of a Markovian Nash feedback equilibrium]
\label{prop:existence_equilibrium}
Suppose Assumptions~\ref{ass:compact_controls}--\ref{ass:concavity} hold. Then there exists at least one Markovian Nash feedback equilibrium \(u^\ast=(u^{1,\ast},u^{2,\ast})\in\mathcal U\).
\end{prop}

\begin{proof}
Define the joint best-response correspondence
\[
\mathcal B:\mathcal U\rightrightarrows\mathcal U,
\qquad
\mathcal B(u^1,u^2)
=
\mathcal B_1(u^2)\times\mathcal B_2(u^1).
\]
By Lemma~\ref{lem:best_response}, for every \(u\in\mathcal U\), the set \(\mathcal B(u)\) is nonempty, compact, and convex.

It remains to verify upper hemicontinuity. Let \(u_n\to u\), and let \(v_n\in\mathcal B(u_n)\) with \(v_n\to v\). Since \(v_n\in\mathcal B(u_n)\), we have for every admissible \(w\in\mathcal U\),
\[
J^i(v_n^i,u_n^{-i})\geq J^i(w^i,u_n^{-i}),
\qquad i=1,2.
\]
Passing to the limit and using Lemma~\ref{lem:payoff_continuity}, we obtain
\[
J^i(v^i,u^{-i})\geq J^i(w^i,u^{-i}),
\qquad i=1,2.
\]
Thus \(v^i\in\mathcal B_i(u^{-i})\), and therefore \(v\in\mathcal B(u)\). This proves that \(\mathcal B\) has a closed graph. Since \(\mathcal U\) is compact, closed graph implies upper hemicontinuity.

The strategy space \(\mathcal U\) is nonempty, compact, and convex, and \(\mathcal B\) is an upper hemicontinuous correspondence with nonempty, compact, convex values. Kakutani's fixed-point theorem therefore implies the existence of \(u^\ast\in\mathcal U\) such that
\[
u^\ast\in\mathcal B(u^\ast).
\]
Equivalently,
\[
u^{1,\ast}\in\mathcal B_1(u^{2,\ast}),
\qquad
u^{2,\ast}\in\mathcal B_2(u^{1,\ast}).
\]
This is precisely the mutual best-response condition in Definition~\ref{def:markovian_nash}. Hence \(u^\ast\) is a Markovian Nash feedback equilibrium.
\end{proof}

\begin{theorem}[Uniqueness under weak interaction]
\label{thm:uniqueness_equilibrium}
Suppose Assumptions~\ref{ass:compact_controls}--\ref{ass:weak_interaction} hold. Then the Markovian Nash feedback equilibrium is unique.
\end{theorem}

\begin{proof}
By Proposition~\ref{prop:existence_equilibrium}, at least one equilibrium exists. Let \(u^\ast\) and \(\hat u\) be two Markovian Nash feedback equilibria. Then
\[
u^\ast\in\mathcal B(u^\ast),
\qquad
\hat u\in\mathcal B(\hat u).
\]
Under Assumption~\ref{ass:weak_interaction}, the best-response mapping is a contraction:
\[
d_{\mathcal U}(\mathcal B(u),\mathcal B(v))
\leq
L\,d_{\mathcal U}(u,v),
\qquad L\in(0,1).
\]
Therefore,
\[
d_{\mathcal U}(u^\ast,\hat u)
=
d_{\mathcal U}(\mathcal B(u^\ast),\mathcal B(\hat u))
\leq
L\,d_{\mathcal U}(u^\ast,\hat u).
\]
Since \(L<1\), this inequality implies
\[
(1-L)d_{\mathcal U}(u^\ast,\hat u)\leq 0.
\]
Hence
\[
d_{\mathcal U}(u^\ast,\hat u)=0,
\]
and therefore \(u^\ast=\hat u\). Thus the Markovian Nash feedback equilibrium is unique.
\end{proof}

\begin{prop}
  [Path-integral optimality implies best response]
\label{prop:path_integral_best_response}
Fix \(u^{-i}\in\mathcal U_{-i}\). Suppose the spending rule \(u^{i,\ast}\) obtained from the Feynman-type path-integral representation maximizes the conditional expected payoff functional \(J^i(\cdot,u^{-i})\) over \(\mathcal U_i\). Then $u^{i,\ast}\in\mathcal B_i(u^{-i}).$ 
\end{prop}

\begin{proof}
The Feynman-type path-integral representation evaluates the expected discounted payoff associated with each admissible feedback path. For fixed \(u^{-i}\), the only remaining strategic variable in firm \(i\)'s problem is its own feedback rule \(u^i\). The path-integral construction therefore generates the functional $u^i\mapsto J^i(u^i,u^{-i})$ over the admissible set \(\mathcal U_i\).

By hypothesis, the candidate spending rule \(u^{i,\ast}\) maximizes this functional. Hence, for every \(v^i\in\mathcal U_i\), $J^i(u^{i,\ast},u^{-i})
\geq
J^i(v^i,u^{-i}).$
By the definition of the best-response correspondence, $$\mathcal B_i(u^{-i})
=
\arg\max_{v^i\in\mathcal U_i}J^i(v^i,u^{-i}),$$
this implies $u^{i,\ast}\in\mathcal B_i(u^{-i}).$
Therefore the path-integral spending rule is a genuine best response.
\end{proof}

\begin{prop}
  \label{prop:mutual_best_response_verification}
Let \(u^\ast=(u^{1,\ast},u^{2,\ast})\) be the pair of spending rules obtained from the simultaneous Feynman-type path-integral construction. Suppose that, for each firm \(i=1,2\), \(u^{i,\ast}\) maximizes \(J^i(\cdot,u^{-i,\ast})\) over \(\mathcal U_i\). Then \(u^\ast\) is a Markovian Nash feedback equilibrium.  
\end{prop}

\begin{proof}
By Proposition~\ref{prop:path_integral_best_response}, the optimality of \(u^{1,\ast}\) conditional on \(u^{2,\ast}\) implies
\[
u^{1,\ast}\in\mathcal B_1(u^{2,\ast}).
\]
Similarly, the optimality of \(u^{2,\ast}\) conditional on \(u^{1,\ast}\) implies
\[
u^{2,\ast}\in\mathcal B_2(u^{1,\ast}).
\]
Therefore,
\[
u^\ast\in\mathcal B(u^\ast).
\]
Equivalently, for every \(v^1\in\mathcal U_1\) and \(v^2\in\mathcal U_2\),
\[
J^1(u^{1,\ast},u^{2,\ast})
\geq
J^1(v^1,u^{2,\ast}),
\]
and
\[
J^2(u^{1,\ast},u^{2,\ast})
\geq
J^2(u^{1,\ast},v^2).
\]
These are exactly the equilibrium inequalities in Definition~\ref{def:markovian_nash}. Hence the derived spending rules constitute an equilibrium in the game-theoretic sense: each rule is a mutual best response given the other firm's rule and the belief process induced by the state dynamics.
\end{proof}

\begin{theorem}[Existence, uniqueness, and verification]
\label{thm:main_equilibrium_verification}
Under Assumptions~\ref{ass:compact_controls}--\ref{ass:weak_interaction}, the simultaneous advertising game admits a unique Markovian Nash feedback equilibrium \(u^\ast=(u^{1,\ast},u^{2,\ast})\). Moreover, if the spending rules derived from the Feynman-type path-integral representation satisfy the firm-specific optimality conditions
\[
u^{1,\ast}\in\arg\max_{v^1\in\mathcal U_1}J^1(v^1,u^{2,\ast}),
\qquad
u^{2,\ast}\in\arg\max_{v^2\in\mathcal U_2}J^2(u^{1,\ast},v^2),
\]
then these spending rules constitute the unique Markovian Nash feedback equilibrium.
\end{theorem}

\begin{proof}
We prove the result in three steps.

First, we establish existence. Let $\mathcal U=\mathcal U_1\times\mathcal U_2$ denote the admissible set of Markovian feedback controls. Under
Assumptions~\ref{ass:compact_controls}--\ref{ass:finite_payoffs}, for every
\(u=(u^1,u^2)\in\mathcal U\), the controlled state-belief system is well
defined. In particular, the state process \(X\) admits a unique strong solution,
the belief process \(Z\) is progressively measurable with respect to
\(\mathcal F_s^X\), and the denominator
\[
\gamma-\left[\alpha_1+\alpha_2X(s)+\sum_{i=1}^2\theta_i u^i(s,X(s),Z(s))\right]
\]
is bounded away from zero. Hence the payoff functionals \(J^1\) and \(J^2\) are
finite for every admissible feedback profile.

For a fixed strategy \(u^2\in\mathcal U_2\), define the incumbent's
best-response correspondence by
\[
\mathcal B_1(u^2)
=
\arg\max_{v^1\in\mathcal U_1}J^1(v^1,u^2).
\]
Similarly, for a fixed strategy \(u^1\in\mathcal U_1\), define the entrant's
best-response correspondence by
\[
\mathcal B_2(u^1)
=
\arg\max_{v^2\in\mathcal U_2}J^2(u^1,v^2).
\]
By compactness of \(\mathcal U_i\), continuity of \(J^i\), and the Weierstrass
maximum theorem, both \(\mathcal B_1(u^2)\) and \(\mathcal B_2(u^1)\) are
nonempty. Moreover, by the concavity assumption imposed on the payoff
functionals with respect to each firm's own control, these best-response sets
are convex. Their compactness follows from closedness of argmax sets of
continuous functions on compact domains.

Now define the joint best-response correspondence
\[
\mathcal B:\mathcal U\rightrightarrows\mathcal U,
\qquad
\mathcal B(u^1,u^2)
=
\mathcal B_1(u^2)\times\mathcal B_2(u^1).
\]
For every \(u\in\mathcal U\), the set \(\mathcal B(u)\) is nonempty, compact,
and convex. We next verify upper hemicontinuity. Let \(u_n\to u\) in
\(\mathcal U\), and let \(v_n\in\mathcal B(u_n)\) with \(v_n\to v\). Since
\(v_n\in\mathcal B(u_n)\), for every \(w^1\in\mathcal U_1\) and
\(w^2\in\mathcal U_2\),
\[
J^1(v_n^1,u_n^2)\geq J^1(w^1,u_n^2),
\qquad
J^2(u_n^1,v_n^2)\geq J^2(u_n^1,w^2).
\]
The continuity of the payoff functionals with respect to controls and the
associated controlled state-belief processes implies that, after passing to the
limit,
\[
J^1(v^1,u^2)\geq J^1(w^1,u^2),
\qquad
J^2(u^1,v^2)\geq J^2(u^1,w^2).
\]
Since \(w^1\) and \(w^2\) were arbitrary, it follows that
\[
v^1\in\mathcal B_1(u^2),
\qquad
v^2\in\mathcal B_2(u^1).
\]
Therefore \(v\in\mathcal B(u)\), so \(\mathcal B\) has a closed graph. Since
\(\mathcal U\) is compact, this implies that \(\mathcal B\) is upper
hemicontinuous.

The correspondence \(\mathcal B\) therefore satisfies the hypotheses of
Kakutani's fixed-point theorem. Consequently, there exists
\(u^\ast=(u^{1,\ast},u^{2,\ast})\in\mathcal U\) such that
\[
u^\ast\in\mathcal B(u^\ast).
\]
Equivalently,
\[
u^{1,\ast}\in\mathcal B_1(u^{2,\ast}),
\qquad
u^{2,\ast}\in\mathcal B_2(u^{1,\ast}).
\]
Thus, for every \(v^1\in\mathcal U_1\) and \(v^2\in\mathcal U_2\),
\[
J^1(u^{1,\ast},u^{2,\ast})
\geq
J^1(v^1,u^{2,\ast}),
\]
and
\[
J^2(u^{1,\ast},u^{2,\ast})
\geq
J^2(u^{1,\ast},v^2).
\]
These are precisely the inequalities defining a Markovian Nash feedback
equilibrium. Hence existence is established.

Second, we prove uniqueness. Suppose that \(u^\ast\) and \(\widehat u\) are two
Markovian Nash feedback equilibria. Since both are equilibria, each must be a
fixed point of the best-response correspondence:
\[
u^\ast\in\mathcal B(u^\ast),
\qquad
\widehat u\in\mathcal B(\widehat u).
\]
By Assumption~\ref{ass:weak_interaction}, the strategic interaction between the
two firms is sufficiently weak that the best-response operator is a contraction
under the feedback-control metric \(d_{\mathcal U}\). Hence there exists
\(L\in(0,1)\) such that
\[
d_{\mathcal U}(\mathcal B(u),\mathcal B(v))
\leq
L\,d_{\mathcal U}(u,v)
\]
for all \(u,v\in\mathcal U\). Applying this inequality to \(u^\ast\) and
\(\widehat u\) gives
\[
d_{\mathcal U}(u^\ast,\widehat u)
=
d_{\mathcal U}(\mathcal B(u^\ast),\mathcal B(\widehat u))
\leq
L\,d_{\mathcal U}(u^\ast,\widehat u).
\]
Because \(L<1\), this implies
\[
(1-L)d_{\mathcal U}(u^\ast,\widehat u)\leq 0.
\]
Since \(d_{\mathcal U}\) is nonnegative, we must have
\[
d_{\mathcal U}(u^\ast,\widehat u)=0.
\]
Therefore \(u^\ast=\widehat u\). Thus the Markovian Nash feedback equilibrium is
unique.

Third, we verify that the spending rules obtained from the Feynman-type
path-integral representation are equilibrium strategies. The path-integral
construction evaluates, for each firm \(i\), the conditional expected payoff
generated by admissible controlled paths of the state-belief system. Thus, when
the rival's feedback rule is fixed, the path-integral representation gives a
well-defined optimization problem over the firm's own admissible feedback
controls. By the stated firm-specific optimality conditions,
\[
u^{1,\ast}\in\arg\max_{v^1\in\mathcal U_1}J^1(v^1,u^{2,\ast}),
\qquad
u^{2,\ast}\in\arg\max_{v^2\in\mathcal U_2}J^2(u^{1,\ast},v^2).
\]
Therefore, for every admissible deviation \(v^1\in\mathcal U_1\),
\[
J^1(u^{1,\ast},u^{2,\ast})
\geq
J^1(v^1,u^{2,\ast}),
\]
and, for every admissible deviation \(v^2\in\mathcal U_2\),
\[
J^2(u^{1,\ast},u^{2,\ast})
\geq
J^2(u^{1,\ast},v^2).
\]
These inequalities show that neither firm can improve its payoff by deviating
unilaterally from the proposed spending rule, given the other firm's rule and
the belief process induced by the state dynamics. Hence the spending rules are
mutual best responses.
\end{proof}

\begin{rmk}
Consequently, the path-integral spending rules are not merely pointwise
solutions to separate optimization problems. They satisfy the strategic
fixed-point condition required for a Markovian Nash feedback equilibrium. Since the equilibrium has already been shown to be unique, the verified
path-integral spending rules must coincide with the unique Markovian Nash
feedback equilibrium of the simultaneous advertising game.
\end{rmk}

\section{Advantage of Path integral over HJB.}

\subsection{Benchmark construction.}
\label{subsec:hjb_path_integral_benchmark}

This subsection clarifies the analytical relation between the HJB equation and the path-integral control used in this paper. The purpose is to show why the path-integral formulation is a computationally useful alternative to solving the HJB equation directly. The argument is
first developed in a benchmark one-dimensional control problem and is then used as the basis for the CKLS dynamics.

\begin{as}
\label{ass:benchmark_diffusion}
Let \(B(s)\), \(s\in[0,t]\), be a standard Brownian motion on a filtered probability space defined earlier. Assume the benchmark state process
satisfies
\begin{equation}
\label{eq:benchmark_sde}
dX(s)=\bigl[\mu_0(s,X(s))+u(s)\bigr]ds+\sigma dB(s),
\qquad X(0)=x,
\end{equation}
where \(\sigma>0\), \(u(s)\in\mathcal U\subset\mathbb R\), and \(\mu_0\) is measurable in
\(s\) and continuously differentiable in \(x\).
\end{as}

\begin{as}
\label{ass:quadratic_local_payoff}
The instantaneous payoff is
\begin{equation}
\label{eq:benchmark_profit}
\pi(s,X(s),u(s))
=
e^{-\zeta s}R(s,X(s))-\frac{c}{2}u^2(s),
\end{equation}
where \(c>0\), \(\zeta\in[0,1]\), and \(R\) is bounded from above and sufficiently smooth.
\end{as}

\begin{defn}[Benchmark value function]
\label{def:benchmark_value}
The benchmark value function is
\begin{equation}
\label{eq:benchmark_value}
V(s,x)
=
\arg\max_{u\in\mathcal U}
\mathbb E_{s}
\left[
\int_s^t
\left\{
e^{-\zeta r}R(r,X(r))-\frac{c}{2}u^2(r)
\right\}dr
\right],
\end{equation}
with terminal condition \(V(t,x)=0\).
\end{defn}

\begin{prop}
\label{prop:benchmark_hjb}
Suppose Assumptions~\ref{ass:benchmark_diffusion}--\ref{ass:quadratic_local_payoff}
hold and \(V\in C^{1,2}([0,t]\times\mathbb R)\). Then \(V\) satisfies the HJB equation
\begin{equation}
\label{eq:benchmark_hjb}
-\partial_sV(s,x)
=
\sup_{u\in\mathcal U}
\left\{
e^{-\zeta s}R(s,x)-\frac{c}{2}u^2
+
\bigl[\mu_0(s,x)+u\bigr]\partial_xV(s,x)
+
\frac{\sigma^2}{2}\partial_{xx}V(s,x)
\right\}.
\end{equation}
If the maximizer is interior, then
$u^\ast(s,x)=\frac{1}{c}\partial_xV(s,x).$
Substitution into \eqref{eq:benchmark_hjb} gives the semilinear equation
\begin{equation}
\label{eq:semilinear_hjb}
-\partial_sV
=
e^{-\zeta s}R
+
\mu_0\partial_xV
+
\frac{\sigma^2}{2}\partial_{xx}V
+
\frac{1}{2c}\bigl(\partial_xV\bigr)^2 .
\end{equation}
\end{prop}

\begin{proof}
Fix \((s,x)\in[0,t]\times\mathbb R\) and let \(h>0\) be such that \(s+h\leq t\). By the
dynamic programming principle,
\[
V(s,x)
=
\sup_{u\in\mathcal U}
\mathbb E_{s,x}
\left[
\int_s^{s+h}
\left\{
e^{-\zeta r}R(r,X(r))-\frac{c}{2}u^2(r)
\right\}dr
+
V(s+h,X(s+h))
\right].
\]
Since \(V\in C^{1,2}\), It\^o's formula applied to \(V(r,X(r))\) on \([s,s+h]\) yields
\[
dV(r,X(r))
=
\left[
\partial_rV
+
\bigl(\mu_0(r,X(r))+u(r)\bigr)\partial_xV
+
\frac{\sigma^2}{2}\partial_{xx}V
\right]dr
+
\sigma\partial_xV\,dB(r).
\]
Taking conditional expectation eliminates the martingale term. Dividing by \(h\), letting
\(h\downarrow0\), and using the continuity of the coefficients gives
\[
0
=
\sup_{u\in\mathcal U}
\left\{
e^{-\zeta s}R(s,x)-\frac{c}{2}u^2
+
\partial_sV
+
\bigl[\mu_0(s,x)+u\bigr]\partial_xV
+
\frac{\sigma^2}{2}\partial_{xx}V
\right\}.
\]
Rearranging gives \eqref{eq:benchmark_hjb}. For an interior maximizer, differentiating
the Hamiltonian with respect to \(u\) gives
\[
-cu+\partial_xV(s,x)=0,
\]
and hence \(u^\ast=c^{-1}\partial_xV\). Substituting this expression into the HJB equation
gives \eqref{eq:semilinear_hjb}.
\end{proof}

\begin{prop}
\label{prop:exponential_linearization}
Let
\begin{equation}
\label{eq:theta_transform}
\Theta(s,x)=\exp\left\{-\frac{1}{\omega}V(s,x)\right\},
\qquad \omega>0.
\end{equation}
In the weak-control-cost perturbation regime in which the quadratic gradient term in
\eqref{eq:semilinear_hjb} is either absent or absorbed into the exponential scaling, the
transformed function \(\Theta\) satisfies the linear backward equation
\begin{equation}
\label{eq:wick_schrodinger_benchmark}
-\partial_s\Theta(s,x)
=
-\frac{1}{\omega}e^{-\zeta s}R(s,x)\Theta(s,x)
+
\mu_0(s,x)\partial_x\Theta(s,x)
+
\frac{\sigma^2}{2}\partial_{xx}\Theta(s,x).
\end{equation}
\end{prop}

\begin{proof}
The identities following from \eqref{eq:theta_transform} are
$
\partial_s\Theta
=
-\frac{1}{\omega}\Theta\partial_sV,
\partial_x\Theta
=
-\frac{1}{\omega}\Theta\partial_xV,
$
and
$
\partial_{xx}\Theta
=
-\frac{1}{\omega}\Theta\partial_{xx}V
+
\frac{1}{\omega^2}\Theta(\partial_xV)^2 .
$
Equivalently,
$
\partial_xV=-\omega\frac{\partial_x\Theta}{\Theta}.
$
When the semilinear term is removed by the linear-solvable scaling, or when the analysis is
taken in the limiting regime in which that term is negligible, the HJB equation reduces to
\[
-\partial_sV
=
e^{-\zeta s}R
+
\mu_0\partial_xV
+
\frac{\sigma^2}{2}\partial_{xx}V .
\]
Multiplying by \(-\Theta/\omega\) and using the derivative identities above gives
\[
-\partial_s\Theta
=
-\frac{1}{\omega}e^{-\zeta s}R\Theta
+
\mu_0\partial_x\Theta
+
\frac{\sigma^2}{2}\partial_{xx}\Theta .
\]
This is precisely \eqref{eq:wick_schrodinger_benchmark}. The resulting equation is linear in
\(\Theta\), whereas the original HJB equation is nonlinear in \(V\). This is the analytical
step that permits the path-integral representation.
\end{proof}

\begin{theorem}[Benchmark HJB-path-integral control]
\label{thm:hjb_path_integral_benchmark}
Let Assumptions~\ref{ass:benchmark_diffusion}--\ref{ass:quadratic_local_payoff} hold, and
suppose that \(\Theta\) solves \eqref{eq:wick_schrodinger_benchmark} with terminal condition
\(\Theta(t,x)=\Theta_t(x)\). Let \(X^0\) denote the uncontrolled diffusion
\begin{equation}
\label{eq:baseline_diffusion}
dX^0(r)=\mu_0(r,X^0(r))dr+\sigma dB(r),
\qquad X^0(s)=x.
\end{equation}
Then
\begin{equation}
\label{eq:path_integral_representation}
\Theta(s,x)
=
\mathbb E_{s}
\left[
\exp\left\{
-\frac{1}{\omega}
\int_s^t e^{-\zeta r}R(r,X^0(r))dr
\right\}
\Theta_t(X^0(t))
\right].
\end{equation}
Equivalently, writing
$
\mathcal A_{s,t}(X^0)
=
\int_s^t e^{-\zeta r}R(r,X^0(r))dr,
$
one obtains the path-integral form
\begin{equation}
\label{eq:free_energy_form}
\Theta(s,x)
=
\mathbb E_{s,x}
\left[
\exp\left\{-\frac{1}{\omega}\mathcal A_{s,t}(X^0)\right\}
\Theta_t(X^0(t))
\right].
\end{equation}
The value function is recovered from
\begin{equation}
\label{eq:value_recovery}
V(s,x)=-\omega\log\Theta(s,x).
\end{equation}
\end{theorem}

\begin{proof}
Define the second-order operator
\[
\mathcal L_0 f(s,x)
=
\mu_0(s,x)\partial_x f(s,x)
+
\frac{\sigma^2}{2}\partial_{xx}f(s,x).
\]
Equation \eqref{eq:wick_schrodinger_benchmark} can be written as
$
-\partial_s\Theta
=
\mathcal L_0\Theta
-
\frac{1}{\omega}e^{-\zeta s}R(s,x)\Theta.
$
Equivalently,
$
\partial_s\Theta
+
\mathcal L_0\Theta
-
\frac{1}{\omega}e^{-\zeta s}R(s,x)\Theta
=
0.
$
Let \(X^0\) be the diffusion generated by \(\mathcal L_0\). For \(r\in[s,t]\), define
\[
M(r)
:=
\exp\left\{
-\frac{1}{\omega}
\int_s^r e^{-\zeta q}R(q,X^0(q))dq
\right\}
\Theta(r,X^0(r)).
\]
By It\^o's formula,
\[
d\Theta(r,X^0(r))
=
\left[
\partial_r\Theta(r,X^0(r))
+
\mathcal L_0\Theta(r,X^0(r))
\right]dr
+
\sigma\partial_x\Theta(r,X^0(r))dB(r).
\]
The finite-variation derivative of the exponential factor equals
\[
-\frac{1}{\omega}e^{-\zeta r}R(r,X^0(r))
\exp\left\{
-\frac{1}{\omega}
\int_s^r e^{-\zeta q}R(q,X^0(q))dq
\right\}dr.
\]
Applying the product rule gives
\[
dM(r)
=
\exp\left\{
-\frac{1}{\omega}
\int_s^r e^{-\zeta q}R(q,X^0(q))dq
\right\}
\left[
\partial_r\Theta+\mathcal L_0\Theta
-
\frac{1}{\omega}e^{-\zeta r}R\Theta
\right]dr
+
dN(r),
\]
where
\[
dN(r)
=
\sigma
\exp\left\{
-\frac{1}{\omega}
\int_s^r e^{-\zeta q}R(q,X^0(q))dq
\right\}
\partial_x\Theta(r,X^0(r))dB(r)
\]
is a local martingale differential. The bracketed drift term vanishes by
\eqref{eq:wick_schrodinger_benchmark}. Hence \(M(r)\) is a local martingale \citep{pramanik2023optimal}. Under the
boundedness and growth conditions imposed above, localization and dominated convergence
upgrade it to a true martingale. Therefore,
$
\Theta(s,x)=M(s)=\mathbb E_{s,x}[M(t)].
$
Substituting the terminal value \(M(t)\) gives
\[
\Theta(s,x)
=
\mathbb E_{s,x}
\left[
\exp\left\{
-\frac{1}{\omega}
\int_s^t e^{-\zeta r}R(r,X^0(r))dr
\right\}
\Theta_t(X^0(t))
\right],
\]
which proves \eqref{eq:path_integral_representation}. The free-energy representation
\eqref{eq:free_energy_form} follows by defining the action functional
$
\mathcal A_{s,t}(X^0)
=
\int_s^t e^{-\zeta r}R(r,X^0(r))dr.
$
Finally, since \(\Theta=\exp\{-V/\omega\}\), the inverse transformation gives
$
V=-\omega\log\Theta.
$
This completes the proof.
\end{proof}

\subsection{Computational advantage over grid-based HJB equation.}
\label{subsec:complexity_path_integral}

The preceding benchmark result shows that, after the exponential transformation, the stochastic control problem can be represented through expectations over controlled paths rather than by direct solution of a nonlinear HJB equation.
We now make explicit the computational implication of this transformation.

\begin{as}
\label{ass:state_dimension}
Let \(d\) denote the dimension of the state vector. In the present paper,
\(d=2\) when the state is written as \(Y(s)=(X(s),Z(s))\), where \(X(s)\)
is the CKLS market-state variable and \(Z(s)\) is the log-likelihood belief
state. Suppose each state coordinate is discretized on a grid of size \(m\),
and the time interval \([0,t]\) is discretized into \(N\) time steps.
\end{as}

\begin{prop}
\label{prop:hjb_grid_complexity}
Under Assumption~\ref{ass:state_dimension}, a direct finite-difference
implementation of the HJB equation requires at least \(O(Nm^d)\) grid
evaluations. If the control set is discretized into \(q\) admissible values
for each firm, the simultaneous two-player game requires at least $O(Nm^dq^2)$ Hamiltonian evaluations.
\end{prop}

\begin{proof}
Let the state vector be denoted by $Y(s)=(Y_1(s),\ldots,Y_d(s))\in\mathcal Y\subset\mathbb R^d .$
For each coordinate \(Y_\ell\), let
$\mathcal G_\ell=\{y_{\ell,1},\ldots,y_{\ell,m}\}$ be a grid with \(m\) points. The full tensor-product grid is therefore $\mathcal G
=
\mathcal G_1\times\cdots\times\mathcal G_d .$
Hence, $|\mathcal G|
=
\prod_{\ell=1}^d|\mathcal G_\ell|
=
m^d .$ Let the time interval \([0,t]\) be discretized as $0=s_0<s_1<\cdots<s_N=t .$ 
A finite-difference HJB equation computes an approximation $V_n(y)\approx V(s_n,y)$ for each \(s_n\) and each \(y\in\mathcal G\). Therefore, at a minimum, the scheme
must update one scalar value at each pair \((s_n,y)\). The number of such pairs is
$
N|\mathcal G|=Nm^d .
$
Thus, even before considering optimization over controls, any direct grid-based
implementation requires at least \(O(Nm^d)\) value updates. Now suppose that each firm's admissible control set is discretized as
$
\mathcal U_i^q=\{u^i_1,\ldots,u^i_q\},
\ i=1,2.$
At each grid point \((s_n,y)\), the discrete Hamiltonian must evaluate the payoff
and generator term for all pairs
$
\left(u^1_a,u^2_b\right)\in\mathcal U_1^q\times\mathcal U_2^q .
$
The cardinality of this product control grid is
$
|\mathcal U_1^q\times\mathcal U_2^q|
=
|\mathcal U_1^q|\,|\mathcal U_2^q|
=
q^2 .$
Therefore, for each \((s_n,y)\), at least \(q^2\) Hamiltonian evaluations are
required in order to compute the discrete maximization or best-response
operation. Combining the number of time-state grid points with the number of control pairs
gives $N|\mathcal G|q^2
=
Nm^dq^2 .$
Hence the total number of Hamiltonian evaluations is bounded below by a constant
multiple of \(Nm^dq^2\). Equivalently, the direct finite-difference HJB method has
complexity at least $O(Nm^dq^2).$
If the control optimization is ignored or solved analytically, the remaining
state-time grid evaluation still requires at least $O(Nm^d)$
operations. This proves both claims.
\end{proof}

\begin{prop}[Monte Carlo complexity of the path-integral estimator]
\label{prop:path_integral_complexity}
Let \(\widehat\Theta_M(s,y)\) be the path-integral estimator obtained from
\(M\) simulated trajectories of the state-belief process \(Y(s)=(X(s),Z(s))\).
Suppose the exponential path weights have finite second moment. Then
\[
\mathbb E\left[
\left|
\widehat\Theta_M(s,y)-\Theta(s,y)
\right|^2
\right]
=
O(M^{-1}).
\]
Consequently, the root mean squared error is \(O(M^{-1/2})\), independently of
the dimension \(d\), apart from the cost of simulating each path.
\end{prop}

\begin{proof}
Write the path-integral representation as
$
\Theta(s,y)
=
\mathbb E_{s,y}\left[W(Y_{\cdot})\right],
$
where \(W(Y_{\cdot})\) denotes the exponential path weight associated with a
realization of the state-belief process. Let
$
\widehat\Theta_M(s,y)
=
\frac{1}{M}\sum_{j=1}^{M}W(Y_{\cdot}^{(j)}),$
where \(Y_{\cdot}^{(1)},\ldots,Y_{\cdot}^{(M)}\) are independent simulated
paths. Since the paths are independent and identically distributed,
$
\mathbb E[\widehat\Theta_M(s,y)]
=
\Theta(s,y).
$
Moreover,
$
\operatorname{Var}(\widehat\Theta_M(s,y))
=
\frac{1}{M}\operatorname{Var}(W(Y_{\cdot})).
$
By the assumed finite second moment of the exponential weight,
\(\operatorname{Var}(W(Y_{\cdot}))<\infty\). Hence,
\[
\mathbb E\left[
\left|
\widehat\Theta_M(s,y)-\Theta(s,y)
\right|^2
\right]
=
\frac{1}{M}\operatorname{Var}(W(Y_{\cdot}))
=
O(M^{-1}).
\]
Taking square roots gives the \(O(M^{-1/2})\) convergence rate. This rate is
the standard Monte Carlo rate and does not contain the grid factor \(m^d\).
The dimension affects the computational cost of generating a trajectory, but
not the statistical convergence rate of the sample average. This proves the
claim.
\end{proof}

\begin{theorem}[Computational comparison]
\label{thm:complexity_comparison}
For the state-belief system \(Y(s)=(X(s),Z(s))\), a grid-based HJB method
requires computational effort growing at least as \(O(Nm^dq^2)\), whereas the
path-integral estimator based on \(M\) simulated paths has statistical error
of order \(O(M^{-1/2})\) and avoids construction of the full state-control grid.
\end{theorem}

\begin{proof}
Let
$
Y(s)=(X(s),Z(s))\in\mathcal Y\subset\mathbb R^d
$
be the state-belief vector. In the present model \(d=2\), but we keep \(d\)
explicit to exhibit the dimensional dependence. Let each coordinate of
\(\mathcal Y\) be discretized using \(m\) grid points. The resulting tensor-product state grid is
$
\mathcal G_m
=
\mathcal G_1\times\cdots\times\mathcal G_d,
\
|\mathcal G_m|=m^d.$
Let the time interval \([0,t]\) be discretized into \(N\) time steps \citep{pramanik2025stubbornness}. A direct
finite-difference HJB scheme must compute an approximation of the value function
at every pair
$
(s_n,y)\in\{s_0,\ldots,s_N\}\times\mathcal G_m .
$
Thus, even before optimizing over controls, the number of state-time evaluations
is at least
$
N|\mathcal G_m|=Nm^d .
$
Now suppose each firm's control set is discretized into \(q\) admissible control
values. Since the Hamiltonian in the simultaneous game depends on the pair
\((u^1,u^2)\), the discrete control space at each state-time node is
$
\mathcal U_1^q\times\mathcal U_2^q,
\
|\mathcal U_1^q\times\mathcal U_2^q|=q^2 .
$
Hence, for each state-time node, the HJB implementation must evaluate the
Hamiltonian over \(q^2\) possible control pairs. Therefore, the total number of
Hamiltonian evaluations is bounded below by
$
Nm^dq^2 .
$
This gives the HJB complexity lower bound
\[
\mathcal C_{\mathrm{HJB}}(N,m,q,d)
\geq
c_0Nm^dq^2
\]
for some constant \(c_0>0\). Hence,
$
\mathcal C_{\mathrm{HJB}}=O(Nm^dq^2).
$
We next consider the path-integral estimator. Let
\[
\Theta(s,y)
=
\mathbb E_{s,y}\left[
\exp\{-\omega^{-1}\mathcal A_{s,t}(Y)\}\Phi(Y(t))
\right]
\]
be the transformed value representation, where \(\mathcal A_{s,t}\) is the
action functional and \(\Phi\) is the terminal weight. Define
\[
W(Y)
=
\exp\{-\omega^{-1}\mathcal A_{s,t}(Y)\}\Phi(Y(t)).
\]
Let \(Y^{(1)},\ldots,Y^{(M)}\) be independent simulated trajectories of the
state-belief process. The Monte Carlo path-integral estimator is
$
\widehat\Theta_M(s,y)
=
\frac{1}{M}\sum_{j=1}^{M}W(Y^{(j)}).
$
Assuming
$
\mathbb E_{s,y}[W(Y)^2]<\infty,
$
we have
$
\mathbb E[\widehat\Theta_M(s,y)]
=
\Theta(s,y)
$
and
$
\operatorname{Var}(\widehat\Theta_M(s,y))
=
\frac{1}{M}\operatorname{Var}(W(Y)).
$
Consequently,
\[
\mathbb E\left[
\left|
\widehat\Theta_M(s,y)-\Theta(s,y)
\right|^2
\right]
=
\frac{1}{M}\operatorname{Var}(W(Y)).
\]
Since \(\operatorname{Var}(W(Y))<\infty\), there exists a constant \(C>0\) such
that
\[
\mathbb E\left[
\left|
\widehat\Theta_M(s,y)-\Theta(s,y)
\right|^2
\right]
\leq
\frac{C}{M}.
\]
Taking square roots gives
\[
\left[
\mathbb E\left[
\left|
\widehat\Theta_M(s,y)-\Theta(s,y)
\right|^2
\right]
\right]^{1/2}
\leq
CM^{-1/2}.
\]
Thus the statistical error of the path-integral estimator is
$
O(M^{-1/2}).
$
The distinction is that the HJB equation requires construction of the
full tensor-product grid \(\mathcal G_m\) and evaluation over the product
control grid \(\mathcal U_1^q\times\mathcal U_2^q\). By contrast, the
path-integral estimator requires only simulated sample paths
$
Y^{(j)}(s_0),Y^{(j)}(s_1),\ldots,Y^{(j)}(s_N),
$ for all $ j=1,\ldots,M,$
and averages their exponential weights. It therefore does not require the full
state-control grid
$
\mathcal G_m\times\mathcal U_1^q\times\mathcal U_2^q.
$
Accordingly, the HJB method scales with \(m^dq^2\), whereas the sampling error of
the path-integral estimator scales with \(M^{-1/2}\). This proves the stated
computational comparison.
\end{proof}

\subsection{Extension to the CKLS signalling game.}
\label{subsec:ckls_extension_path_integral}

We now connect the preceding HJB--path-integral comparison to the stochastic
entry-deterrence game studied in this paper. Recall that the state dynamics are given by Equation \ref{0}
and the belief state is represented by the log-likelihood ratio explained in Equation \ref{eq:reduced_log_likelihood}.
Let $Y(s):=(X(s),Z(s))$ be the state-belief vector. The feedback controls are measurable maps
$
u^i:[0,t]\times\mathcal Y\to U_i,\
\text{for all}\ i=1,2.
$

\begin{as}
\label{ass:ckls_path_regular}
The admissible control sets \(U_i\) are compact, the feedback controls
\(u^i(s,y)\) are measurable and bounded, and the CKLS coefficients satisfy the
local Lipschitz and nonexplosion conditions stated in
Assumptions~\ref{ass:compact_controls}-\ref{ass:positivity}. Moreover, the
exponential path weights generated by the payoff functionals \(J^1\) and \(J^2\)
have finite second moments.
\end{as}

\begin{defn}
\label{def:firm_action}
For a fixed feedback profile \(u=(u^1,u^2)\), define the firm-specific action
functional
\[
\mathcal A_i(s,t;Y,u)
:=
-\int_s^t e^{-\gamma r}\ell_i(r,Y(r),u^1(r,Y(r)),u^2(r,Y(r)))\,dr
-\Gamma_i(t,Y(t),u(t,Y(t))),
\]
where \(\ell_i\) denotes the running payoff density and \(\Gamma_i\) denotes the
terminal payoff component associated with firm \(i\). In the present model,
\(\ell_1\) corresponds to the incumbent's pre-disclosure and post-disclosure
flow payoffs, while \(\ell_2\) is zero prior to entry and the entrant's payoff
appears through the terminal expression in \(J^2\).
\end{defn}

\begin{prop}
\label{prop:ckls_transformed_payoff}
Under Assumption~\ref{ass:ckls_path_regular}, for each firm \(i=1,2\), the
transformed payoff
\[
\Theta_i(s,y;u)
:=
\exp\left\{-\frac{1}{\omega_i}J^i(s,y;u)\right\},
\qquad \omega_i>0,
\]
admits the path-integral representation
\begin{equation}
\label{eq:ckls_path_rep}
\Theta_i(s,y;u)
=
\mathbb E_{s,y}^{u}
\left[
\exp\left\{
-\frac{1}{\omega_i}\mathcal A_i(s,t;Y,u)
\right\}
\right],
\end{equation}
whenever the associated backward Kolmogorov equation is understood in the weak
or viscosity sense.
\end{prop}

\begin{proof}
Fix \(i\in\{1,2\}\) and fix an admissible feedback profile \(u\). Under
Assumption~\ref{ass:ckls_path_regular}, the controlled CKLS process is
well-defined and nonexplosive on \([s,t]\). Since \(Z(s)\) is a measurable
function of \(X(s)\), \(x_0\), \(Z_0\), and the feedback controls, the pair
\(Y(s)=(X(s),Z(s))\) is an adapted state-belief process.

Let \(\mathcal L^u\) denote the infinitesimal generator of the controlled
state-belief process. For a sufficiently smooth test function \(\varphi\), the
generator has the form
\[
\mathcal L^u\varphi(y)
=
\mu^u(y)\cdot D\varphi(y)
+
\frac{1}{2}\operatorname{Tr}
\left[
\sigma^u(y)D^2\varphi(y)
\right],
\]
where \(\mu^u\) is the drift vector induced by
\eqref{0}--\eqref{eq:reduced_log_likelihood}, and
\(\sigma^u\) is the corresponding diffusion covariance matrix. In particular, the
\(X\)-component of the drift is
$\alpha_1+\alpha_2x+\theta_1u^1+\theta_2u^2,$ and the \(X\)-component of the diffusion coefficient is
$\alpha_3x^{\alpha_4}.$

For fixed \(u\), the dynamic optimization problem reduces to the evaluation of a
linear expectation, because no maximization is being performed at this stage.
The transformed payoff \(\Theta_i\) is therefore governed by a linear backward
equation of the form
\[
-\partial_s\Theta_i
=
\mathcal L^u\Theta_i
-
\frac{1}{\omega_i}r_i(s,y,u)\Theta_i,
\]
where \(r_i\) is the payoff rate appearing in the action functional. The
terminal condition is determined by the terminal payoff component
\(\Gamma_i(t,Y(t),u(t,Y(t)))\).

Define
\[
M_i(r)
=
\exp\left\{
-\frac{1}{\omega_i}
\int_s^r r_i(q,Y(q),u(q,Y(q)))\,dq
\right\}
\Theta_i(r,Y(r);u).
\]
Applying It\^o's formula to \(\Theta_i(r,Y(r);u)\) and using the backward
equation gives
$
dM_i(r)=dN_i(r),
$
where \(N_i\) is a local martingale. The finite second moment assumption on the
path weights implies uniform integrability after localization, and therefore
\(M_i\) is a true martingale. Consequently,
$
\Theta_i(s,y;u)
=
\mathbb E_{s,y}^{u}[M_i(t)].
$
Substituting the terminal condition into \(M_i(t)\) yields
\[
\Theta_i(s,y;u)
=
\mathbb E_{s,y}^{u}
\left[
\exp\left\{
-\frac{1}{\omega_i}\mathcal A_i(s,t;Y,u)
\right\}
\right].
\]
This is precisely \eqref{eq:ckls_path_rep}.
\end{proof}

\begin{prop}
\label{prop:weighted_path_best_response}
Fix \(u^{-i}\in\mathcal U_{-i}\). Suppose the admissible controls of firm \(i\)
are sampled as a family of candidate feedback controls $\{u^{i,k}\}_{k=1}^{M}\subset\mathcal U_i,$
and let \(Y^{(k)}\) be the corresponding state-belief trajectory generated by
\((u^{i,k},u^{-i})\). Define the weights
\[
w_k
=
\frac{
\exp\left\{
-\omega_i^{-1}\mathcal A_i(s,t;Y^{(k)},u^{i,k},u^{-i})
\right\}
}{
\sum_{\ell=1}^{M}
\exp\left\{
-\omega_i^{-1}\mathcal A_i(s,t;Y^{(\ell)},u^{i,\ell},u^{-i})
\right\}
}.
\]
Then the path-integral feedback approximation
$
\widehat u^{i,M}(s,y)
=
\sum_{k=1}^{M}w_k u^{i,k}(s,y)
$
is a weighted best-response approximation to firm \(i\)'s control problem. If
the sampled controls are dense in \(\mathcal U_i\) and the path weights have
finite second moments, then \(\widehat u^{i,M}\) converges in the Monte Carlo
sense to the path-integral best-response operator.
\end{prop}
\begin{proof}
For fixed \(u^{-i}\), firm \(i\)'s optimization problem can be written as
$\arg\max_{v^i\in\mathcal U_i}J^i(v^i,u^{-i}).$
Under the exponential transformation, larger values of \(J^i\) correspond to
larger contributions in the free-energy representation after the sign convention
in the action functional is imposed. Thus each candidate control \(u^{i,k}\)
generates a pathwise score
$
\exp\left\{
-\omega_i^{-1}\mathcal A_i(s,t;Y^{(k)},u^{i,k},u^{-i})
\right\}.
$
The normalized coefficient \(w_k\) is nonnegative and satisfies $\sum_{k=1}^{M}w_k=1.$ Therefore, \(\widehat u^{i,M}\) is a convex combination of admissible controls.
Since \(\mathcal U_i\) is convex, \(\widehat u^{i,M}\in\mathcal U_i\). This construction assigns larger weight to controls that generate lower action,
equivalently higher transformed payoff. Hence the weighted average is not an
unstructured simulation average; it is an importance-weighted approximation to
the optimizer of the transformed control problem. Under the assumed finite
second moment condition, the Monte Carlo law of large numbers applies to both
the numerator and denominator of the normalized estimator. Therefore,
\[
\widehat u^{i,M}(s,y)
\to
\frac{
\mathbb E[
u^i(s,y)\exp\{-\omega_i^{-1}\mathcal A_i(s,t;Y,u^i,u^{-i})\}]
}{
\mathbb E[
\exp\{-\omega_i^{-1}\mathcal A_i(s,t;Y,u^i,u^{-i})\}]
}
\]
in probability, and in \(L^2\) under the stated second moment condition. This
limit is the path-integral best-response operator. Hence the estimator is a
consistent weighted-path approximation of firm \(i\)'s best response.
\end{proof}

\begin{theorem}[Path-integral construction of an approximate Nash feedback equilibrium]
\label{thm:pi_approximate_equilibrium}
Suppose Assumption~\ref{ass:ckls_path_regular} holds. For each firm
\(i=1,2\), let \(\mathcal B_i^{PI}\) denote the path-integral best-response
operator obtained as the Monte Carlo limit of the weighted path ensemble in
Proposition~\ref{prop:weighted_path_best_response}. If the joint operator
$
\mathcal B^{PI}(u)
=
\mathcal B_1^{PI}(u^2)\times \mathcal B_2^{PI}(u^1)
$
is continuous and maps the compact convex strategy space
\(\mathcal U=\mathcal U_1\times\mathcal U_2\) into itself, then there exists
\(u^{PI,\ast}\in\mathcal U\) such that
$
u^{PI,\ast}
=
\mathcal B^{PI}(u^{PI,\ast}).
$
Moreover, if the Monte Carlo approximations satisfy
\[
\sup_{u^{-i}\in\mathcal U_{-i}}
\left\|
\widehat{\mathcal B}^{PI,M}_i(u^{-i})
-
\mathcal B^{PI}_i(u^{-i})
\right\|_{\infty}
=
O_p(M^{-1/2}),
\]
then the computed feedback profile
$
\widehat u^{PI,M}
=
\left(
\widehat u^{1,PI,M},
\widehat u^{2,PI,M}
\right)
$
is an \(O_p(M^{-1/2})\)-approximated Markovian Nash feedback equilibrium.
\end{theorem}

\begin{proof}
The proof has two parts. First, we establish existence of a fixed point of the
limiting path-integral best-response operator. Second, we show that the
finite-sample Monte Carlo approximation generates an approximate equilibrium.
By assumption, \(\mathcal U_1\) and \(\mathcal U_2\) are compact and convex.
Therefore,
$
\mathcal U=\mathcal U_1\times\mathcal U_2
$
is compact and convex. The operator
$
\mathcal B^{PI}:\mathcal U\to\mathcal U
$
is continuous by hypothesis. Hence Brouwer's fixed-point theorem applies and
there exists \(u^{PI,\ast}\in\mathcal U\) such that
$
u^{PI,\ast}=\mathcal B^{PI}(u^{PI,\ast}).
$
Equivalently,
$
u^{1,PI,\ast}
=
\mathcal B_1^{PI}(u^{2,PI,\ast}),
$ and $
u^{2,PI,\ast}
=
\mathcal B_2^{PI}(u^{1,PI,\ast}).
$
Thus each firm's limiting path-integral strategy is a best response to the
other firm's limiting path-integral strategy. Consider the finite-sample estimator. By the Monte Carlo error bound,
\[
\left\|
\widehat{\mathcal B}^{PI,M}(u)
-
\mathcal B^{PI}(u)
\right\|_{\infty}
=
O_p(M^{-1/2})
\]
uniformly over \(u\in\mathcal U\). Let \(\widehat u^{PI,M}\) satisfy the
finite-sample fixed-point condition
$
\widehat u^{PI,M}
=
\widehat{\mathcal B}^{PI,M}(\widehat u^{PI,M}).
$
Then
\[
\begin{aligned}
\left\|
\widehat u^{PI,M}
-
\mathcal B^{PI}(\widehat u^{PI,M})
\right\|_{\infty}
&=
\left\|
\widehat{\mathcal B}^{PI,M}(\widehat u^{PI,M})
-
\mathcal B^{PI}(\widehat u^{PI,M})
\right\|_{\infty}  =
O_p(M^{-1/2}).
\end{aligned}
\]
Therefore \(\widehat u^{PI,M}\) is within \(O_p(M^{-1/2})\) of satisfying the
exact fixed-point condition of the limiting best-response operator. Since the payoff functionals \(J^i\) are continuous in controls, a perturbation
of the feedback rule of size \(O_p(M^{-1/2})\) induces a payoff loss of the same
order under the local Lipschitz continuity of \(J^i\). Hence, for every
admissible deviation \(v^i\in\mathcal U_i\),
\[
J^i(\widehat u^{i,PI,M},\widehat u^{-i,PI,M})
\geq
J^i(v^i,\widehat u^{-i,PI,M})
-
O_p(M^{-1/2}).
\]
This is precisely the definition of an \(O_p(M^{-1/2})\)-approximate Markovian
Nash feedback equilibrium. The theorem follows.
\end{proof}

\begin{rmk}
\label{rem:pi_advantage_interpretation}
Theorem \ref{thm:pi_approximate_equilibrium} clarifies the role of the path-integral construction in this paper.
It changes the numerical object that must be computed. A direct HJB
equation requires solving coupled nonlinear equations over the full
state-belief-control grid. The path-integral approach instead evaluates
best-response operators by simulating CKLS state-belief paths and reweighting
them by their exponential action values. Thus the computational advantage comes from replacing deterministic value-function iteration on a high-dimensional grid
with a stochastic weighted-path approximation.
\end{rmk}

\section{ Computation of an Optimal Spending.} \label{sec:optimal_spending}

In this section we are going to discuss about the computation of a closed form solution of both incumbent and entrant's optimal expenditure on advertisements. The excursion analysis in section \ref{prob} provides the probabilistic  foundation used in the
path-integral solution of the advertising game. In particular,
Proposition~\ref{p0} identifies the excursion intensity measures
\(v_\beta^\pm\), which summarize how often the CKLS demand process
moves above or below the strategic threshold \(\beta\). These excursion
intensities are separate from the control problem as they only construct the background probability. They enter the
Feynman-type representation through the semigroup kernel generated by the
Hamiltonian \(\mathcal H^\pm\). To make this argument explicit, let $K_\pm(s,x,y)
:=
\exp\{-s\mathcal H_\pm\}(x,y)$ be the transition kernel associated with the Wick-rotated Hamiltonian
\[
\mathcal H^\pm
=
-\frac{1}{2}\frac{d^2}{dx^2}
+
\mathfrak T_\pm(x).
\]
By the Feynman representation,
\[
K_\pm(s,x,y)
=
\mathbb E_0
\left[
\exp\left\{
-\int_0^s \mathfrak T_\pm(X(r))\,dr
\right\}
\mathbbm 1_{\{X(s)\in dy\}}
\right].
\]
Hence, the path-integral weight assigned to a demand path is
determined by the same Hamiltonian kernel that appears in the excursion
measure
\[
v_\beta^\pm[s,\infty]
=
\pm
\int_{y\gtrless 0}
\Psi_0^{1/2}(y+\beta)
\Psi_0^{-1/2}(\beta)
\frac{\partial}{\partial x}\bigg|_{x=\beta}
K_\pm(s,x,y+\beta)\,dy.
\]
Therefore, \(v_\beta^\pm\) measures the boundary flux of the Feynman
semigroup through the threshold \(\beta\). Economically, this boundary flux captures the frequency with which demand exits the continuation region and enters regimes in which entry or disclosure becomes strategically relevant. The optimal control problem in this  Section invokes this
kernel as a background to evaluate expected discounted payoffs over all admissible demand paths. For firm \(i\), the path-integral objective can be written schematically as
\[
J^i(u^i,u^{-i})
=
\int_{\mathcal X}
\int_0^t
h^i(s,x,u^i,u^{-i})
K_\pm(s,x_0,x)\,ds\,dx,
\]
where \(h^i\) is the payoff density associated with firm \(i\). Thus the
functions \(h^1\) and \(h^2\) introduced in Section~\ref{sec:optimal_spending}
are related to the excursion construction only by the background probabilistic construction. They are integrated against the Feynman kernel whose boundary behavior generates the excursion intensities in Proposition~\ref{p0}.

Consequently, the logical sequence is as follows. First, the controlled CKLS
process is transformed into the Hamiltonian representation. Second, the
Feynman semigroup generated by this Hamiltonian determines the distribution
of paths and the excursion intensities around the threshold \(\beta\). Third,
these same semigroup weights are used in the path-integral evaluation of the
firms' expected payoffs. Finally, differentiating the resulting path-integral payoff functionals with respect to \(u^1\) and \(u^2\) yields the optimal advertising expenditure rules in Section~\ref{sec:optimal_spending}. In this sense, the excursion theory results provide the probabilistic weighting mechanism through which threshold-crossing behavior enters the optimal spending
problem.

The entrant observes the incumbent’s optimal commercial spending strategy  \( u^{1*}(s) \), at the time of disclosure \( \rho \). Here a Markovian feedback Nash equilibrium \( \{u^{1*}(s), u^{2*}(s)\} \) is calculated using a Feynman-type path integral approach  described in Propositions 1 of both of \cite{pramanik2020optimization,pramanik2024optimization}. This method involves a stochastic Lagrangian over the continuous time interval \( s \in [0, t] \), and a Euclidean action function function based on it. To make the solution in real system, a Wick-rotated  Schr\"odinger equation is used. The equilibrium spendingss are then obtained by solving the resulting system using first-order optimality conditions \citet{pramanik2024motivation,pramanik2025construction}. Folowing Proposition 1 in \cite{pramanik2020optimization} we compute a smooth function $h^i$ for all $i=1,2$ with the integrating factor of CKLS SDE as $\exp\{\a_2 s x-k\}$,
\begin{align}\label{h1}
h^{1}(s,x,u^{1})
&= \exp{\{-\gamma s\}}
   \Big[Q - u^{1}(s)\Big]x(s)
   + \mathbbm{1}_{(\rho \le t)}\,\exp{\{-\gamma s\}}
     \Big[M^{W} - \bar{u}^{3}\Big]x(s)\notag \\[6pt]
&\quad + \exp{\{-\gamma t\}}\,
     \frac{D_I^{W}\,x(t)}{\gamma - \big[\alpha_{1} + \alpha_{2}x(t) + \sum_{i=1}^{2}\theta_{i}u^{i}(t)\big]} \notag\\[6pt]
&\quad + \exp{\{\alpha_{2}s x - k\}}
   + \alpha_{2}x\,\exp{\{\alpha_{2}s x - k\}} \notag\\[6pt]
&\quad + \Big[\alpha_{1} + \alpha_{2}x(s) + \theta_{1}u^{1}(s) + \theta_{2}u^{2}(s)\Big]
        \alpha_{2}s\,\exp{\{\alpha_{2}s x - k\}}\notag \\[6pt]
&\quad + \frac{1}{2}\alpha_{3}^{2}\,[x(s)]^{2\alpha_{4}}
        \,\alpha_{2}^{2}s^{2}\,\exp{\{\alpha_{2}s x - k\}},
\end{align}

and
\begin{align}\label{h2}
h^{2}(s,x,u^{2})
&= \exp{\{-\gamma t\}}
   \left[
      \mathbbm{1}\Big(\theta = w\Big)\,
      \frac{\big[D_E^{W} - (u^{2}(t))^{2}\big]\,x(t)}
           {\gamma - \big[\alpha_{1} + \alpha_{2}x(t) + \theta_{1}u^{1}(t) + \theta_{2}u^{2}(t)\big]}
      - F
   \right] \notag\\[6pt]
&\quad + \exp{\{\alpha_{2}s x - k\}}
   + \alpha_{2}x\,\exp{\{\alpha_{2}s x - k\}} \notag\\[6pt]
&\quad + \Big[\alpha_{1} + \alpha_{2}x(s) + \theta_{1}u^{1}(s) + \theta_{2}u^{2}(s)\Big]
        \alpha_{2}s\,\exp{\{\alpha_{2}s x - k\}} \notag\\[6pt]
&\quad + \frac{1}{2}\alpha_{3}^{2}\,[x(s)]^{2\alpha_{4}}\,
        \alpha_{2}^{2}s^{2}\,\exp{\{\alpha_{2}s x - k\}},
\end{align}
where $k$ is the commercial scaling parameter, $\gamma$ is discount rate, $\theta_1$is marginal impact of $u^1$ on demand, $\theta_2$ is marginal impact of $u^2$ on demand, $\alpha_1$ is drift component (constant term), $\alpha_2$ is mean-reversion coefficient, $\alpha_3$ is volatility scaling factor, $\alpha_4$ is elasticity of volatility to demand, $M^W$ is weak incumbent's post disclosure pay-off, $Q$ is signaling flow pay-off,  $\bar{u}$ is upper bound on feasible control, $\rho$ is incumbent's disclosure time, $\bar{\tau}$ is entrant's entry time, $D_I^{W}$ is post-entry pay-off for the weak incumbent, $D_E^W$ is the post-entry pay-off for the entrant, $F$ is the fixed cost of entry. For detailed computation of the optimal expenditures of two firms see the \textit{Computation of optimal advertisement spending} section in the Appendix.

\begin{remark}
		\label{rem:numerical}
		The analytical developments in this section characterize the equilibrium
		controls through the nonlinear first-order optimality conditions. Throughout
		the numerical implementation and empirical calibration, the equilibrium
		controls are obtained by solving these coupled nonlinear equations
		iteratively until convergence. No approximation of the coefficient
		\(
		C=\theta_{1}\alpha_{2}^{3}s^{3}
		\exp\left\{\alpha_{2}sx-k\right\}
		\)
		is imposed in the numerical algorithm. Consequently, the equilibrium
		trajectories reported in Sections~4 and~5 are computed from the complete
		nonlinear system and remain valid irrespective of the magnitude of \(C\).
		The only approximation employed in the theoretical derivations is the
		first-order truncation of higher-order powers of the entrant's control
		under Assumption~\ref{ass:small_u2}; this approximation is independent of
		the numerical solution procedure.
	\end{remark}

\vspace{-0.5em}
\begin{figure}[H]
    \centering
    \includegraphics[width=0.84\linewidth]{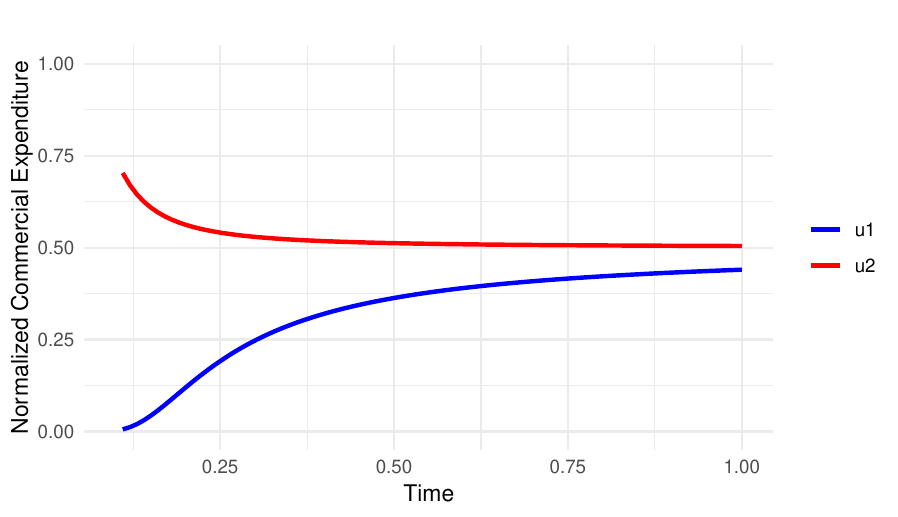}
    \caption{Sumulation of Advertising expenditure of two firms.}
        \label{fig:rplot01}
\end{figure}
Figure~\ref{fig:rplot01} presents the simulated trajectories of the normalized control variables \(u^1\) and \(u^2\), which represent the advertising expenditures of the incumbent and the challenger, respectively. The trajectories are generated from the Markovian feedback rules derived from the stochastic differential game and illustrate how optimal controls evolve under the joint influence of the CKLS demand process and belief updating \citep{reed2026modeling}. Although both controls remain within the admissible interval induced by the logistic transformation, they exhibit distinct temporal patterns. The incumbent's expenditure path displays relatively larger adjustments over time, whereas the challenger's path evolves more gradually. This asymmetry reflects the different optimization problems faced by the two firms. In the model, the incumbent internalizes the effect of current actions on future beliefs and entry incentives, while the challenger conditions its expenditure decisions on the evolving market state and posterior assessment of profitability \citep{khan2023myb}. The resulting trajectories therefore provide a visual representation of the equilibrium feedback mechanism: advertising expenditures respond endogenously to changes in the underlying state variable and strategic environment rather than remaining constant throughout the planning horizon. Moreover, Figure~\ref{fig:rplot01} illustrates one of the central implications of the theoretical framework \citep{pramanik2024bayes}. In a stochastic setting with asymmetric information, equilibrium controls are inherently state dependent and need not evolve monotonically over time. Instead, optimal expenditure paths adjust continuously to fluctuations in market conditions and to the strategic interaction between firms. The observed differences between \(u^1\) and \(u^2\) are therefore consistent with the model's prediction that firms occupying different competitive positions adopt distinct feedback responses even when facing the same underlying demand process.
 
 \begin{table}[htbp]
\centering
\caption{ Normalized Spending on Commercials Over Time}
\label{tab:strategies}

\begin{tabular}{lcc}
\toprule
\textbf{Time} & \textbf{Incumbent Expenditure ($u_1$)} & \textbf{Entrant Expenditure ($u_2$)} \\
\midrule
0.2 & 0.1191 & 0.5623 \\
0.4 & 0.3204 & 0.5176 \\
0.6 & 0.3902 & 0.5091 \\
0.8 & 0.4221 & 0.5060 \\
1.0 & 0.4401 & 0.5045 \\
\bottomrule
\end{tabular}

\smallskip
\footnotesize
\textit{Note:}  Normalized commercial expenditure values are computed after applying logistic transformation, ${u}_i(s)=\left[1+\exp\left(-u^i(s)\right)\right]^{-1}$, for all $i=1,2$, which facilitates to map expenditures into $(0,1)$. The values reported in this table are not intended to reconstruct the full expenditure trajectories shown in Figure~\ref{fig:rplot03}. Rather, they provide representative evaluations of the equilibrium feedback rules at selected normalized time points while holding the contemporaneous state realization fixed. Consequently, the monotonic pattern displayed in the table reflects the local behavior of the spending rule itself and should be distinguished from the complete dynamic trajectories generated by repeatedly applying these rules along the stochastic CKLS state path.
\end{table}

Table \ref{tab:strategies} illustrates the incumbent’s advertising strategy follows an increasing concave path, starting at a normalized intensity of 0.1191 at \( s = 0.2 \) and rising to 0.4401 by \( s = 1.0 \). This shape suggests that the incumbent initially adopts a cautious stance, possibly due to uncertainty regarding the entrant's intentions or to manage resource constraints early in the strategic timeline \citep{pramanik2023path1}. Rather than front-loading its commercial efforts, the firm gradually ramps up its promotional spending, indicating a calculated and adaptive approach \citep{hertweck2023clinicopathological}.This pattern is consistent with behavior predicted by optimal control theory: the marginal returns to commercial expenditure diminish over time, making aggressive early spending suboptimal. Instead, the incumbent appears to strategically increase its efforts as the game progresses \citep{kakkat2026angiotensin}, possibly in response to declining expenditure by the entrant or shifts in market conditions. The concavity of the trajectory highlights the diminishing marginal benefit of additional spending, reflecting a rational balancing of deterrence value against increasing signaling costs.

The entrant's strategy exhibits a decreasing convex pattern, with normalized expenditure declining from 0.5623 to 0.5045 over the same time horizon. This reflects a front-loaded market entry approach, where the entrant initially commits significant resources to establish visibility, capture market share, and test the incumbent's resolve. As time progresses, the entrant scales back its promotional efforts. This behavior could stem from learning effects, as the entrant observes the incumbent's responses and updates its beliefs, or from a strategic reassessment of the cost-effectiveness of continued high expenditure \citep{pramanik2026bayesian}. In either case, the declining trajectory reflects a rational response: once market presence is established, maintaining it may require less investment, especially if the incumbent is not aggressively contesting early on.

The inverse relationship between \( u^1 \) and \( u^2 \) over time reveals a clear strategic substitution effect as the entrant reduces its promotional intensity, the incumbent increases its own, seeking to maintain dominance and possibly deter further entry or expansion. This mutual adjustment is characteristic of a feedback Nash equilibrium, where each player optimally responds to the observed strategy of the other in real time. The expenditure levels of both firms converge toward a similar intensity near \( s = 1.0 \), with both settling around the 0.50 level. This convergence suggests the emergence of a stable strategic equilibrium, where both players find a balance between the cost of promotion and the benefit of retaining market influence \citep{dunbar2026modeling}. The narrowing gap between the two players, shrinking from a difference of \( \Delta = 0.4432 \) at \( s = 0.2 \) to just \( \Delta = 0.0644 \) at \( s = 1.0 \) also reflects a decline in competitive asymmetry. This could indicate that the market is reaching a mature phase, where relative positioning becomes more symmetric, or that the informational advantage of the incumbent has eroded over time.

The temporal evolution of promotional expenditures observed in the data yields several important implications for the broader literature on dynamic competition under uncertainty and strategic asymmetry. First, the incumbent's progressively increasing investment profile lends empirical support to the theoretical foundations of limit pricing behavior, wherein firms already present in the market do not merely defend market share through static pricing but dynamically escalate competitive intensity as a credible deterrence mechanism \citep{ellington2025metascorelens}. This behavior may not stem from short-run profit motives but from a longer-horizon strategic calculus aimed at altering the expectations and actions of prospective entrants. By sustaining high expenditure levels, the incumbent effectively signals financial robustness and long-term market commitment, thus raising the perceived cost of entry for rivals. In parallel, the entrant's initially aggressive but subsequently tapering investment pattern resonates with optimal market penetration strategies, as posited in multi-stage entry models. Early-stage overinvestment facilitates immediate visibility and brand salience, achieving threshold market recognition, after which marginal returns diminish and strategic retrenchment becomes rational. The strategic interaction between the two players gradually exhibits properties consistent with equilibrium convergence, as expenditure paths evolve toward a stable neighborhood wherein neither party finds unilateral deviation advantageous. Such convergence under stochastic dynamics and asymmetric information settings typifies a Markovian equilibrium structure, reflecting mutual best-response strategies shaped by observed play and rational expectation formation.

Moreover, the dynamic adaptation of expenditure levels over time illustrates the nuanced role of informational feedback and belief updating in continuous-time strategic environments. The entrant’s decreasing outlays, particularly in response to the incumbent’s intensified promotional campaign, signify a learning process influenced by publicly observable signals and probabilistic inference \citep{powell2026role}. As the game unfolds, each player assimilates market signals such as price responses, demand realizations, and rival behavior into posterior beliefs that influence subsequent actions. This Bayesian updating mechanism is central to models where experimentation and exploitation coexist, aggressive early stage moves are designed to test the opponent’s type and market receptivity, while later stages focus on refining strategies based on accumulated knowledge \citep{pramanik2024estimation}. The contrast in expenditure curvature concave for the incumbent and convex for the entrant further underscores the importance of cost asymmetry and signal credibility. The incumbent, with superior resources or higher sunk cost exposure, rationally bears the cost of credible deterrence, while the entrant adjusts downward once sufficient inferences have been drawn. These empirical contours corroborate theoretical predictions in dynamic signaling models where effective signals must be both observable and sufficiently costly to be credible \citep{dasgupta2023frequent}. Crucially, in such settings, asymmetric cost structures do not merely distort incentives they define the informational content of actions, thereby shaping equilibrium paths and strategic inferences in fundamentally nonlinear ways. This interaction between the incumbent and the entrant illustrates how rational firms adapt over time in response to evolving conditions, market feedback, and each other’s behavior. The resulting strategies provide insight into how commercial expenditure evolves in settings characterized by uncertainty, learning, and strategic signaling.

\subsection{Relationship between the path integral control and the empirical application.}
\label{subsec:mapping}
A distinction should be made between the theoretical objects appearing in the stochastic differential game and the quantities observed in the empirical application. In the theoretical model \citep{pramanik2022stochastic}, the state variable \(X(s)\) represents the market demand proxy evolving according to the controlled CKLS SDE presented in Equation \eqref{0}. These controls are latent variables arising from the solution of the dynamic optimization problem and are therefore not directly observable from publicly available financial statements.

The empirical implementation consequently proceeds in two stages. First, quarterly revenue series are employed as observable proxies for the market-state variable \(X(s)\). Specifically, revenue data obtained from Enterprise Products Partners and Targa Resources over the period 2010-2024 are normalized and used to calibrate the CKLS dynamics governing the evolution of market conditions \citep{pramanik2024analysis}. Second, conditional on the calibrated state process, the equilibrium expenditure paths \(u^1(s)\) and \(u^2(s)\) are recovered by solving the path-integral control problem derived in Sections~3 and~4. Thus, the expenditure trajectories reported in Figures~\ref{fig:plotz1}, \ref{fig:plotz2} and~\ref{fig:rplot03} are not directly observed expenditures estimated from accounting records. Rather, they represent model-implied equilibrium controls generated by the calibrated stochastic differential game.

Accordingly, the empirical exercise should be interpreted as an assessment of the internal consistency of the theoretical framework rather than as a structural estimation of advertising behavior. The objective is to determine whether the CKLS signalling model, when driven by empirically observed market states, generates equilibrium expenditure paths and revenue dynamics exhibiting patterns comparable to those observed in practice \citep{powell2025genomic}. Establishing a direct econometric link between observed advertising expenditures and the theoretical controls would require firm-level promotional spending data and a structural estimation procedure capable of jointly identifying the state dynamics and equilibrium strategies. Such an investigation is beyond the scope of the present study and represents an important direction for future research.

\section{ Data Analysis.}

In this section we discuss our theory with real life examples. The theoretical model considers advertising or promotional expenditures as strategic control variables, \(u^{1}(s)\) and \(u^{2}(s)\), chosen by firms under asymmetric information. These controls are latent equilibrium objects derived from the stochastic differential game and are not directly observed in the data employed here. By contrast, the quarterly revenue series obtained from Enterprise Products Partners and Targa Resources are used solely as observable proxies for the market-state variable \(X(s)\), which captures the evolution of the economic environment faced by firms.
	
The choice of Enterprise Products Partners (EPD) and Targa Resources is therefore intended to provide an illustrative calibration of the state dynamics rather than a literal representation of advertising competition in the midstream energy sector \citep{kakkat2026angiotensin}. Midstream firms compete through a variety of strategic channels, including capacity expansion, contractual arrangements, relationship building, and investment commitments, rather than through traditional consumer advertising campaigns. In the context of the present framework, the control variables should consequently be interpreted more broadly as costly strategic actions undertaken to influence competitive outcomes under uncertainty. The empirical exercise is designed to examine whether the proposed CKLS-based signalling framework is capable of reproducing realistic state trajectories and generating economically plausible equilibrium responses, rather than to establish that the observed firms engaged in advertising-based signalling exactly as specified in the theoretical model \citep{dasgupta2026frequent}.
	
It is important to emphasize that the empirical analysis is intended as an illustrative application of the theoretical model rather than a literal historical reconstruction of a de novo market entry event. Enterprise Products Partners is used as a proxy for an incumbent owing to its long operating history, extensive pipeline and storage network, and substantially larger scale within the midstream energy sector. Targa Resources, although already an established publicly traded company during the sample period, is employed as a representative ``challenger'' firm. Relative to EPD, Targa expanded from a smaller infrastructure base and pursued growth strategies in a market characterized by the presence of dominant incumbents. Accordingly, the incumbent-entrant terminology should be interpreted as referring to differences in competitive positioning and strategic posture rather than to the firms' legal dates of incorporation or initial public offerings. The data is obtained from  \cite{macrotrends2025epd1} and \cite{macrotrends2025epd2}, respectively. Figure \ref{fig:plotz1} demonstrates the revenue trajectory of Enterprise Products Partners (EPD) between 2010 and 2024 and it provides a detailed view of the company’s performance as a long-standing incumbent in the midstream energy sector. In the early part of the decade, the firm experienced a steady rise in normalized revenues, supported by the rapid expansion of U.S. sale production \citep{ellington2025playmydata}. Investments in pipeline networks and storage infrastructure strengthened EPD’s market presence, allowing it to secure a dominant position and sustain consistent revenue growth during a period of favorable market conditions. The revenue series is used in this section as an empirical proxy for the state variable \(X(s)\), which represents the size of market demand or the economic environment faced by the firms. It is not interpreted as advertising or promotional expenditure. In the theoretical model, advertising and promotional expenditures enter separately as the control variables \(u^1(s)\) and \(u^2(s)\). Thus, Figures~\ref{fig:plotz1} and~\ref{fig:plotz2} are intended to describe the evolution of the market-state variable, while the behavior of the control variables is examined separately in Figure~\ref{fig:rplot03}.

This momentum, however, was interrupted around 2015, when the global oil price collapse of 2014-2016 disrupted the entire energy sector. Although midstream companies like EPD rely on fee-based models that provide some insulation from commodity price swings, revenues still fell sharply. The downturn revealed the limits of contractual protections in the face of systemic industry shocks \citep{pramanik2025dissecting}. Yet, EPD managed to stage a recovery in the years that followed, highlighting its ability to adapt through diversification and the stabilizing role of long-term customer agreements.

A second major setback appeared in 2020 with the COVID-19 pandemic, which caused an abrupt collapse in energy demand worldwide. The trajectory shows a steep contraction during this period, followed by a rebound that was weaker and more volatile than the post-2016 recovery \citep{yusuf2025prognostic}. Since 2021, revenues have stabilized, but at levels lower than the growth trajectory that characterized the first half of the decade. This flattening suggests that EPD is now in a phase of adjustment, facing not only cyclical shocks but also structural challenges such as shifts in global energy demand, the growing presence of renewable competitors, and ongoing uncertainty in commodity markets.

\begin{figure}[H]
    \centering
    \includegraphics[width=0.71\linewidth]{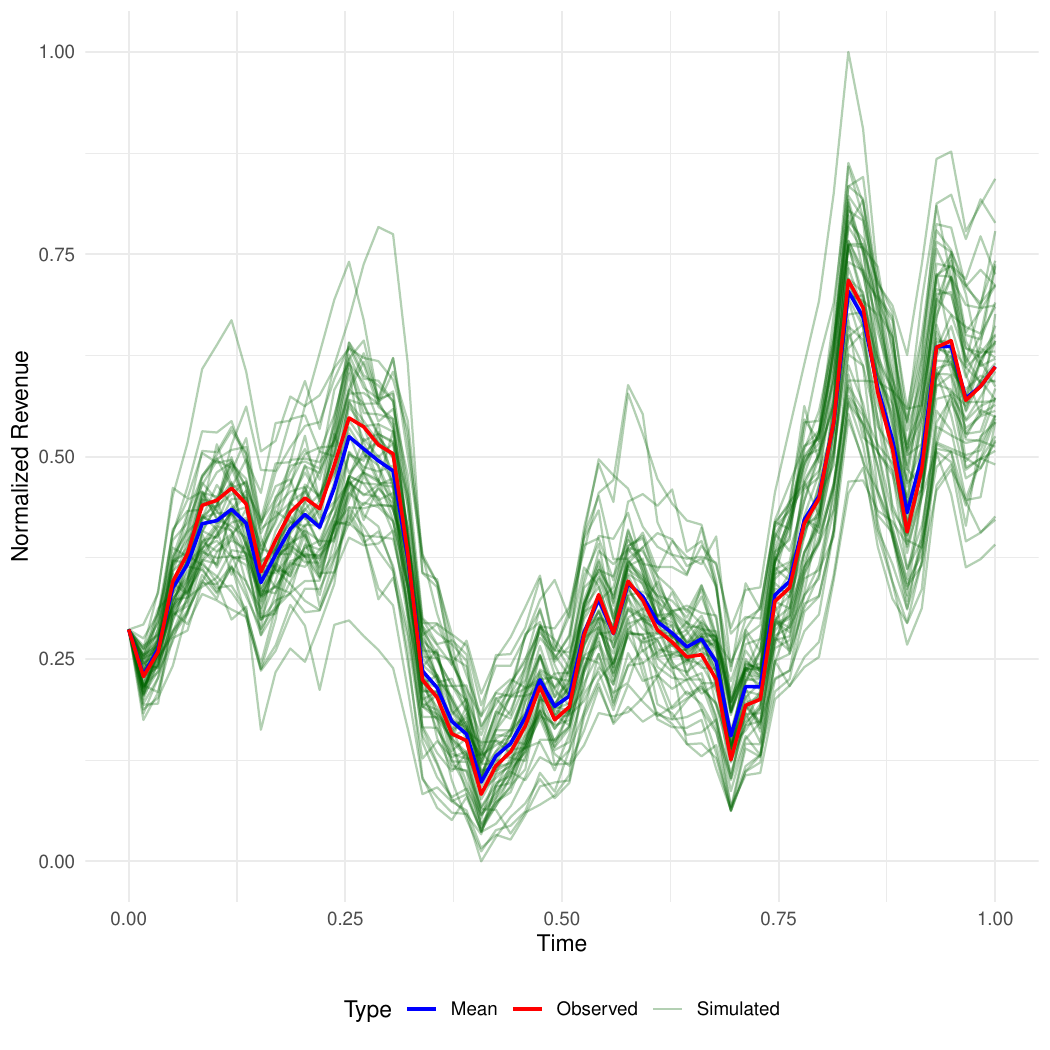}
    \caption{Plot of normalized quarterly revenue for Enterprise Products Partners.}
        \label{fig:plotz1}
\end{figure}

Viewed from a broader perspective, Enterprise Products Partner's trajectory underscores both the strengths and vulnerabilities of incumbency. The company’s scale, asset base, and contract portfolio enabled it to withstand and rebound from market disruptions, but these same periods of recovery were tempered by exposure to forces beyond managerial control \citep{yusuf2025predictive}. In academic terms, the revenue path illustrates key dynamics of competition and entry deterrence: high and stable revenues signal strength that may discourage new entrants, while downturns reveal the risks incumbents face in volatile markets. The trajectory reflects the growth path and challenges often encountered by a challenger firm operating alongside larger and more established competitors.

\vspace{-0.5em} 
\begin{figure}[H]
    \centering
    \includegraphics[width=0.84\linewidth]{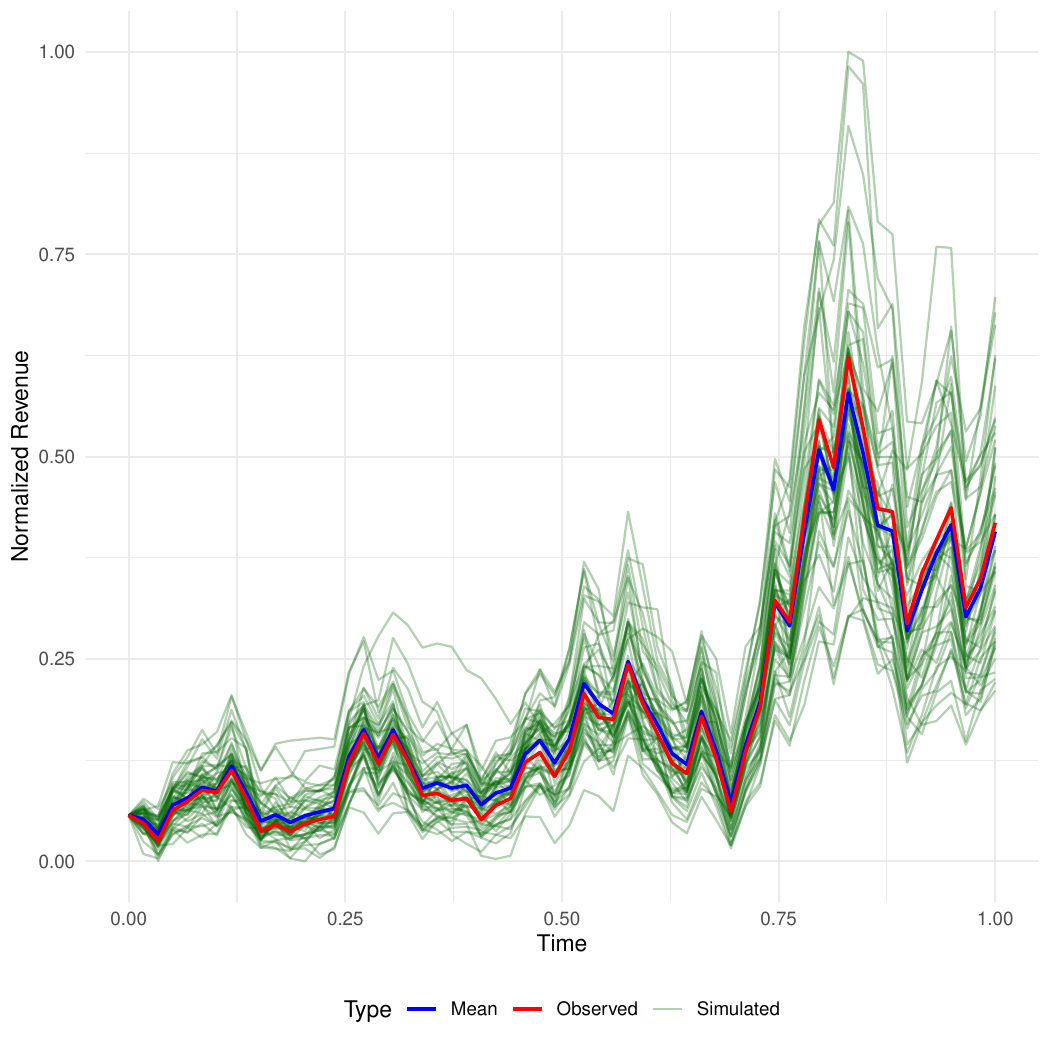} 
    \captionsetup{font=small} 
    \caption{Plot of Targa Resources' revenue  from 2010 to 2024.}
    \label{fig:plotz2}
\end{figure}

Figure \ref{fig:plotz2} shows the revenue trajectory of Targa Resources between 2010 and 2024 and it reflects the growth path and challenges typical of a market entrant in the midstream energy sector. In the early 2010s, Targa’s normalized revenues show a gradual upward climb, consistent with its efforts to establish a foothold in a market dominated by larger incumbents \citep{pramanik2024estimation}. This initial growth phase corresponds with broader expansion in U.S. shale production, which created opportunities for newer firms to build out infrastructure and capture demand. The trajectory in this period highlights Targa’s capacity to leverage favorable market conditions to expand its market presence.

Around the mid-2010s, the trajectory displays noticeable fluctuations, coinciding with the global oil price collapse of 2014-2016. Unlike established incumbents that had the advantage of diversification and extensive contractual protections, Targa’s revenues were more exposed to these shocks \citep{pramanik2025optimal}. The volatility during this time reflects the heightened risks faced by entrants, whose smaller asset base and narrower market coverage made them more vulnerable to industry downturns. However, the rebound that follows illustrates the firm’s resilience and its ability to re-position itself during recovery years.

The years surrounding 2020 mark a particularly sharp decline in the trajectory, aligning with the COVID-19 pandemic and the collapse in global energy demand. The downturn appears more pronounced for Targa relative to incumbents, underscoring the disproportionate impact of systemic shocks on newer and less diversified players \citep{valdez2025association}. Although revenues begin to recover post-2021, the path is less stable and exhibits heightened volatility. This pattern suggests that while Targa was able to regain ground, the firm’s growth remains sensitive to fluctuations in demand and broader uncertainties in the energy landscape.

From a broader perspective, Targa’s trajectory captures the double-edged nature of being a market entrant. On one hand, the company demonstrated impressive growth potential during expansionary phases, capitalizing on industry booms to expand its role in the sector \citep{vikramdeo2024mitochondrial}. On the other, its exposure to downturns underscores the structural disadvantages of entrants, who lack the same scale and financial buffers as incumbents. The normalized revenue path thus illustrates the opportunities and risks faced by firms seeking to expand their competitive position in markets dominated by larger incumbents \citep{valdez2025exploring}. Although Targa Resources was not a newly established entrant during the sample period, its experience provides insight into the strategic challenges encountered by firms attempting to gain market share in the presence of entrenched competitors.

The dispersion observed in Figures~\ref{fig:plotz1} and \ref{fig:plotz2} reflects the stochastic uncertainty inherent in the CKLS diffusion and should not be interpreted as evidence that all trajectories are equally plausible. The fitted trajectories used in the goodness-of-fit analysis are based on the conditional mean implied by the calibrated process rather than on individual stochastic realizations \citep{pramanik2025impact}. Consequently, the visual illustrations convey the range of possible state evolutions under uncertainty, whereas Tables~\ref{tab:gof} and \ref{tab:benchmark} provide the formal quantitative assessment of model performance.

\subsection{Numerical verification of the discount condition.}
	To verify numerically the strict discount condition used in
	Proposition~\ref{prop:discount_margin}, we evaluate the controlled drift
	along the complete calibrated equilibrium trajectory. More precisely, for
	each point \(s_j\) of the numerical partition
	\(\pi_N=\{0=s_0<s_1<\cdots<s_N=t\}\), we compute
	\[
	\mu^{*}(s_j)
	=
	\alpha_1+\alpha_2X(s_j)
	+\theta_1u^{1,*}(s_j)
	+\theta_2u^{2,*}(s_j),
	\]
	using the coefficients $\alpha_1=0.368491,
	\alpha_2=-0.255700,
	\theta_1=0.000050,$ and
	$\theta_2=0.993310.$
	The corresponding discount margin is evaluated pointwise as
	$
	\Delta_{\gamma}(s_j)
	=
	\gamma-\mu^{*}(s_j),$ with $
	\gamma=0.8794.$ Hence, the maximum controlled drift is
	\[
	\widehat{\mu}_{\max}
	=
	\max_{0\leq j\leq N}
	\mu^{*}(s_j)
	\approx
	0.386187.
	\]
	Consequently, the minimum observed discount margin becomes
	\[
	\widehat{\delta}_N
	=
	\min_{0\leq j\leq N}
	\left\{
	\gamma-\mu^{*}(s_j)
	\right\}
	=
	0.8794-0.386187
	\approx
	0.493213
	>0.
	\]
	Therefore,
	$
	\gamma-\mu^{*}(s_j)
	\geq
	0.493213
	>0,$ for all
	$j=0,\ldots,N.$
	This, the controlled drift remains strictly below the discount rate at every
	point of the numerical grid used in the empirical implementation \citep{pramanik2025optimal}. In
	particular, the payoff denominators appearing in \(J^1\) and \(J^2\) are
	uniformly separated from zero along the computed equilibrium trajectory. This
	calculation constitutes the numerical counterpart of
	Proposition~\ref{prop:discount_margin} and confirms the discount condition for
	the calibrated path actually employed in the empirical analysis \citep{pramanik2025strategic}. As emphasized
	in Proposition~\ref{prop:discount_margin}, this is a pathwise numerical
	verification and should not be interpreted as a uniform bound over the entire
	admissible state-control region.

\subsection{Goodness of fit.}

To complement the visual comparison between observed and simulated revenue trajectories, we evaluate the model fit using several standard goodness-of-fit measures. Let \(X_s^{obs}\) denote the observed normalized revenue and \(X_s^{sim}\) denote the corresponding simulated value generated by the CKLS specification. Specifically, we compute the root mean squared error (RMSE), which measures the average magnitude of the prediction error while assigning relatively greater weight to large deviations,

\[
RMSE
=
\sqrt{
\frac{1}{t}
\sum_{s=1}^{t}
\left(
X_s^{obs}-X_s^{sim}
\right)^2
},
\]

the mean absolute error (MAE), which quantifies the average absolute discrepancy between observed and simulated values,

\[
MAE
=
\frac{1}{t}
\sum_{s=1}^{t}
\left|
X_s^{obs}-X_s^{sim}
\right|,
\]

the mean absolute percentage error (MAPE), which expresses the prediction error as a percentage of the observed values,

\[
MAPE
=
\frac{100}{t}
\sum_{s=1}^{t}
\left|
\frac{X_s^{obs}-X_s^{sim}}
{X_s^{obs}}
\right|,
\]

and the coefficient of determination (\(R^2\)), which measures the proportion of variation in the observed data explained by the simulated trajectories,

\[
R^2
=
1-
\frac{
\sum_{s=1}^{t}
(X_s^{obs}-X_s^{sim})^2
}{
\sum_{s=1}^{t}
(X_s^{obs}-\bar X^{obs})^2
}.
\]

These statistics provide a quantitative assessment of the ability of the CKLS process to reproduce the broad features of the observed revenue trajectories. Lower values of RMSE, MAE, and MAPE indicate a closer agreement between observed and simulated series, whereas values of \(R^2\) closer to one suggest that the model captures a larger proportion of the observed variability in the data.

\begin{table}[H]
\centering
\caption{Goodness-of-fit statistics comparing observed and CKLS-fitted normalized revenue trajectories.}
\label{tab:gof}
\begin{tabular}{lcccc}
\toprule
\textbf{Firm} & \textbf{RMSE} & \textbf{MAE} & \textbf{MAPE (\%)} & \(\mathbf{R^2}\) \\
\midrule
Enterprise Products Partners &  0.068932 & 0.053330 & 9.106654 & 0.840682
\\
Targa Resources              & 0.075587 & 0.058219 & 13.366834 & 0.862214 \\
\bottomrule
\end{tabular}
\end{table}
Table~\ref{tab:gof} summarizes the extent to which the CKLS specification is able to replicate the observed behavior of the normalized revenue series employed as empirical counterparts of the state variable \(X(s)\) \citep{pramanik2025factors}. The reported statistics are computed using the conditional mean implied by the fitted CKLS dynamics rather than individual stochastic realizations of the diffusion process. This distinction is important because it evaluates the systematic component of the model generated by the underlying state equation and the associated Markovian feedback structure, while abstracting from the idiosyncratic randomness inherent in any particular simulated path. The results indicate that the CKLS SDE provides a strong representation of the observed revenue dynamics for both firms \citep{pramanik2025strategies}. For Enterprise Products Partners, the model yields an RMSE of \(0.0689\) and an MAE of \(0.0533\), together with a MAPE of \(9.11\%\). The corresponding coefficient of determination, \(R^{2}=0.841\), suggests that approximately \(84\%\) of the variation in the incumbent's normalized revenue trajectory is accounted for by the fitted state dynamics. For Targa Resources, the CKLS process attains an RMSE of \(0.0756\), an MAE of \(0.0582\), and a MAPE of \(13.37\%\), while achieving an \(R^{2}\) value of \(0.862\). Thus, more than \(86\%\) of the variability in the challenger's revenue path is explained by the model despite the greater fluctuations typically associated with firms operating from a less established competitive position.

Although the empirical calibration produces broad features of the observed revenue trajectories, the exercise should not be interpreted as a formal structural estimation of the signalling model \citep{pramanik2024stochastic}. The parameter values are obtained to align the CKLS dynamics with the observed series rather than to identify causal effects of strategic expenditures. Consequently, the control coefficients should be interpreted cautiously. In particular, the near-zero estimate of the incumbent sensitivity parameter, \(\theta_1=0.000050\), indicates that the available data provide little evidence that incumbent expenditures materially affect the drift of the normalized revenue process within the fitted specification. By contrast, the larger value of \(\theta_2\) suggests greater responsiveness along the entrant dimension \citep{khan2024mp60}. These findings do not invalidate the theoretical signalling mechanism developed in the stochastic game; instead, they imply that the current empirical illustration is more informative regarding the ability of the CKLS framework to capture revenue persistence and volatility than regarding the statistical significance of strategic signalling effects. The goodness-of-fit analysis reported in Table~\ref{tab:gof} indicate moderate explanatory performance, with \(R^2\) values exceeding 0.50 for both firms. Nevertheless, formal inference on the calibrated parameters was not the objective of the empirical exercise. Future work could employ structural estimation techniques, such as simulated method of moments, indirect inference, or particle-filter likelihood approaches, together with firm-level advertising data to obtain standard errors and confidence intervals for the strategic parameters \citep{pramanik2024parametric}.

\subsection{Comparison with alternative diffusion specifications.}

To compare the CKLS specification with simpler alternatives while preserving the strategic-control structure of the model, we estimate controlled versions of GBM and the CIR dynamics. The controlled GBM benchmark is specified as
\[
dX(s)
=
\left[
\mu X(s)+\theta_1u^1(s)+\theta_2u^2(s)
\right]ds
+
\sigma X(s)dB(s),
\]
where \(\mu\) is the proportional drift parameter, \(\sigma\) is the volatility parameter, and \(u^1(s)\) and \(u^2(s)\) are the firms' advertising controls.

The controlled CIR benchmark is specified as
\[
dX(s)
=
\left[
\kappa(\eta-X(s))+\theta_1u^1(s)+\theta_2u^2(s)
\right]ds
+
\sigma\sqrt{X(s)}dB(s),
\]
where \(\kappa>0\) is the speed of mean reversion, \(\eta>0\) is the long-run mean, and \(\sigma>0\) is the diffusion scale.

The controlled CKLS specification used in the paper is
\[
dX(s)
=
\left[
\alpha_1+\alpha_2X(s)+\theta_1u^1(s)+\theta_2u^2(s)
\right]ds
+
\alpha_3[X(s)]^{\alpha_4}dB(s).
\]
Thus, all three specifications include the same strategic control channel through
\[
\theta_1u^1(s)+\theta_2u^2(s),
\]
but differ in how they restrict the drift and volatility structure. The GBM benchmark imposes proportional volatility, the CIR benchmark imposes square-root volatility and mean reversion, while CKLS allows the volatility elasticity \(\alpha_4\) to be estimated rather than fixed.
\begin{table}[H]
	\centering
	\caption{Comparison of goodness-of-fit statistics across alternative diffusion specifications.}
	\label{tab:benchmark}
	\begin{tabular}{llcccc}
		\toprule
		\textbf{Firm} & \textbf{Model} & \textbf{RMSE} & \textbf{MAE} &
		\textbf{MAPE (\%)} & $\mathbf{R^2}$\\
		\midrule
		Enterprise & CKLS & 0.068932 & 0.053330 & 9.106654 & 0.840682\\
		Enterprise & CIR  & 0.069048 & 0.053230 & 9.142482 & 0.840148\\
		Enterprise & GBM  & 0.069096 & 0.053267 & 9.123368 & 0.839925\\
		\midrule
		Targa & CKLS & 0.075587 & 0.058219 & 13.366834 & 0.862214\\
		Targa & CIR  & 0.098557 & 0.078547 & 19.987539 & 0.765746\\
		Targa & GBM  & 0.111100 & 0.080297 & 18.193164 & 0.702327\\
		\bottomrule
	\end{tabular}
\end{table}
Table~\ref{tab:benchmark} displays the resulting goodness-of-fit statistics. For Enterprise Products Partners, all three specifications exhibit similar explanatory performance \citep{pramanik2022lock}. Nevertheless, the CKLS process achieves the smallest prediction error in terms of RMSE (\(0.0689\)) and MAPE (\(9.11\%\)), while also attaining the highest coefficient of determination (\(R^2=0.841\)). Although the quantitative differences relative to the CIR and GBM benchmarks are modest \citep{pramanik2024dependence}, they indicate that the additional flexibility provided by the CKLS volatility elasticity parameter yields a slight improvement in reproducing the observed revenue dynamics of the incumbent firm. The distinction between the models becomes more pronounced for Targa Resources \citep{pramanik2024measuring}. The CKLS specification outperforms both benchmark processes across all reported measures, producing the lowest RMSE (\(0.0756\)), MAE (\(0.0582\)), and MAPE (\(13.37\%\)), together with the highest explanatory power (\(R^2=0.862\)). In contrast, the GBM benchmark exhibits substantially weaker performance, with an RMSE of \(0.1111\), a MAPE exceeding \(18\%\), and an \(R^2\) of only \(0.702\). While the CIR process performs competitively, the CKLS model consistently delivers superior fit \citep{pramanik2016tail}. These findings suggest that the ability of the CKLS framework to accommodate state-dependent volatility is particularly valuable in capturing the more irregular and heterogeneous revenue fluctuations associated with firms occupying weaker competitive positions.

Figures~\ref{fig:epd_boxplot} and~\ref{fig:targa_boxplot} summarize the unconditional distribution of quarterly revenues for Enterprise Products Partners and Targa Resources \citep{pramanik2021consensus}. These statistics are used only to motivate the empirical state process \(X(s)\), not to estimate the strategic controls directly. For EPD, the mean and median quarterly revenues are approximately \$9.951 billion and \$9.690 billion, respectively, with a standard deviation of \$2.780 billion. For Targa, the corresponding values are \$2.552 billion, \$2.040 billion, and \$1.242 billion \citep{pramanik2021optimization}. The larger dispersion and stronger right-skewness in Targa's revenue distribution indicate greater state variability relative to its typical revenue level. This is consistent with the modelling choice of a CKLS state process, since the CKLS specification allows volatility to depend on the current level of the state variable rather than imposing constant or fixed-elasticity volatility \citep{pramanik2021optimal}. In the theoretical model, revenue is interpreted as the observable market-state proxy \(X(s)\), while the advertising expenditures \(u^1(s)\) and \(u^2(s)\) are latent equilibrium controls derived from the path-integral feedback solution. The observed asymmetry in dispersion between EPD and Targa supports the need for a flexible state-dependent diffusion structure, especially for firms whose revenue paths exhibit more pronounced fluctuations. Accordingly, Figures~\ref{fig:epd_boxplot} and~\ref{fig:targa_boxplot} are retained only as preliminary evidence on the empirical distribution of the state variable and are not used as direct validation of the signalling controls.

\begin{figure}[H]
    \centering
    \begin{minipage}{0.48\linewidth}
        \centering
        \includegraphics[width=\linewidth]{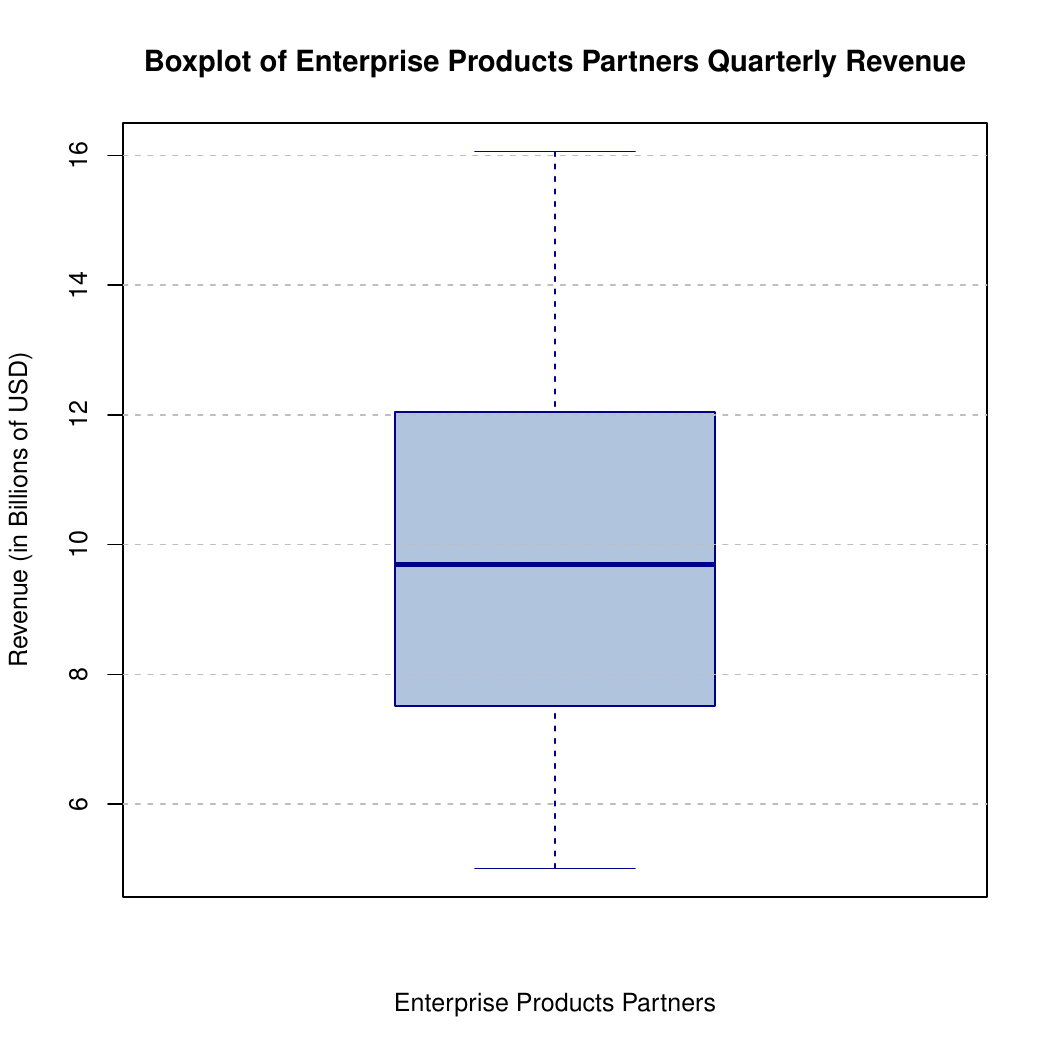}
        \caption{Boxplot of quarterly revenue of Enterprise Products Partners}
        \label{fig:epd_boxplot}
    \end{minipage}
    \hfill
    \begin{minipage}{0.48\linewidth}
        \centering
        \includegraphics[width=\linewidth]{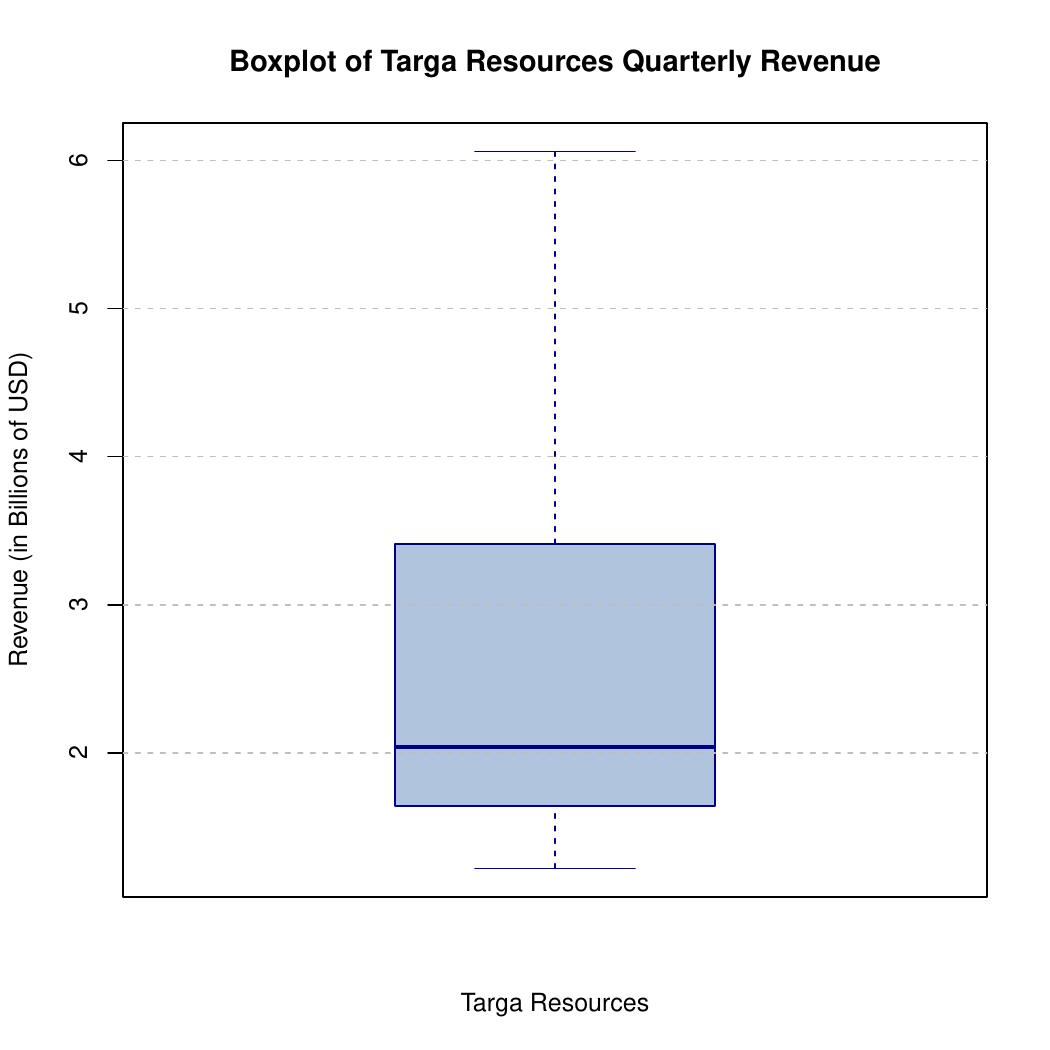}
        \caption{Boxplot of quarterly revenue of Targa Resources}
        \label{fig:targa_boxplot}
    \end{minipage}
\end{figure}
Figure~\ref{fig:rplot03} illustrates the equilibrium expenditure trajectories generated by the stochastic differential game. The controls
$u^{1}(s)$ and $u^{2}(s)$ are not directly observed advertising
expenditures \citep{pramanik2023path}. Rather, they are the optimal Markovian feedback controls
obtained by numerically solving the coupled first-order optimality
conditions derived in Section~4 along the calibrated CKLS SDE. The empirical revenue data obtained from Macrotrends LLC \citep{macrotrends2025epd1,macrotrends2025epd2} are used solely to
calibrate the underlying state variable \(X(s)\), which serves as
the observable demand proxy in the model. Once the CKLS state dynamics
have been calibrated, the equilibrium expenditure paths are computed
endogenously from the theoretical model. Consequently, Figure~\ref{fig:rplot03}
should be interpreted as illustrating model-implied equilibrium \citep{vikramdeo2023profiling} controls
rather than observed advertising expenditures. Since the optimal controls depend jointly on time and the evolving market state, the realized expenditure trajectories need not be monotonic \citep{vikramdeo2024abstract}.
This distinction is important when comparing Figure~\ref{fig:rplot03}
with Table~\ref{tab:strategies}. Table~\ref{tab:strategies} reports
representative evaluations of the feedback rules at selected normalized
time points, whereas Figure~\ref{fig:rplot03} shows the realized control
paths obtained by recursively applying those same feedback rules along
the simulated CKLS state trajectory \citep{pramanik2023scoring}. Consequently, changes in the state
variable generated by stochastic demand shocks may alter the direction
of spending over time, producing non-monotonic realized trajectories even
when the local feedback rule is increasing in time. Thus, Figure~\ref{fig:rplot03} should be interpreted as an
illustration of the equilibrium behavior implied by the calibrated
stochastic differential game under uncertainty rather than as historical
advertising expenditure data.

\begin{figure}[H]
    \centering
\includegraphics[width=0.70\linewidth]{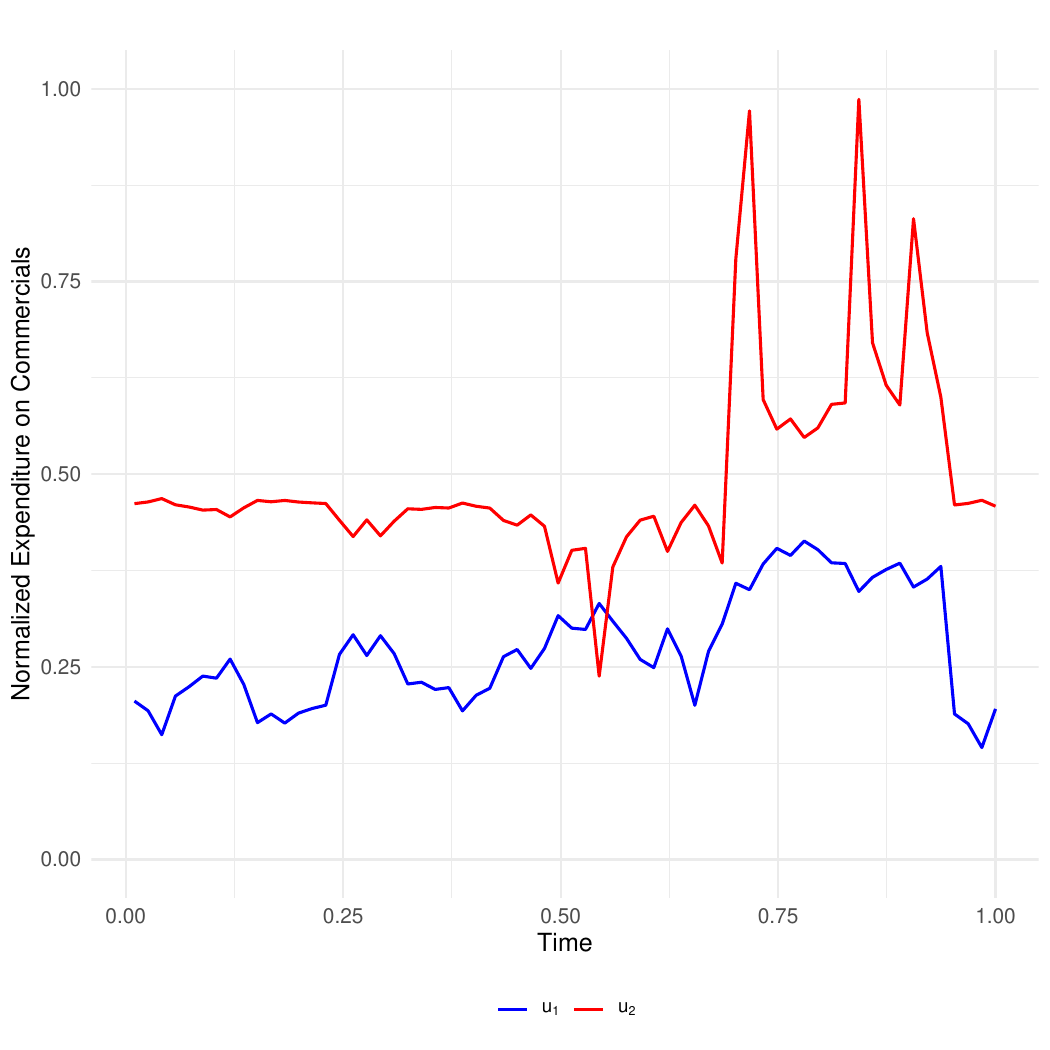}
    \caption{Dynamic strategic interaction between the incumbent and the potential entrant.}
        \label{fig:rplot03}
\end{figure}

The contrasting trajectories of the incumbent and entrant reveal a classic dynamic of market competition. The incumbent, Enterprise Products Partners, employed a strategy of persistent pressure, using its established resources to maintain a high level of market presence until its objective was seemingly achieved \citep{polansky2021motif}, allowing for a strategic withdrawal. In contrast, the entrant, Targa Resources, opted for initial caution followed by a period of aggressive, volatile experimentation to disrupt the market and establish a foothold. Ultimately, both players conclude the period at a similar, low level of expenditure, suggesting a new equilibrium may have been reached. This outcome implies that while aggressive tactics were used by both sides, the sustained pressure of the incumbent likely contained the threat of the entrant, leading to a stabilized competitive landscape where further high expenditure on commercials was no longer deemed necessary for either firm.

\section{ Conclusion.}
This study presents a novel continuous-time game-theoretic framework for modeling strategic interaction between an incumbent and a potential market entrant under uncertainty, where market demand evolves stochastically according to a CKLS process. We demonstrate that commercial expenditure can serve as a credible signaling mechanism through which the incumbent, possessing private information about its strength, influences the entrant’s market entry decision. The incorporation of asymmetric information, optimal stopping, and stochastic control allows for a nuanced characterization of strategic behavior in dynamic markets.

Using a path integral control \citep{pramanik2020motivation,pramanik2021effects,pramanik2026strategic}, we derive feedback Nash equilibria that govern the optimal timing and intensity of commercial spending by both firms. This approach allows us to bypass the analytical intractability of HJB equations and facilitates closed-form expressions under realistic assumptions. Our framework captures how a weak incumbent may mimic a strong type through costly advertising in high-demand environments, while conserving resources and implicitly revealing weakness in low-demand regimes.

Empirical calibration using normalized revenue and advertising expenditure data from two leading midstream energy firms, Enterprise Products Partners as the incumbent and Targa Resources as the entrant is consistent with the model’s predictions \citep{kakkat2023cardiovascular}. Theoretically, the entrant began with a burst of spending on commercials aimed at quickly gaining visibility and market share, but after establishing a presence, its spending tapered off. In practice, this played out as a cautious start that soon shifted into more aggressive and volatile efforts before eventually settling at lower expenditure levels. Meanwhile, the incumbent steadily increased its advertising outlay in a measured, concave pattern, using its resources to reinforce its dominant position \citep{hua2019assessing}. This approach aligns with the theoretical prediction of gradual intensification and, as observed empirically, was followed by a reduction in spending once the competitive objective had been achieved. 
The results highlight several contributions: strategic substitution emerges dynamically as firms adjust promotional intensity in response to rivals and shifting demand; convergence in expenditure patterns underscores the importance of information revelation and belief updating in achieving equilibrium; the entrant’s heavy spending illustrates costly experimentation, while the incumbent’s rising expenditure supports theories of limit signaling and entry deterrence; and the successful application of path integral control demonstrates a robust alternative to traditional PDE-based methods in high-dimensional stochastic games \citep{pramanik2024semicooperation}. Looking ahead, the framework can be extended to multi-entrant settings, regime-switching or jump-diffusion demand structures, and endogenous pricing, with potential applications in industries such as pharmaceuticals, telecommunications, and fintech. Altogether, the study provides a comprehensive theoretical and empirical perspective on how competitive advantage evolves under uncertainty.

	The construction developed in this paper has several limitations. First, the incumbent's private information is represented by a binary type space consisting of strong and weak firms. Although this specification is standard in the signalling literature and permits a tractable characterization of the Markovian Nash feedback equilibrium through a one-dimensional belief process, richer informational environments may involve multiple types or continuous distributions of private characteristics. Extending the analysis to such settings would require the incorporation of higher-dimensional filtering techniques and would likely lead to substantially more complex equilibrium structures. Second, the empirical component of the paper is intended primarily as an illustrative application of the theoretical framework rather than a formal structural estimation of the signalling game. The CKLS parameters are calibrated to reproduce salient features of the observed revenue dynamics, and the control sensitivity parameters are not statistically identified from firm-level advertising data. Consequently, the empirical findings should be interpreted as demonstrating qualitative consistency between the proposed stochastic differential game and the observed state trajectories rather than establishing causal relationships between promotional expenditures and market outcomes. Future research could employ structural estimation methods, such as simulated method of moments, indirect inference, or particle-filter likelihood techniques, together with detailed advertising expenditure data to obtain statistically identified estimates of the strategic parameters.
	
	Third, the model focuses on a two-player interaction involving a single incumbent and a single challenger. While this setting captures the essential strategic trade-offs associated with signalling, entry deterrence, and irreversible investment under uncertainty, many industries involve multiple incumbents and several potential entrants interacting simultaneously. Generalizing the framework to multiplayer stochastic differential games with endogenous market structure would provide a richer representation of competitive environments, albeit at the cost of increased analytical and computational complexity. Fourth, although the CKLS specification offers greater flexibility than benchmark diffusion models by allowing the elasticity of volatility to vary with the state variable, the demand process is assumed to evolve continuously over time and does not explicitly accommodate abrupt structural changes, regime shifts, or jump behavior. In practice, industries may experience sudden disruptions arising from geopolitical events, technological innovation, regulatory interventions, or macroeconomic crises. Incorporating jump-diffusion dynamics or regime-switching mechanisms within the signalling framework constitutes a promising direction for future work.
	
	Finally, the equilibrium analysis is conducted within the class of Markovian Nash feedback strategies. Restricting attention to Markovian controls permits a transparent equilibrium characterization and facilitates the path-integral approach adopted in this paper. Nevertheless, firms may condition their decisions on broader histories of play, reputational considerations, or commitment mechanisms that generate non-Markovian behavior. Investigating alternative equilibrium concepts, including history-dependent strategies and learning dynamics, represents an important avenue for extending the analysis. Despite these limitations, the structure developed herein contributes to the literature by integrating asymmetric information, irreversible entry, state-dependent uncertainty, and path-integral methods within a unified continuous-time setting. The limitations identified above should therefore be viewed not as shortcomings of the present approach, but rather as opportunities for future research aimed at broadening the applicability and empirical content of dynamic signalling models under uncertainty.

\section*{Appendix.}

To prove Proposition \ref{p0} we need to go through several lemmas. Here we are going to discuss one by one.For Laplace transformation of the transitional probability $\tilde{\Psi}_\kappa(.,.)=\int_0^\infty\exp(-\kappa s)\Psi_s(.,.)ds$ with $\zeta^\be(t)$ being the inverse of $\mathfrak L^\be$, the L\'evy's result is 
\[
\E_\be\bigg\{\exp\{-\kappa\zeta^\be(t)\}\bigg\}=\exp\left\{\frac{-t}{\Psi_\kappa(\be,\be)}\right\}.
\]
It is important to note that the generator 
\[
\mathfrak G=\frac{1}{2}\frac{d^2}{dx^2}+\left[\alpha_1 + \alpha_2 x + \sum_{i=1}^2 \theta_i u^i\right],
\]
is the similarity transformation of the quantum Hamiltonian $\mathcal H$, $\mathfrak G=-\omega_0^{-1}\mathcal H \omega_0$, and $(\mathcal H\omega_0)(x)=0$, so that $\omega_0=\sqrt{\Psi_0}$ \citep{revuz2013continuous}.

\begin{lem}\label{l0}
Consider
\[
\omega_0(x)=\mathcal C\cdot\exp\left\{ \int_0^x \left[\alpha_1 + \alpha_2 y + \sum_{i=1}^2 \theta_i u^i\right]dy\right\}\in L^2(-\infty,\infty),
\]
and
\[
\mathfrak T(x)=\frac{1}{2}\left\{\left[\alpha_1 + \alpha_2 x + \sum_{i=1}^2 \theta_i u^i\right]\right\}
\]
is continuous and bounded below. Then the transition probability of market share $X$ yields
\[
\Psi_t(x,y)=\omega_0^{-1}\exp(-t\mathcal H)(x,y)\omega_0(y),\ \ \text{for}\ x,y\in\mathbb R^2,
\]
the heat kernel $\exp(-t\mathcal H)(x,y)$ is written in terms of Wiener integral over the Brownian bridge on $(0,t)$, $\phi_t(s)\cong B(s)-(s/t)B(t)$,
\begin{equation*}
    \exp(-t\mathcal H)(x,y)=(2\pi t)^{-1/2}\exp\left\{-(x-y)^2/2t\right\}\cdot
  \E_{\be_t}\left\{\exp\left[-\int_0^t\mathfrak T\left[\left(1-\frac{s}{t}\right)x+\frac{s}{t} y+\be_t(s)\right]\right]ds\right\},
\end{equation*}
where $\E_{\be_t}$ is the corresponding integral, for $x,y\in\mathbb R^2$. The invariant density of $X$ is $[\omega_0(x)]^2$.
\end{lem}
Define the first hitting time of the level $\be$ for the process with initial condition $x_0$ is 
\[
\tau_x(\be):=\inf\left\{s>0;\ X(s)=\be|X(0)=x_0\right\}.
\]
Suppose there exists a point $\be \in(x,y)$. As $X$ is two-dimensional time-homogenous process,
\[
\Phi_t(x,y)=\int_0^t\mathbb P\left[\tau_x(\be)\in du \right]\Psi_{t-u}(\be,y).
\]
Laplace transformation, semiboundedness of $\mathfrak T$ and the joint continuity of $\Psi_t(.,.)$ yield
\begin{lem}\label{exp}
  For every positive $\kappa$, and $x,\be\in\mathbb R^2,$ we have
  \[
  \E\bigg\{\exp\{-\kappa\tau_x(\be)\}\bigg\}=\frac{\tilde\Psi_\kappa(x,\be)}{\tilde\Psi_\kappa(\be,\be)},
  \]
  where $\tilde\Psi_\kappa(x,y)=\int_0^\infty \exp(-\kappa s)\Psi_s(x,y)$, for all $x,y\in\mathbb R^2$.
\end{lem}

\begin{proof}
Let \( \{X(s)\}_{s \geq 0} \) be a time-homogeneous diffusion process on \( \mathbb{R}^2 \) with generator \( \mathfrak{G} \), and suppose that \( \Psi_s(x, y) \) denotes its transition density with respect to the Lebesgue measure so that
\[
\mathbb{P}_x( X(s) \in A ) = \int_A \Psi_s(x, y)\, dy,
\]
for all Borel sets \( A \subset \mathbb{R}^2 \) and \( s > 0 \). The existence of \( \Psi_s \) is guaranteed under mild ellipticity and smoothness assumptions on the coefficients of the diffusion. Suppose \( \tau_x(\beta) \) denotes the first hitting time of the point \( \beta \in \mathbb{R}^2 \), defined by
\[
\tau_x(\beta) := \inf\left\{ s > 0 \mid X(s) = \beta \mid X(0) = x\right \}.
\]
We aim to compute the Laplace transform of this first hitting time,
\[
u(x) := \mathbb{E}_x \bigg\{ \exp\{-\kappa \tau_x(\beta)\} \bigg\}.
\]
Thus, \( u(x) \) satisfies a Dirichlet boundary value problem. More precisely, \( u(x) \) is the unique solution to the following boundary value problem
\[
\begin{cases}
\kappa u(x) - \mathfrak{G} u(x) = 0, & x \in \mathbb{R}^2 \setminus \{\beta\}, \\
u(\beta) = 1.
\end{cases}
\]
 The idea is to represent the solution in terms of the resolvent kernel of the semigroup \( \{P_s\}_{s \geq 0} \) associated to \( X(s) \), where \( P_s f(x) := \mathbb{E}_x\{f(X(s))\} \).

Define the resolvent operator:
\[
R_\kappa f(x) := \int_0^\infty \exp(-\kappa s) P_s f(x)\, ds.
\]
Then, the kernel associated to \( R_\kappa \) is given by
\[
\tilde{\Psi}_\kappa(x, y) := \int_0^\infty \exp(-\kappa s) \Psi_s(x, y)\, ds,
\]
which is often referred to as the Green function or resolvent kernel of the semigroup.

Let us now derive \( u(x) \) using the probabilistic interpretation of the kernel \( \tilde{\Psi}_\kappa \). Consider the additive functional:
\[
F(x) := \mathbb{E}_x \left\{\int_0^{\tau_x(\beta)} \exp(-\kappa s)\cdot \delta_\beta(X(s))\, ds \right\},
\]
where \( \delta_\beta \) is the Dirac delta at \( \beta \). By conditioning on the hitting time \( \tau_x(\beta) \). Now,
\[
F(x) = \int_0^\infty \exp(-\kappa s) \mathbb{P}_x(\tau_x(\beta) \in ds) = \mathbb{E}_x \left\{ \exp\big\{-\kappa \tau_x(\beta)\big\} \right\}.
\]
Now, consider the total resolvent kernel
\[
\tilde{\Psi}_\kappa(x, \beta) = \mathbb{E}_x \left\{ \int_0^\infty \exp(-\kappa s) \delta_\beta(X(s))\, ds \right\}.
\]
Since the process hits \( \beta \) only once, and remains infinitesimally there (almost surely), we may relate this to the Laplace transform of the first hitting time as follows. Strong Markov property, for \( x \neq \beta \) implies
\[
\tilde{\Psi}_\kappa(x, \beta)
= \mathbb{E}_x \left\{ \int_0^\infty \exp(-\kappa s) \Psi_s(x, \beta)\, ds \right\}
= \mathbb{E}_x \left\{ \exp\big\{-\kappa \tau_x(\beta)\big\} \tilde{\Psi}_\kappa(\beta, \beta) \right\}
= u(x) \cdot \tilde{\Psi}_\kappa(\beta, \beta),
\]
which implies 
\[
u(x) = \frac{ \tilde{\Psi}_\kappa(x, \beta) }{ \tilde{\Psi}_\kappa(\beta, \beta) }.
\]
Therefore,
\[
\E\left\{ \exp\big\{-\kappa \tau_x(\beta)\big\} \right\} = \frac{ \tilde{\Psi}_\kappa(x, \beta) }{ \tilde{\Psi}_\kappa(\beta, \beta) },
\]
as claimed. It is important to note that the denominator \( \tilde{\Psi}_\kappa(\beta, \beta) \) is finite due to the assumption that the process is non-explosive, and the transition density \( \Psi_s(x, y) \) is locally integrable in both \( s \) and \( y \). Moreover, the numerator \( \tilde{\Psi}_\kappa(x, \beta) \) is well-defined since the kernel is smooth in both arguments away from the diagonal, and satisfies the backward Kolmogorov equation with respect to the first argument. This completes the proof.
\end{proof}

\begin{lem}\label{exp1}
   For every $\be\in\mathbb R$ and positive continuous time $t$, we have
   \begin{equation*}
       v_0^\pm[t,\infty)=\Psi_0^{-1}(\be)\int_{x\gtrless\be}\Psi_0(x)\frac{d}{dt}\mathbb P\left[\tau_x(\be)<t\right]dx.
   \end{equation*}
\end{lem}
\begin{proof}
    From Lemma \ref{exp} we know that,
    \begin{equation*}
        \left[\tilde{\Psi}_\kappa(\beta, \beta)\right]^{-1}\int_{\mathbb R}\tilde{\Psi}_\kappa(x, \beta)\Psi_0(x)dx=\int_{\mathbb R}\Psi_0(x)\E\bigg\{ \exp\big\{-\kappa \tau_x(\beta)\big\} \bigg\}dx.
    \end{equation*}
    Since $\tilde{\Psi}_\kappa(x, \beta)=\int_0^\infty\exp(-\kappa s) {\Psi}_s(x, \beta)ds$, and as boundedness of the invariant density is $\Psi$, by Fubini's theorem we get
    \begin{equation*}
    \int_{\mathbb R}\tilde{\Psi}_\kappa(x, \beta)\Psi_0(x)dx= \int_0^\infty\exp(-\kappa s)ds\left[\int_{\mathbb R}{\Psi}_s(x, \beta)\Psi_0(x)ds\right] =\kappa^{-1} \Psi_0(\be). 
    \end{equation*}
    Thus,
    \begin{align*}
     \left[\tilde{\Psi}_\kappa(\beta, \beta)\right]^{-1}&=\kappa\int_{\mathbb R} \frac{\Psi_0(x)}{\Psi_0(\be)} \E\bigg\{ \exp\big\{-\kappa \tau_x(\beta)\big\} \bigg\} dx\\
     &=\int_{\mathbb R} \frac{\Psi_0(x)}{\Psi_0(\be)}\left[\int_0^\infty d_s[1-\exp(-\kappa s)]\frac{d}{ds}\mathbb P\left[\tau_x(\be)<t\right]\right]dx.
    \end{align*}
    After exchanging the order of the integration we get $\int_{\mathbb R} dx=\int_{-\infty}^\be dx+\int_\be^\infty dx$. This completes the proof.
\end{proof}
Finally, by the consequence of the Girsanov-Cameron-Martin theorem, and for all $x \gtrless\be$, we have
\[
\mathbb P\left[\tau_x(\be)<t\right]=\lim_{\kappa\uparrow\infty}\E_x\left\{\exp\left[-\kappa\int_0^t\Phi_\mp[X(s)]ds\right]\right\},
\]
where $\Phi_{\mp}$ is the characteristic function of $\left\{x:x\gtrless\be\right\}$.

\begin{lem}\label{exp2}
 For every positive time $t$, and for all  $x\gtrless\be$, we can have
 \[
 \mathbb P\left[\tau_x(\be)<t\right]=\omega_0^{-1}(x)\left[\exp\bigg(-t\mathcal H_{\pm}(\be)\bigg)\omega_0\right](x),
 \]
 where $\mathcal H_{\pm}(\be)$ is the Diritchlet Hamiltonian which is defined as
 \[
 \mathcal H_{\pm}(\be):=\lim_{\kappa\uparrow\infty}\left[-\frac{1}{2}\frac{d^2}{dx^2}+\mathfrak T(x)+\kappa\Phi_\mp(x)\right].
 \]
\end{lem}
\begin{proof}
Let \( X(s) \) be a diffusion process on \( \mathbb{R} \) governed by the generator
\[
\mathfrak{G} = \frac{1}{2} \frac{d^2}{dx^2} + b(x)\frac{d}{dx},
\]
where the drift term \( b(x) = \alpha_1 + \alpha_2 x + \sum_{i=1}^2 \theta_i u^i \), with constants \( \alpha_j \in \mathbb{R} \), controls \( u^i \), and smooth potential \( \mathfrak{T}(x) \) satisfying conditions of Assumption \ref{a0}. The infinitesimal generator \( \mathfrak{G} \) corresponds under a Doob–h transform to the quantum Hamiltonian
\[
\mathcal{H} := -\Psi_0^{-1/2} \mathfrak{G} \Psi_0^{1/2} = -\frac{1}{2} \frac{d^2}{dx^2} + \mathfrak{T}(x),
\]
where \( \Psi_0(x) \propto \exp\left\{2\int_0^x b(y) \, dy\right\} \) is the invariant density of \( X(s) \). Let \( \tau_x(\beta) \) be the first hitting time of the point \( \beta \in \mathbb{R} \) by the process starting at \( x \). The probability that \( \tau_x(\beta) < t \) can be characterized through semigroup with absorbing boundary conditions imposed at \( \beta \). To define the absorption at \( \beta \) consider a penalization potential \( \kappa \Phi_\mp(x) \) such that
\[
\Phi_\mp(x) := \mathbf{1}_{\mathbb{R}_\mp} (x - \beta)
=
\begin{cases}
1, & x \lessgtr \beta, \\
0, & x \gtrless \beta.
\end{cases}
\]
 This ensures that as \( \kappa \to \infty \), paths that enter the forbidden side are killed with overwhelming probability, effectively implementing Dirichlet boundary conditions at \( \beta \). Define the penalized Hamiltonian operator as
\[
\mathcal{H}_\kappa^\pm := -\frac{1}{2} \frac{d^2}{dx^2} + \mathfrak{T}(x) + \kappa \Phi_\mp(x).
\]
This operator is self-adjoint on a suitable dense domain (e.g., \( C_c^\infty(\mathbb{R}) \)) in \( L^2(\mathbb{R}) \), and as \( \kappa \to \infty \), the family \( \{\mathcal{H}_\kappa^\pm\}_{\kappa > 0} \) increases in the sense of quadratic forms. Therefore, Kato–Trotter monotone convergence theorem implies
\[
\lim_{\kappa \uparrow \infty} \exp\left(-t \mathcal{H}_\kappa^\pm\right) f = \exp\left\{-t \mathcal{H}_\pm(\beta)\right\} f,
\]
where \( \mathcal{H}_\pm(\beta) \) is the Dirichlet Hamiltonian corresponding to absorbing behavior at \( \beta \). Now, let \( \omega_0 \) be the square root of the invariant measure density \( \Psi_0 \), i.e., \( \omega_0(x) = \Psi_0^{1/2}(x) \). The Feynman-Kac representation for the penalized semigroup yields
\[
\left[\exp\left(-t \mathcal{H}_\kappa^\pm\right) \omega_0\right](x) = \mathbb{E}_x\left\{ \exp\left( -\int_0^t \kappa \Phi_\mp(X(s)) \, ds \right) \cdot \omega_0(X(t)) \right\},
\]
Dividing \( \omega_0(x) \) on both sides yields
\[
\omega_0^{-1}(x) \left[\exp\left(-t \mathcal{H}_\kappa^\pm\right) \omega_0\right](x)
= \mathbb{E}_x\left\{ \exp\left( -\int_0^t \kappa \Phi_\mp(X(s)) \, ds \right) \cdot \frac{\omega_0(X(t))}{\omega_0(x)} \right\}.
\]

After taking the limit as \( \kappa \uparrow \infty \), it can be observed that the exponential term converges pointwise to the indicator function of the event that \( X(s) \notin \mathbb{R}_\mp(\beta) \) for all \( s \in [0, t] \), i.e., that the process stays on the correct side and may reach \( \beta \). Therefore,
\[
\lim_{\kappa \uparrow \infty}
\mathbb{E}_x\left\{ \exp\left( -\int_0^t \kappa \Phi_\mp(X(s)) \, ds \right) \right\}
= \mathbb{P}_x\left[ \tau_x(\beta) < t \right].
\]
Thus, we arrive at the desired representation:
\[
\mathbb{P}\left[ \tau_x(\beta) < t \right] = \omega_0^{-1}(x) \left[ \exp\left( -t \mathcal{H}_\pm(\beta) \right) \omega_0 \right](x).
\]
This proves that the first-hitting probability of \( \beta \) up to time \( t \) can be represented via a heat semigroup associated with the Dirichlet Hamiltonian, applied to the invariant weight \( \omega_0 \), and renormalized at the initial point \( x \).
\end{proof}

\begin{prop}\label{p1}
For all $\be\in\mathbb R^2$, we can write
\[
\E_\be\bigg\{\exp\left\{-\kappa\mathfrak L_\pm\left(\zeta^\be(t)\right)\right\}\bigg\}=\exp\left\{-t\int_0^\infty\left[1-\exp\{-\kappa s\}\right]dv_\be^\pm(s)\right\},
\]
such that $\kappa,t\in(0,\infty)^2$, and 
\[
v_\be^\pm[t,\infty)=\frac{1}{\omega_0^2}\bigg[\omega_0,\mathcal H_\pm(\be)\cdot\exp\left\{-t\mathcal H_\pm(\be)\right\}\cdot\omega_0\bigg]_{L^2},
\]
where $\exp\left\{-t\mathcal H_\pm(\be)\right\}\cdot\omega_0\in\mathfrak D[\mathcal H_\pm(\be)],\ \forall\ \ t>0$.
\end{prop}

\begin{proof}
We are going to use the fact that
\[
v_\beta^\pm[t,\infty) = \frac{1}{\Psi_0(\beta)} \int_{-\infty}^{\beta} \Psi_0(x) \frac{d}{dt} \mathbb{P}\left[ \tau_x(\beta) < t \right] dx + \frac{1}{\Psi_0(\beta)} \int_{\beta}^{\infty} \Psi_0(x) \frac{d}{dt} \mathbb{P}\left[ \tau_x(\beta) < t \right] dx,
\]
and that \( \mathbb{P}[\tau_x(\beta) < t] \) is already known from lemma \ref{exp2} to satisfy
\[
\mathbb{P}[\tau_x(\beta) < t] = \omega_0^{-1}(x) \left[ \exp\left( -t \mathcal{H}_\pm(\beta) \right) \omega_0 \right](x),
\]
where \( \omega_0 = \Psi_0^{1/2} \) is the square root of the invariant density, and \( \mathcal{H}_\pm(\beta) \) is the Dirichlet Hamiltonian operator defined via absorbing behavior at the point \( \beta \). 

Let us denote \( u(x,t) := \mathbb{P}[\tau_x(\beta) < t] \). Then the time derivative with respect to \( t \) exists and satisfies
\[
\frac{d}{dt} u(x,t) = -\frac{1}{\omega_0(x)} \bigg[ \mathcal{H}_\pm(\beta) \cdot \exp(-t \mathcal{H}_\pm(\beta)) \omega_0 \bigg](x),
\]
because for a strongly continuous semigroup \( \{e^{-tA}\}_{t \geq 0} \) on a Hilbert space, and any vector \( f \in \mathcal{D}(A) \), the function \( t \mapsto \exp(-t\cdot A)f \) is differentiable with derivative
\[
\frac{d}{dt} \exp(-t\cdot A)f = -A \cdot \exp(-t\cdot A)f.
\]
The above expression holds pointwise due to smoothness and positivity of \( \omega_0 \), which follows from the assumptions placed on \( \mathfrak{T}(x) \), and due to the fact that \( \exp(-t \mathcal{H}_\pm(\beta))\omega_0 \in \mathfrak{D}[\mathcal{H}_\pm(\beta)] \) for all \( t > 0 \), by semigroup. Substituting this into the expression for \( v_\beta^\pm[t,\infty) \), we obtain
\begin{align*}
v_\beta^\pm[t,\infty) &= - \frac{1}{\Psi_0(\beta)} \int_{-\infty}^{\beta} \Psi_0(x) \cdot \frac{1}{\omega_0(x)} \bigg[ \mathcal{H}_\pm(\beta) \cdot \exp(-t \mathcal{H}_\pm(\beta)) \omega_0 \bigg](x) dx \\
&\hspace{1cm} - \frac{1}{\Psi_0(\beta)} \int_{\beta}^{\infty} \Psi_0(x) \cdot \frac{1}{\omega_0(x)} \bigg[ \mathcal{H}_\pm(\beta) \cdot \exp(-t \mathcal{H}_\pm(\beta)) \omega_0 \bigg](x) dx.
\end{align*}
Since \( \Psi_0(x) = \omega_0^2(x) \), and \( \omega_0^{-1}(x) \cdot \Psi_0(x) = \omega_0(x) \), the above yields
\[
v_\beta^\pm[t,\infty) = \frac{1}{\omega_0^2(\beta)} \int_{\mathbb{R}} \omega_0(x) \cdot \left[ \mathcal{H}_\pm(\beta) \cdot \exp(-t \mathcal{H}_\pm(\beta)) \omega_0 \right](x) dx.
\]
This is precisely the \( L^2 \) inner product
\[
v_\beta^\pm[t,\infty) = \frac{1}{\omega_0^2(\beta)} \bigg[ \omega_0, \mathcal{H}_\pm(\beta) \cdot \exp(-t \mathcal{H}_\pm(\beta)) \omega_0 \bigg]_{L^2}.
\]
Let us now address the main expression of the proposition. Note that \( \mathfrak{L}_\pm(\zeta^\beta(t)) \) denotes the cumulative local time of excursions from \( \beta \), accumulated over the time interval \( [0, t] \). From classical results in excursion theory and the Laplace transform of Poisson integrals, we have that the Laplace functional of the local time process at \( \beta \) can be expressed via the L\'evy-Khintchine representation
\[
\E_\beta\left\{ \exp\left( -\kappa \mathfrak{L}_\pm(\zeta^\beta(t)) \right) \right\} = \exp\left\{ -t \int_0^\infty \left(1 - e^{-\kappa s}\right) dv_\beta^\pm(s) \right\}.
\]
This follows by modeling the cumulative local time \( \mathfrak{L}_\pm \) as the sum over a Poisson process on the space of excursions, where the jump sizes are distributed according to the excursion length measure \( v_\beta^\pm \). Each such excursion contributes a random length \( s \), and the Poisson summation formula gives the Laplace transform in the form of an exponential of the tail Laplace measure.

Hence, using the previously derived spectral expression for \( v_\beta^\pm[t,\infty) \), we conclude that
\[
\E_\beta\left\{ \exp\left( -\kappa \mathfrak{L}_\pm(\zeta^\beta(t)) \right) \right\}= \exp\left\{ -t \int_0^\infty \left(1 - e^{-\kappa s}\right) dv_\beta^\pm(s) \right\},
\]
with
\[
v_\beta^\pm[t,\infty) = \frac{1}{\omega_0^2(\beta)} \left[ \omega_0, \mathcal{H}_\pm(\beta) \cdot \exp(-t \mathcal{H}_\pm(\beta)) \omega_0 \right]_{L^2},
\]
as required. This completes the proof.
\end{proof}

By the reflection principle we know that for every path generated by the Brownian motion $B^x(.)=x+B(.)$, with $\tau_x(\be)<t$, there exists an equally likely reflected Brownian path $\tilde B^x(.)$ such that the following conditions hold,
\[
\tilde B^x(.):=\begin{cases}
    B^x(s),\ \text{for all}\ s<\tau_x(\be),\\
    2\be-B^x(s),\ \text{for all}\ s>\tau_x(\be).
\end{cases}
\]
This yields the absorbing Brownian motion with the generator $\left(2^{-1}\right)d^2/dx^2$ and domain $\mathfrak D=\{g\in L^2(\mathbb R^\pm):g(0)=0\}$ the transition probability 
\[
\Psi_0^t(x,z)=\frac{1}{\sqrt{2\pi t}}\bigg[\exp\left\{-\frac{(x-z)^2}{2t}\right\}-\exp\left\{-\frac{(x+z)^2}{2t}\right\}\bigg].
\]
\begin{lem}\label{exp3}
   For a continuous function $W$, bounded from below, and symmetric around $x=\be$, we can write
   \begin{align*}
   & \E\biggr\{\exp\left[-\int_0^t W[x+B(s)]ds\right]\bigg |\tau_x(\be)<t,\ x+B(t)\in dz\biggr\} \\
   &\hspace{1cm}=\E\biggr\{\exp\left[-\int_0^t W[x+B(s)]ds\right]\bigg |x+B(t)\in d(2\be-z)\biggr\},
   \end{align*}
   such that $x,z>\be$, and $\tau_x(\be)=\inf\left\{s>0:\ X(s)=\be\big |\ X(0)=x\right\}$.
\end{lem}

\begin{proof}
We begin by noting that \( x > \beta \) and \( z > \beta \) are fixed throughout, and we consider Brownian motion starting at \( x \). Define the path space
\[
\Omega := \left\{ \omega \in C([0,t], \mathbb{R}) : \omega(0) = x \right\},
\]
with the canonical Wiener measure \( \mathbb{P}_x \), under which \( B(s) \) is a standard Brownian motion starting at \( x \). Define the shifted process
$X(s) := x + B(s)$, which is simply Brownian motion started at \( x \). Fix time at \( t > 0 \), and define the event
\[
A := \{ \tau_x(\beta) < t,\ X(t) \in dz \}.
\]
We are interested in computing the conditional expectation of the Feynman–Kac exponential functional of \( W \), restricted to paths that both hit \( \beta \) before time \( t \) and end at position \( z \), and we aim to show this coincides with the corresponding expectation over reflected paths ending at \( 2\beta - z \). The reflected path is $\tilde{X}(s) := 2\beta - X(s)$,
and let \( \mathcal{R} \) be the reflection operator acting on path space is $(\mathcal{R} \omega)(s) = 2\beta - \omega(s).$
The key point is that the Brownian motion is symmetric under reflection, for any \( x > \beta \), the law of \( X(s) \) under \( \mathbb{P}_x \) conditioned on \( X(t) = z \) is the same as the law of \( \tilde{X}(s) \) under \( \mathbb{P}_x \) conditioned on \( \tilde{X}(t) = 2\beta - z \). This is a consequence of the pathwise time-reversal symmetry of Brownian bridges and the reflection invariance of Brownian motion.

 Consider the Brownian bridge \( \mathbb{P}_x^{z,t} \), i.e., Brownian motion started at \( x \), conditioned to end at \( z \) at time \( t \). Under this measure, the process \( X(s) \) has a well-defined transition density and is symmetric in the sense that the path \( \tilde{X}(s) := 2\beta - X(s) \) has the same law under \( \mathbb{P}_x^{z,t} \) as \( X(s) \) under \( \mathbb{P}_x^{2\beta - z,t} \). That is,
\[
\mathcal{R}_* \mathbb{P}_x^{z,t} = \mathbb{P}_x^{2\beta - z,t}.
\]
Consider the functional
\[
F[X] = \exp\left\{ -\int_0^t W(X(s))\,ds \right\},
\]
which is path-dependent, but crucially, because \( W \) is symmetric about \( \beta \), i.e., \( W(\beta + u) = W(\beta - u) \), we have
\[
W(X(s)) = W(2\beta - X(s)) = W(\tilde{X}(s)),
\]
for all paths \( X \in \Omega \) and all times \( s \in [0,t] \). Consequently,
$F[X] = F[\tilde{X}]$,
for all paths \( X \). In other words, the Feynman-Kac functional is invariant under reflection about \( \beta \). The conditional expectation is
\[
\mathbb{E}\left\{ F[X] \,\bigg|\, \tau_x(\beta) < t,\ X(t) = z \right\}.
\]
Note that the event \( \{ \tau_x(\beta) < t \} \) is preserved under reflection: a path that starts at \( x > \beta \), hits \( \beta \), and ends at \( z > \beta \), when reflected about \( \beta \), becomes a path that starts at \( x \), hits \( \beta \), and ends at \( 2\beta - z < \beta \), with the same path measure. Therefore, the conditional measure of paths hitting \( \beta \) and ending at \( z \), when pushed forward under \( \mathcal{R} \), becomes the conditional measure of paths hitting \( \beta \) and ending at \( 2\beta - z \). Using this invariance and the fact that \( F[X] = F[\tilde{X}] \), we obtain
\begin{equation*}
\E\left\{ F[X] \,\bigg|\, \tau_x(\beta) < t,\ X(t) = z \right\}
= \mathbb{E}\left\{ F[\tilde{X}] \,\bigg|\, \tau_x(\beta) < t,\ \tilde{X}(t) = 2\beta - z \right\}.
\end{equation*}
Since \( \tilde{X} \) has the same law as \( X \) under the reflected endpoint condition, this is equal to $\mathbb{E}\left\{ F[X] \,\bigg|\, X(t) = 2\beta - z \right\}$.
Hence we obtain the desired identity
which proves the claim.
\end{proof}
The result from lemma \ref{exp3} leads to 
\[
\exp\left\{-t\cdot\mathcal H_\pm(\be)\right\}(x,z)=\exp\left\{-t\cdot\mathcal H_\pm^{SYM}(\be)\right\}(x,z)-\exp\left\{-t\cdot\mathcal H_\pm^{SYM}(\be)\right\}(2\be-x,z),
\]
for all $x,z>\be$, and $\mathcal H_\pm^{SYM}$ is the symmetric Hamiltonian \citep{truman1992elementary}. Invoking proposition \ref{p1} implies
\begin{align*}
v_\be^\pm[t,\infty)&=-\frac{1}{\Psi_0(\be)}\int_{x\gtrless \be}\sqrt{\Psi_0(\be+x)}\cdot\biggr[\int_{z\gtrless 0}\frac{\partial}{\partial s}\bigg[\exp\left\{-s\cdot\mathcal H_\pm^{SYM}(\be)\right\}(\be+x,\be+z)\\
&\hspace{2cm}-\exp\left\{-s\cdot\mathcal H_\pm^{SYM}(\be)\right\}(\be-x,\be+z)\bigg]\sqrt{\Psi_0(\be+z)}\ dz\biggr]dx.
\end{align*}
Finally, an integration by parts implies 
\begin{align*}
   v_\be^\pm[t,\infty)&=\pm\frac{1}{2\sqrt{\Psi_0(\be)}} \int_{y\gtrless 0}\biggr[\frac{\partial}{\partial x}\bigg |_{x=\be_+}\exp\left\{-s\cdot\mathcal H_\pm^{SYM}(\be)\right\}(x,\be+z)\\
   &\hspace{2cm}+\frac{\partial}{\partial x}\bigg |_{x=\be_-}\exp\left\{-s\cdot\mathcal H_\pm^{SYM}(\be)\right\}(x,\be+z)\biggr]\cdot \sqrt{\Psi_0(\be+z)}\ dz,
\end{align*}
which yields the statement of proposition \ref{p0}.

\subsection*{Computation of optimal advertisement spending.}
Here we are going to provide the detailed calculation of optimal expenditure on advertisement for two firms. We start with Equations \eqref{h1} and \eqref{h2}. Now

\begin{equation}\label{c1}
  \frac{\partial h^{1}}{\partial u^{1}}
= - e^{-\gamma s} \, x
\;+\;
\theta_{1} \alpha_{2} s \, \exp{\{\alpha_{2} s x - k\}}.  
\end{equation}
Equation \eqref{c1} describes how the value function $h^{1}$ defined in Equation \eqref{h1} responds to changes in the 
incumbent firm's advertisement strategy $u^{1}$. The first term, $-e^{-\gamma s}x$, shows a direct negative sensitivity that is discounted over time through the exponential 
factor $e^{-\gamma s}$. The second term reflects the effect of the advertisement strategy on the state dynamics, scaled by the parameters $\theta_{1}$ and $\alpha_{2}$, and weighted 
by the exponential factor $\exp\{\alpha_{2} s x - k\}$. Together, these terms describe 
the trade-off between immediate costs and future state-dependent benefits in the 
firm's optimization problem. Here
\[
Y = -\,\exp{\{-\gamma s\}}\, x 
    + \theta_{1}\,\alpha_{2}\,s\,\exp{\{\alpha_{2} s x - k\}},
\]
represents how the variable $Y$ combines a discounted state term with a control–sensitive exponential adjustment that depends on the parameters $\theta_{1}$, $\alpha_{2}$, $s$, and the state variable $x$. Moreover,
\[
\begin{aligned}
\frac{\partial h^{1}}{\partial x}
&= \exp{\{-\gamma s\}}\,( Q - u^{1} )
   + \mathbbm{1}_{(\rho \le t)}\,\exp{\{-\gamma s\}}\,( M^{W} - \bar{u}^{3} )  \\
&\quad + \alpha_{2}s\,\exp{\{\alpha_{2} s x - k\}}
   + \alpha_{2}\,\exp{\{\alpha_{2} s x - k\}}
   + \alpha_{2}^{2} s x\,\exp{\{\alpha_{2} s x - k\}} \\
&\quad + \Big[
      \alpha_{2}^{2} s\,\exp{\{\alpha_{2} s x - k\}}
      + \big( \alpha_{1} + \alpha_{2} x + \theta_{1}u^{1} + \theta_{2}u^{2} \big)
        \alpha_{2}^{2} s^{2}\,\exp{\{\alpha_{2} s x - k\}}
    \Big] \\
&\quad + \frac{1}{2}\alpha_{3}^{2}\alpha_{2}^{2} s^{2}
    \Big[
        2\alpha_{4} x^{2\alpha_{4}-1}\,\exp{\{\alpha_{2} s x - k\}}
        + x^{2\alpha_{4}}\,\alpha_{2} s\,\exp{\{\alpha_{2} s x - k\}}
    \Big],
\end{aligned}
\]
captures how the value function $h^{1}$ changes with respect to $x$, combining discounted payoff terms, exponential sensitivity to the state through $\exp\{\alpha_{2}s x - k\}$, and additional contributions from both the drift and diffusion components of the underlying CKLS-type dynamics. Now $\frac{\partial h^1}{\partial x} = u^{1}A+B$ shows that the sensitivity of $h^{1}$ with respect to the state variable $x$ depends linearly on the control $u^{1}$, with coefficients $A$ and $B$ capturing the contributions from the model parameters and state dynamics. Furthermore, $A = -e^{-\gamma s} +\theta_1\alpha_2^2s^2e^{\alpha_2 sx-k}$ is the coefficient which captures both the discounted effect of the state variable and the exponential sensitivity introduced by the control parameter $\theta_{1}$ and the CKLS dynamics. Now,
\begin{align}\label{c2}
B
&= \exp\{-\gamma s\}\,Q
   + \mathbbm{1}_{(\rho \le t)}\,\exp\{-\gamma s\}\,\{M^{W} - \bar{u}^{3}\}\notag \\[6pt]
&\quad + \alpha_{2}s\,\exp\{\alpha_{2}s x - k\}
   + \alpha_{2}\,\exp\{\alpha_{2}s x - k\}
   + \alpha_{2}^{2}s x\,\exp\{\alpha_{2}s x - k\} \notag\\[6pt]
&\quad + \Big(
      \alpha_{2}^{2}s\,\exp\{\alpha_{2}s x - k\}
      + (\alpha_{1} + \alpha_{2}x + \theta_{2}u^{2})\,
        \alpha_{2}^{2}s^{2}\,\exp\{\alpha_{2}s x - k\}
    \Big) \notag\\[6pt]
&\quad + \frac{1}{2}\alpha_{3}^{2}\alpha_{2}^{2}s^{2}\Big(
        2\alpha_{4}x^{\,2\alpha_{4}-1}\,\exp\{\alpha_{2}s x - k\}
        + x^{\,2\alpha_{4}}\,\alpha_{2}s\,\exp\{\alpha_{2}s x - k\}
    \Big).
\end{align}
The term \(B\) in Equation \eqref{c2} aggregates all components of \(\frac{\partial h^{1}}{\partial x}\) that do not depend on the control \(u^{1}\). It incorporates discounted payoff terms, indicator–driven adjustments, exponential sensitivity from the CKLS dynamics, and contributions from both the drift and diffusion structures of the underlying stochastic process. The term $B =u^{2}B_1+B_2$ shows that the term \(B\) depends linearly on the control variable \(u^{2}\), where \(B_{1}\) captures the portion influenced by the entrant’s control and \(B_{2}\) collects all remaining components independent of \(u^{2}\). Moreover, $B_{1} = \theta_{2}\alpha_{2}^{2}s^{2}\,\exp\{\alpha_{2}s x - k\}$describes the component of \(B\) that scales with the control \(u^{2}\), capturing how the entrant’s influence enters the system through the exponential sensitivity of the CKLS-SDE.

Now the term
\[
\begin{aligned}
B_2
&= \exp\{-\gamma s\}\,Q
   + \mathbbm{1}_{(\rho \le t)}\,\exp\{-\gamma s\}\,\{M^{W} - \bar{u}^{3}\} \\[6pt]
&\quad + \alpha_{2}s\,\exp\{\alpha_{2}s x - k\}
   + \alpha_{2}\,\exp\{\alpha_{2}s x - k\}
   + \alpha_{2}^{2}s x\,\exp\{\alpha_{2}s x - k\} \\[6pt]
&\quad + \Big(
      \alpha_{2}^{2}s\,\exp\{\alpha_{2}s x - k\}
      + (\alpha_{1} + \alpha_{2}x)\,\alpha_{2}^{2}s^{2}\,
        \exp\{\alpha_{2}s x - k\}
    \Big) \\[6pt]
&\quad + \frac{1}{2}\alpha_{3}^{2}\alpha_{2}^{2}s^{2}\Big(
        2\alpha_{4}x^{\,2\alpha_{4}-1}\,\exp\{\alpha_{2}s x - k\}
        + x^{\,2\alpha_{4}}\,\alpha_{2}s\,\exp\{\alpha_{2}s x - k\}
    \Big),
\end{aligned}
\]
collects all components of \(B\) that are independent of the control \(u^{2}\). It includes discounted payoff contributions, indicator-based adjustments, and several exponential terms arising from the drift and diffusion of the CKLS-type dynamics, reflecting how the state variable \(x\) influences the system through both linear and nonlinear channels. The partial derivative
\[
\begin{aligned}
\frac{\partial^{2} h^{1}}{\partial x^{2}}
&= u^{1}\theta_{1}\alpha_{2}^{3}s^{3}\,\exp\{\alpha_{2}s x - k\}
   + \alpha_{2}^{2}s^{2}\,\exp\{\alpha_{2}s x - k\}
   + \alpha_{2}^{2}s\,\exp\{\alpha_{2}s x - k\} \\
&\quad + \alpha_{2}^{3}s^{2}\,\exp\{\alpha_{2}s x - k\}
   + \alpha_{2}^{2}s\,\exp\{\alpha_{2}s x - k\}
   + \alpha_{2}^{3}s^{2}x\,\exp\{\alpha_{2}s x - k\} \\[4pt]
&\quad + \alpha_{2}^{3}s^{2}\,\exp\{\alpha_{2}s x - k\}
   + (\alpha_{1} + \alpha_{2}x + \theta_{2}u^{2})\,
     \alpha_{2}^{3}s^{3}\,\exp\{\alpha_{2}s x - k\} \\[6pt]
&\quad + \frac{1}{2}\alpha_{3}^{2}\alpha_{2}^{2}s^{2}\Big(
      2\alpha_{4}(2\alpha_{4}-1)\,x^{\,2\alpha_{4}-2}\,\exp\{\alpha_{2}s x - k\} \\
&\qquad + 4\alpha_{4}\alpha_{2}s\,x^{\,2\alpha_{4}-1}\,\exp\{\alpha_{2}s x - k\}+ x^{\,2\alpha_{4}}\,\alpha_{2}^{2}s^{2}\,\exp\{\alpha_{2}s x - k\}
    \Big) ,
\end{aligned}
\]
represents the curvature of the value function with respect to the state variable \(x\). It incorporates contributions from the control variables \(u^{1}\) and \(u^{2}\), multiple exponential sensitivity terms arising from the CKLS SDE, and higher-order effects from the diffusion component, reflecting how nonlinearities in both drift and volatility shape the responsiveness of \(h^{1}\) to changes in \(x\). Now, the term $\frac{\partial^2 h^1}{\partial x^2} = u^{1}C+D$ shows that the curvature of the value function with respect to \(x\) depends linearly on the control \(u^{1}\), where \(C\) captures the portion influenced directly by the control, while \(D\) contains all remaining terms independent of \(u^{1}\). The term $C = \theta_{1}\alpha_{2}^{3}s^{3}\,\exp\{\alpha_{2}s x - k\}$ is the component of the second derivative that scales with the control \(u^{1}\), capturing how the incumbent’s action influences the curvature of the value function through the exponential sensitivity inherent in the CKLS SDE, and 
\[
\begin{aligned}
D
&= \alpha_{2}^{2}s^{2}\,\exp\{\alpha_{2}s x - k\}
   + \alpha_{2}^{2}s\,\exp\{\alpha_{2}s x - k\}
   + \alpha_{2}^{3}s^{2}\,\exp\{\alpha_{2}s x - k\}
   + \alpha_{2}^{2}s\,\exp\{\alpha_{2}s x - k\} \\[4pt]
&\quad + \alpha_{2}^{3}s^{2}x\,\exp\{\alpha_{2}s x - k\}
   + \alpha_{2}^{3}s^{2}\,\exp\{\alpha_{2}s x - k\}
   + (\alpha_{1} + \alpha_{2}x + \theta_{2}u^{2})\,
     \alpha_{2}^{3}s^{3}\,\exp\{\alpha_{2}s x - k\} \\[6pt]
&\quad + \frac{1}{2}\alpha_{3}^{2}\alpha_{2}^{2}s^{2}\Big(
      2\alpha_{4}(2\alpha_{4}-1)\,x^{\,2\alpha_{4}-2}\,\exp\{\alpha_{2}s x - k\}  + 4\alpha_{4}\alpha_{2}s\,x^{\,2\alpha_{4}-1}\,\exp\{\alpha_{2}s x - k\} \\
&\qquad + x^{\,2\alpha_{4}}\,\alpha_{2}^{2}s^{2}\,\exp\{\alpha_{2}s x - k\}
    \Big),
\end{aligned}
\]
collects all components of the second derivative \(\frac{\partial^{2} h^{1}}{\partial x^{2}}\) that do not depend on the control \(u^{1}\). It consists of several exponential terms arising from the drift and diffusion of the CKLS-type dynamics, as well as nonlinear contributions involving the state variable \(x\), reflecting how curvature in the value function is shaped by both the model parameters and the underlying stochastic structure. In equation $D =  u^{2}D_1+D_2$ describes that the term \(D\), representing the control-independent portion of the second derivative, depends linearly on the entrant’s control \(u^{2}\), with \(D_{1}\) capturing the part influenced by \(u^{2}\) and \(D_{2}\) containing all remaining terms independent of the control, and $D_{1} = \theta_{2}\alpha_{2}^{3}s^{3}\,\exp\{\alpha_{2}s x - k\}$ represents the portion of \(D\) that varies with the entrant’s control \(u^{2}\), capturing how the entrant’s action influences the curvature of the value function through the exponential sensitivity of the CKLS SDE. The condition
\[
\begin{aligned}
D_{2}
&= \alpha_{2}^{2}s^{2}\,\exp\{\alpha_{2}s x - k\}
   + \alpha_{2}^{2}s\,\exp\{\alpha_{2}s x - k\}
   + \alpha_{2}^{3}s^{2}\,\exp\{\alpha_{2}s x - k\}
   + \alpha_{2}^{2}s\,\exp\{\alpha_{2}s x - k\} \\[4pt]
&\quad + \alpha_{2}^{3}s^{2}x\,\exp\{\alpha_{2}s x - k\}
   + \alpha_{2}^{3}s^{2}\,\exp\{\alpha_{2}s x - k\}
   + (\alpha_{1} + \alpha_{2}x)\,\alpha_{2}^{3}s^{3}\,\exp\{\alpha_{2}s x - k\} \\[6pt]
&\quad + \frac{1}{2}\alpha_{3}^{2}\alpha_{2}^{2}s^{2}\Big(
        2\alpha_{4}(2\alpha_{4}-1)\,x^{\,2\alpha_{4}-2}\,\exp\{\alpha_{2}s x - k\} + 4\alpha_{4}\alpha_{2}s\,x^{\,2\alpha_{4}-1}\,\exp\{\alpha_{2}s x - k\}  \\
&\qquad + x^{\,2\alpha_{4}}\,\alpha_{2}^{2}s^{2}\,\exp\{\alpha_{2}s x - k\}
    \Big).
\end{aligned}
\]
contains all components of \(D\) that are independent of the control \(u^{2}\). It includes multiple exponential terms generated by the drift and diffusion of the CKLS SDE, along with nonlinear contributions involving the state variable \(x\), reflecting how model parameters and the underlying stochastic structure shape the curvature of the value function. The derivative
\[
\frac{\partial^{2} h^{1}}{\partial x \,\partial u^{1}}
= -\exp\{-\gamma s\} + \theta_{1}\alpha_{2}^{2}s^{2}\,\exp\{\alpha_{2}s x - k\}
\]
describes how the sensitivity of \(h^{1}\) with respect to the state variable \(x\) changes as the control \(u^{1}\) varies, combining a discounted term with an exponential component driven by the model parameters. Moreover,
\[
Z = -\exp\{-\gamma s\} + \theta_{1}\alpha_{2}^{2}s^{2}\,\exp\{\alpha_{2}s x - k\}
\]
represents the mixed sensitivity term combining a discounted effect with an exponential response driven by the model parameters, reflecting how changes in the control influence the state-dependent behavior of the system. Finally,
\[
\frac{\partial h}{\partial u^{1}}\cdot\left( \frac{\partial^2 h}{\partial x^2}\right)^2-2\frac{\partial h}{\partial x}\cdot\left( \frac{\partial^2 h}{\partial x \partial u^{1}}\right) = 0
\]
represents an equilibrium-type relation linking the control derivative, the curvature of the value function, and the mixed sensitivity term, characterizing a balance condition that arises in the optimization of the incumbent’s strategy.
Since, \[
Y\bigg[u^{1}C + D\bigg]^{2}
\;-\;
2\bigg[u^{1}A + B\bigg]Z
= 0
\]
expresses the optimality condition in terms of the decomposed components \(A, B, C, D, Y,\) and \(Z\), showing how the control \(u^{1}\) enters quadratically and interacts with the system’s sensitivity and curvature terms to determine equilibrium behavior, then
\[
\bigg[YC^2\bigg]{u^{1}}^2
\;+\;
2\bigg[CDY - AZ\bigg]u^1+\;
\bigg[YD^2 - 2BZ\bigg]
= 0,
\]
which gives us an optimal expenditure on advertisement for the incumbent as
\[
u^{1}
= \frac{1}{YC^{2}}\left[AZ-CDY\pm \Big[(CDY - AZ)^{2} - YC^{2}\big(YD^{2} - 2BZ\big)\Big]^{1/2}\right].
\]

Now, assume the expenditure on advertisement is positive. Therefore,
{\small
\[
u^{1}
= \frac{1 }{YC^{2}}\left[AZ - CY\Big(u^{2}D_{1}+D_{2}\Big)+\Bigg[
    \Big[CY\Big(u^{2}D_{1}+D_{2}\Big) - AZ\Big]^{2}
    - YC^{2}\Big[Y\Big(u^{2}D_{1}+D_{2}\Big)^{2} - 2Z\Big(u^{2}B_{1}+B_{2}\Big)\Big]
\Bigg]^{1/2}\right] 
\]
}
arises as the solution to a quadratic equation obtained from the optimality
condition. Since the quadratic yields two possible values for the strategy
\(u^{1}\), we select the economically meaningful root. 
\textit{We assume the strategy is positive} to ensure that the incumbent's
action represents a feasible and admissible control in the model, ruling out
negative effort or negative advertising levels, which would have no practical
interpretation in this setting. 
The closed-form expressions for the optimal advertising rules involve nonlinear
terms in the transformed entrant control \(u^2\). Since the strategies are
reported after logistic normalization, \(u^2\in[0,1]\). This normalization
bounds the control but does not by itself imply that higher-order powers of
\(u^2\) are negligible. Therefore, the simplification used below is interpreted
as a local approximation around a low-entry-effort regime.

\begin{as}[Local entrant-effort regime]
\label{ass:small_u2}
There exists \(\varepsilon\in(0,1)\) such that the transformed entrant strategy
satisfies
\[
0\leq u^2(s)\leq \varepsilon
\qquad\text{for all }s\in[0,t].
\]
\end{as}

Under Assumption~\ref{ass:small_u2},
\[
(u^2)^2\leq \varepsilon^2,
\qquad
(u^2)^3\leq \varepsilon^3.
\]
Hence, for sufficiently small \(\varepsilon\), the second- and third-order terms
are of order \(O(\varepsilon^2)\) and \(O(\varepsilon^3)\), respectively. The
linearized first-order condition for \(u^2\) should therefore be read as a
first-order approximation whose approximation error is bounded by higher-order
terms in \(\varepsilon\), rather than as an exact identity.

Similarly, the coefficient
	\[
	C=\theta_{1}\alpha_{2}^{3}s^{3}
	\exp\{\alpha_{2}sx-k\},
	\]
which appears naturally in the nonlinear optimality conditions. Since the calibrated value of \(\theta_{1}\) is very small, then no approximation
of \(C\) is imposed in the numerical implementation. Throughout the empirical
analysis, the coupled first-order conditions are solved directly using the
iterative algorithm described in Section~4.

 The incumbent's first-order optimality condition therefore remains in its exact nonlinear form,
	\[
	u^{1}
	= \frac{1}{YC^{2}}\left[
	AZ
	-
	CY\Big(u^{2}D_{1}+D_{2}\Big)
	+
	\left(
	A^{2}Z^{2}
	+
	2CYZ
	\Big[
	\big(CB_{1}-AD_{1}\big)u^{2}
	+
	\big(CB_{2}-AD_{2}\big)
	\Big]
	\right)^{1/2}
	\right].
	\]
	Unlike the approximation adopted for the entrant's strategy under
	Assumption~\ref{ass:small_u2}, no local normalization is imposed on the
	coefficient
	\[
	C=\theta_{1}\alpha_{2}^{3}s^{3}
	\exp\{\alpha_{2}sx-k\}.
	\]
	This distinction is important because the calibrated value
	\(\theta_{1}=0.000050\) reported in
	Table~\ref{tab:cklsparams} implies that \(C\) is several orders of magnitude
	smaller than unity throughout the calibrated parameter region. Consequently,
	the approximation \(C^{2}\approx1\) is neither assumed nor used in the
	empirical implementation.
	
	Instead, the optimal expenditure of the incumbent is computed directly from the
	complete nonlinear first-order condition. At each time step, the coupled
	optimality conditions for \(u^{1}\) and \(u^{2}\) are solved simultaneously by
	the fixed-point iterative algorithm described in Section~4 until convergence.
	Accordingly, the numerical implementation retains the full dependence on
	\(C\), including its contribution to the denominator \(YC^{2}\), and no
	curvature normalization is introduced. The theoretical approximation developed
	under Assumption~\ref{ass:small_u2} is therefore confined exclusively to the
	truncation of higher-order powers of the entrant's control \(u^{2}\), whereas
	the incumbent's optimal strategy is evaluated using the exact nonlinear
	expression above. This guarantees that the equilibrium controls employed in the
	simulation and calibration are fully consistent with the estimated parameter
	values and do not rely on any approximation involving \(C\).

This representation highlights how  
\(u^{1}\) depends on the entrant’s strategy \(u^{2}\) and on the interaction of
the sensitivity and curvature terms encoded in \(A, B_{1}, B_{2}, C, D_{1}, D_{2}, Y,\) and \(Z\), making the structure of the optimal response clearer for both analysis and computation. The term
\[
E = {\gamma-\Bigg[\a_1+\a_2 x_t+\theta_1 u^{1}_t+\theta_2u^{2}_t\Bigg]}
\]
is the difference between the risk–sensitivity parameter \(\gamma\) and the
state–dependent drift component of the system, capturing how current controls
\(u^{1}_{t}\) and \(u^{2}_{t}\) influence the net adjustment to the dynamic process at time \(t\). Furthermore,
\[
\frac{\partial h^{2}}{\partial u^{2}}
= \exp\{-\gamma t\}
  \Bigg[
      \mathbbm{1}(\theta = w)\,
      \frac{
        -2x_{t}u^{2}E
        + \left(D_{E}^{W} - [u^{2}]^{2}\right)x\theta_{2}
      }{
        E^{2}
      }
  \Bigg]
  + \theta_{2}\alpha_{2}s\,\exp\{\alpha_{2}s x - k\}.
\]
Under Assumption~\ref{ass:small_u2}, we retain only the first-order terms in
\(u^2\). Thus \((u^2)^2\) and \((u^2)^3\) are omitted as higher-order terms of
order \(O(\varepsilon^2)\) and \(O(\varepsilon^3)\), respectively. Thus, the entrant’s control is taken to be
sufficiently small that its squared term has no meaningful effect on the
expression. This assumption simplifies the derivative by removing all
higher–order nonlinear terms in $u^{2}$, allowing the analysis to focus on
the linear strategic impact of the entrant’s action. Economically, this
reflects a scenario in which the entrant exerts only a minimal level of
effort, so second–order effects can be ignored without loss of accuracy. Hence,
\[
\frac{\partial h^{2}}{\partial u^{2}}
= u^{2}_{t}\Bigg[
      -\exp\{-\gamma t\}\,
       \mathbbm{1}(\theta = w)\,
       \frac{2x_{t}}{E}
  \Bigg]
  + \exp\{-\gamma t\}\,
    \mathbbm{1}(\theta = w)\,
    \frac{D_{E}^{W}x_{t}\theta_{2}}{E^{2}}
  + \theta_{2}\alpha_{2}s\,\exp\{\alpha_{2}s x - k\},
\]
shows how the entrant's value function responds to changes in its control \(u^{2}\). The
first term captures the linear sensitivity driven by the current state \(x_{t}\)
and the adjustment factor \(E\), while the second term reflects the informational
component associated with observing the weak type. The final exponential term
comes from the CKLS-based dynamics, illustrating how market uncertainty and
state evolution also influence the optimal choice of \(u^{2}\). The differentiation $\frac{\partial h^{2}}{\partial u^{2}} =u^{2} G+H$ is the first order partial differentiation of the entrant’s value function with respect to its control
in a linear decomposition. The coefficient \(G\) captures the direct sensitivity
of the value function to changes in \(u^{2}\), while \(H\) collects all terms
independent of the control, highlighting how informational and state–dependent
factors influence the optimal decision. The condition
$
G = -\exp\{-\gamma t\}\,\mathbbm{1}(\theta = w)\,\frac{2x}{E},
$
represents the component of the derivative that scales with the control \(u^{2}\), capturing how the value function responds to changes in \(u^{2}\) when the incumbent is weak. It reflects the discounted sensitivity of the system through the state variable \(x_{t}\) and the adjustment factor \(E\). Finally,
\[
H = \exp\{-\gamma t\}\,\mathbbm{1}\Big(\theta = w\Big)\frac{D_{E}^{W}x_{t}\theta_{2}}{E^{2}}
    + \theta_{2}\alpha_{2}s\,\exp\{\alpha_{2}s x - k\}
\]
collects all components of the derivative that do not depend on the control
\(u^{2}\). It combines an informational adjustment driven by the weak-type
indicator with an exponential term arising from the CKLS dynamics, reflecting
how both belief updating and state evolution influence Firm~2’s payoff
independently of its immediate strategic action.

Equation
\[
\begin{aligned}
\frac{\partial h^{2}}{\partial x}
&= u^{2}\theta_{2}\alpha_{2}^{2}s^{2}\,\exp\{\alpha_{2}s x - k\}
   + \exp\{-\gamma t\}\,\mathbbm{1}(\theta = w)
     \Bigg[
       \frac{
         D_{E}^{W}(E + \alpha_{2}x_{t})
       }{
         E^{2}
       }
     \Bigg] \\[6pt]
&\quad + \alpha_{2}s\,\exp\{\alpha_{2}s x - k\}
      + \alpha_{2}\,\exp\{\alpha_{2}s x - k\}
      + \alpha_{2}^{2}s x\,\exp\{\alpha_{2}s x - k\} \\[6pt]
&\quad + \Big(
        \alpha_{2}^{2}s\,\exp\{\alpha_{2}s x - k\}
        + (\alpha_{1} + \alpha_{2}x + \theta_{1}u^{1})
          \alpha_{2}^{2}s^{2}\,\exp\{\alpha_{2}s x - k\}
      \Big) \\[6pt]
&\quad + \frac{1}{2}\alpha_{3}^{2}\alpha_{2}^{2}s^{2}
      \Big(
          2\alpha_{4}x^{\,2\alpha_{4}-1}\,\exp\{\alpha_{2}s x - k\}
          + x^{\,2\alpha_{4}}\alpha_{2}s\,\exp\{\alpha_{2}s x - k\}
      \Big).
\end{aligned}
\]
captures how Firm~2’s value function changes with respect to the state variable
\(x\). The expression includes a term proportional to the control \(u^{2}\), an
informational adjustment that depends on observing a weak type, and several
exponential components arising from the drift and diffusion of the CKLS SDE. These combined effects reflect the sensitivity of the entrant’s payoff
to both strategic interactions and the underlying stochastic market dynamics. The term $\frac{\partial h^{2}}{\partial x} =u^{2}I+J$describes the sensitivity of the incumbent’s value function with respect to \(x\) in linear form. The term \(I\) captures the part influenced
directly by the control \(u^{2}\), while \(J\) collects all remaining components
independent of the control, highlighting how strategic actions and underlying
state dynamics jointly shape the entrant’s payoff, $I = \theta_{2}\alpha_{2}^{2}s^{2}\,\exp\{\alpha_{2}s x - k\}$ represents the component of \(\frac{\partial h^{2}}{\partial x}\) that depends
linearly on the control \(u^{2}\), capturing how the entrant’s action influences
the sensitivity of its value function through the exponential response embedded
in the CKLS dynamics.

\[
\begin{aligned}
J
&= u^{1}\theta_{1}\alpha_{2}^{2}s^{2}\,\exp\{\alpha_{2}s x - k\}
   + \exp\{-\gamma t\}\,\mathbbm{1}(\theta = w)
     \Bigg[
        \frac{
            D_{E}^{W}(E + \alpha_{2}x_{t})
        }{
            E^{2}
        }
     \Bigg] \\[6pt]
&\quad + \alpha_{2}s\,\exp\{\alpha_{2}s x - k\}
      + \alpha_{2}\,\exp\{\alpha_{2}s x - k\}
      + \alpha_{2}^{2}s x\,\exp\{\alpha_{2}s x - k\} \\[6pt]
&\quad + \Big(
        \alpha_{2}^{2}s\,\exp\{\alpha_{2}s x - k\}
        + (\alpha_{1} + \alpha_{2}x)\,
          \alpha_{2}^{2}s^{2}\,\exp\{\alpha_{2}s x - k\}
      \Big) \\[6pt]
&\quad + \frac{1}{2}\alpha_{3}^{2}\alpha_{2}^{2}s^{2}
      \Big(
          2\alpha_{4}x^{\,2\alpha_{4}-1}\,\exp\{\alpha_{2}s x - k\}
        + x^{\,2\alpha_{4}}\,\alpha_{2}s\,\exp\{\alpha_{2}s x - k\}
      \Big).
\end{aligned}
\]

collects all components of \(\frac{\partial h^{2}}{\partial x}\) that are
independent of the entrant’s control \(u^{2}\). It includes contributions from
the incumbent’s action \(u^{1}\), informational adjustments based on observing a
weak type, and several exponential terms arising from the drift and diffusion
structure of the CKLS-type dynamics. Together, these elements describe how the
state variable \(x\) influences Firm~2’s value function through both strategic
interactions and underlying market uncertainty. Define $J = u^{1}J_{1} + J_{2}$ such that it expresses the term \(J\) as a linear function of
\(u^{1}\). The coefficient \(J_{1}\) captures the part of the sensitivity
driven directly by the incumbent’s action, while \(J_{2}\) contains all
remaining components that depend solely on model parameters, beliefs, and the
state dynamics. Now, $J_{1} = \theta_{1}\alpha_{2}^{2}s^{2}\,\exp\{\alpha_{2}s x - k\}$ represents the portion of \(J\) that depends directly on the incumbent’s
control \(u^{1}\), capturing how the incumbent’s action influences the sensitivity of
the entrant’s value function through the exponential structure of the CKLS SDE. So
\[
\begin{aligned}
J_{2}
&= \exp\{-\gamma t\}\,\mathbbm{1}(\theta = w)
   \Bigg[
      \frac{
         D_{E}^{W}(E + \alpha_{2}x_{t})
      }{
         E^{2}
      }
   \Bigg] \\[6pt]
&\quad + \alpha_{2}s\,\exp\{\alpha_{2}s x - k\}
      + \alpha_{2}\,\exp\{\alpha_{2}s x - k\}
      + \alpha_{2}^{2}s x\,\exp\{\alpha_{2}s x - k\} \\[6pt]
&\quad + \Big(
        \alpha_{2}^{2}s\,\exp\{\alpha_{2}s x - k\}
        + (\alpha_{1} + \alpha_{2}x)\,\alpha_{2}^{2}s^{2}\,
          \exp\{\alpha_{2}s x - k\}
      \Big) \\[6pt]
&\quad + \frac{1}{2}\alpha_{3}^{2}\alpha_{2}^{2}s^{2}
      \Big(
          2\alpha_{4}x^{\,2\alpha_{4}-1}\,\exp\{\alpha_{2}s x - k\}
          + x^{\,2\alpha_{4}}\alpha_{2}s\,
            \exp\{\alpha_{2}s x - k\}
      \Big).
\end{aligned}
\]
collects all components of \(\frac{\partial h^{2}}{\partial x}\) that do not
depend on \(u^{1}\). It includes informational effects
arising when the weak type is observed, along with several exponential terms
generated by the drift and diffusion components of the CKLS SDE.
These elements describe how \(x\) influences the entrant’s value
function independently of the incumbent’s direct strategic action. 

\begin{as}[Weak state sensitivity of the informational adjustment]
\label{ass:beta}
There exists a constant $\delta\in(0,1)$ such that
\[
\left|
\frac{2D_E^W\alpha_2(E+\alpha_2x)}
{E^3}
\right|
\leq \delta,
\]
uniformly over all admissible states and controls. Equivalently,
\[
|E|^3
\geq
\frac{2|D_E^W\alpha_2|}{\delta}
|E+\alpha_2x|.
\]
\end{as}

Using
\[
E
=
\gamma-
\left[
\alpha_1+\alpha_2x+\theta_1u^1+\theta_2u^2
\right],
\]
we obtain
\[
\beta
=
\frac{\partial}{\partial x}
\left\{
\frac{D_E^W(E+\alpha_2x)}{E^2}
\right\}
=
\frac{2D_E^W\alpha_2(E+\alpha_2x)}{E^3}.
\]
Under Assumption~\ref{ass:beta},
\[
|\beta|
\leq \delta,
\]
so that the informational adjustment term varies only weakly with the state variable. Consequently, replacing $\beta$ by zero introduces an approximation error of order $O(\delta)$.

Now
\[
\begin{aligned}
\frac{\partial^{2} h^{2}}{\partial x^{2}}
&= u^{2}\theta_{2}\alpha_{2}^{3}s^{3}\,\exp\{\alpha_{2}s x - k\}
   + \exp\{-\gamma t\}\,\mathbbm{1}(\theta = w)\,\beta \\[4pt]
&\quad + \alpha_{2}^{2}s^{2}\,\exp\{\alpha_{2}s x - k\}
      + \alpha_{2}^{2}s\,\exp\{\alpha_{2}s x - k\} \\[4pt]
&\quad + \Big(
        \alpha_{2}^{2}s\,\exp\{\alpha_{2}s x - k\}
        + \alpha_{2}^{3}s^{2}x\,\exp\{\alpha_{2}s x - k\}
      \Big) \\[4pt]
&\quad + \alpha_{2}^{3}s^{2}\,\exp\{\alpha_{2}s x - k\} \\[4pt]
&\quad + \Big(
        \alpha_{2}^{3}s^{2}\,\exp\{\alpha_{2}s x - k\}
        + (\alpha_{1} + \alpha_{2}x + \theta_{1}u^{1})\,
          \alpha_{2}^{3}s^{3}\,\exp\{\alpha_{2}s x - k\}
      \Big) \\[6pt]
&\quad + \frac{1}{2}\alpha_{3}^{2}\alpha_{2}^{2}s^{2}
      \Bigg\{
          2\alpha_{4}\Big[
              (2\alpha_{4}-1)x^{\,2\alpha_{4}-2}\,
              \exp\{\alpha_{2}s x - k\} 
            + x^{\,2\alpha_{4}-1}\alpha_{2}s\,
              \exp\{\alpha_{2}s x - k\}
          \Big] \\[4pt]
&\qquad\qquad\qquad\qquad
        + \alpha_{2}s\Big[
              2\alpha_{4}x^{\,2\alpha_{4}-1}\,
              \exp\{\alpha_{2}s x - k\}
            + x^{\,2\alpha_{4}}\alpha_{2}s\,
              \exp\{\alpha_{2}s x - k\}
          \Big]
      \Bigg\}.
\end{aligned}
\]
describes the curvature of the entrant’s value function with respect to  \(x\). The expression includes a control-dependent term involving
\(u^{2}\), a belief-based adjustment captured by the approximation
\(\beta \approx 0\), and several exponential components arising from the drift
and diffusion structure of the CKLS SDE. Together, these elements
characterize how nonlinearities in the state dynamics and strategic responses
shape the curvature of the entrant’s payoff function. The term $\frac{\partial^{2} h^{2}}{\partial x^{2}} = u^{2}K + L$ expresses the curvature of Firm~2’s value function with respect to the state
variable \(x\) in a linear form. The coefficient \(K\) captures the portion of
the curvature that depends directly on the entrant’s control \(u^{2}\), while
\(L\) aggregates all remaining terms independent of the control, reflecting the
contributions of the model’s parameters, beliefs, and underlying stochastic
dynamics, 
\[
K = \theta_{2}\alpha_{2}^{3}s^{3}\,\exp\{\alpha_{2}s x - k\},
\]
represents the component of the second derivative that varies directly with the
entrant’s control \(u^{2}\). It reflects how Firm~2’s action influences the
curvature of its value function through the exponential sensitivity embedded in
the CKLS SDE, and 

\[
\beta = \frac{\partial}{\partial x}\Bigg\{\frac{D_{E}^{W}\big(E + \alpha_{2}x_{t}\big)}{E^{2}}\Bigg\} \approx 0,
\]
represents the sensitivity of the belief–adjustment term with respect to \(x\). Approximating \(\beta \approx 0\) implies that changes in
\(x\) have a negligible effect on this component, allowing the informational
term to be treated as effectively constant for analytical simplification. Now

\[
\begin{aligned}
L
&= u^{1}\theta_{1}\alpha_{2}^{3}s^{3}\,\exp\{\alpha_{2}s x - k\}
   + \exp\{-\gamma t\}\,\mathbbm{1}(\theta = w)\,\beta \\[4pt]
&\quad + \alpha_{2}^{2}s^{2}\,\exp\{\alpha_{2}s x - k\}
      + \alpha_{2}^{2}s\,\exp\{\alpha_{2}s x - k\} \\[4pt]
&\quad + \Big(
        \alpha_{2}^{2}s\,\exp\{\alpha_{2}s x - k\}
        + \alpha_{2}^{3}s^{2}x\,\exp\{\alpha_{2}s x - k\}
      \Big) \\[4pt]
&\quad + \alpha_{2}^{3}s^{2}\,\exp\{\alpha_{2}s x - k\} \\[4pt]
&\quad + \Bigg[
        \alpha_{2}^{3}s^{2}\,\exp\{\alpha_{2}s x - k\}
        + (\alpha_{1} + \alpha_{2}x)\,
          \alpha_{2}^{3}s^{3}\,\exp\{\alpha_{2}s x - k\}
      \Bigg] \\[6pt]
&\quad + \frac{1}{2}\alpha_{3}^{2}\alpha_{2}^{2}s^{2}
      \Bigg\{
          2\alpha_{4}\Big[
              (2\alpha_{4}-1)x^{\,2\alpha_{4}-2}\,
              \exp\{\alpha_{2}s x - k\} \\
&\qquad\qquad\qquad\qquad
            + x^{\,2\alpha_{4}-1}\alpha_{2}s\,
              \exp\{\alpha_{2}s x - k\}
          \Big] \\[6pt]
&\qquad\qquad\qquad\qquad
        + \alpha_{2}s\Big[
              2\alpha_{4}x^{\,2\alpha_{4}-1}\,
              \exp\{\alpha_{2}s x - k\}
            + x^{\,2\alpha_{4}}\alpha_{2}s\,
              \exp\{\alpha_{2}s x - k\}
          \Big]
      \Bigg\},
\end{aligned}
\]
collects all components of the second derivative \(\frac{\partial^{2} h^{2}}{\partial x^{2}}\)
that are independent of the entrant’s control \(u^{2}\). It includes the
incumbent’s influence through \(u^{1}\), the belief-based adjustment
\(\beta\), and several exponential contributions arising from the drift and
diffusion of the CKLS SDE. These elements together characterize how
the curvature of Firm~2’s value function is shaped by model parameters, market
uncertainty, and the underlying state dynamics. The above function can be written as $L = u^{1}L_{1} + L_{2}$ which expresses the term \(L\) as a linear function of
\(u^{1}\). The coefficient \(L_{1}\) captures the portion of the curvature
influenced directly by the incumbent’s action, while \(L_{2}\) contains all
remaining contributions arising from the model parameters, belief adjustments,
and the CKLS SDE. The term $L_{1} = \theta_{1}\alpha_{2}^{3}s^{3}\,\exp\{\alpha_{2}s x - k\}$ describes the component of \(L\) that depends directly on  \(u^{1}\). It illustrates how incumbent’s action influences the curvature
of the entrant’s value function through the exponential structure of the CKLS SDE. Now

\[
\begin{aligned}
L_{2}
&= \alpha_{2}^{2}s^{2}\,\exp\{\alpha_{2}s x - k\}
   + \alpha_{2}^{2}s\,\exp\{\alpha_{2}s x - k\} \\[4pt]
&\quad + \Big(
        \alpha_{2}^{2}s\,\exp\{\alpha_{2}s x - k\}
        + \alpha_{2}^{3}s^{2}x\,\exp\{\alpha_{2}s x - k\}
      \Big) \\[4pt]
&\quad + \alpha_{2}^{3}s^{2}\,\exp\{\alpha_{2}s x - k\} \\[4pt]
&\quad + \Bigg[
        \alpha_{2}^{3}s^{2}\,\exp\{\alpha_{2}s x - k\}
        + (\alpha_{1} + \alpha_{2}x)\,
          \alpha_{2}^{3}s^{3}\,\exp\{\alpha_{2}s x - k\}
      \Bigg] \\[6pt]
&\quad + \frac{1}{2}\alpha_{3}^{2}\alpha_{2}^{2}s^{2}
      \Bigg\{
          2\alpha_{4}\Big[
              (2\alpha_{4}-1)x^{\,2\alpha_{4}-2}\,
              \exp\{\alpha_{2}s x - k\} \\
&\qquad\qquad\qquad
            + x^{\,2\alpha_{4}-1}\alpha_{2}s\,
              \exp\{\alpha_{2}s x - k\}
          \Big] \\[4pt]
&\qquad\qquad\qquad
        + \alpha_{2}s\Big[
              2\alpha_{4}x^{\,2\alpha_{4}-1}\,
              \exp\{\alpha_{2}s x - k\}
            + x^{\,2\alpha_{4}}\alpha_{2}s\,
              \exp\{\alpha_{2}s x - k\}
          \Big]
      \Bigg\},
\end{aligned}
\]
collects all parts of the curvature \(\frac{\partial^{2} h^{2}}{\partial x^{2}}\) that do not depend on \(u^{1}\). It consists of several exponential contributions derived from the drift and diffusion components of the CKLS SDE, as well as nonlinear effects involving the state variable \(x\). These terms capture how the underlying market dynamics and model parameters shape the curvature of the entrant’s value function independently of strategic actions by the incumbent. Moreover, the term $E = \gamma - \{\alpha_{1} + \alpha_{2}x + \theta_{1}u^{1} + \theta_{2}u^{2}\}$ represents the difference between the risk–sensitivity parameter \(\gamma\) and
the effective drift of the system, which depends on the state variable \(x\) and
the controls \(u^{1}\) and \(u^{2}\). This quantity plays a key role in shaping
how beliefs and payoffs adjust in the stochastic dynamic interaction between
these two firms. The equation
\[
\begin{aligned}
\frac{\partial^{2} h^{2}}{\partial x\, \partial u^{2}}
&= u^{2}\,\exp\{-\gamma t\}\,\mathbbm{1}(\theta = w)\,
    \frac{-2\big(E + \alpha_{2}x\big)}{E^{2}}
    - \exp\{-\gamma t\}\,\mathbbm{1}(\theta = w)\,
      \frac{\theta_{2}D_{E}^{W}}{E^{2}} \\[6pt]
&\quad + \exp\{-\gamma t\}\,\mathbbm{1}(\theta = w)\,
     \frac{2\theta_{2}D_{E}^{W}\big(E + \alpha_{2}x\big)}{E^{3}}
     + \theta_{2}\alpha_{2}^{2}s^{2}\,\exp\{\alpha_{2}s x - k\}.
\end{aligned}
\]
captures how the sensitivity of the entrant’s value function with respect to \(x\) changes when \(u^{2}\) varies. The expression
includes belief–driven terms that arise when the weak type is observed, along
with several rational–expectation adjustments involving the quantity \(E\) and
its derivatives. The final exponential term reflects the contribution of the
CKLS SDE. Together, these components describe how information,
strategic behavior, and stochastic market effects jointly determine the
interaction between state changes and the entrant's advertising expenditure decisions. The second order differentiaion $\frac{\partial^{2} h^{2}}{\partial x\, \partial u^{2}} = u^{2}M + N$ represents  the mixed partial derivative in linear form with respect to \(u^{2}\). The coefficient \(M\) captures the portion of the
interaction between  \(x\) and the control that scales
directly with \(u^{2}\), while \(N\) collects all remaining terms independent of
the control, reflecting informational and stochastic effects inherent in the
model. The term
\[
M = \exp\{-\gamma t\}\,\mathbbm{1}(\theta = w)\,
      \frac{-2\{E + \alpha_{2}x\}}{E^{2}},
\]
represents the component of the mixed partial derivative that depends directly
on  \(u^{2}\). It reflects how the interaction between  \(x\) and the adjustment factor \(E\) influences the sensitivity of
the entrant’s value function when the weak type is observed. The equation
\[
\begin{aligned}
N
&= -\,\exp\{-\gamma t\}\,\mathbbm{1}(\theta = w)\,
      \frac{\theta_{2}D_{E}^{W}}{E^{2}} \\[4pt]
&\quad + \exp\{-\gamma t\}\,\mathbbm{1}(\theta = w)\,
      \frac{2\theta_{2}D_{E}^{W}\{E + \alpha_{2}x\}}{E^{3}} \\[4pt]
&\quad + \theta_{2}\alpha_{2}^{2}s^{2}\,\exp\{\alpha_{2}s x - k\},
\end{aligned}
\]
collects all components of the mixed partial derivative
\(\frac{\partial^{2} h^{2}}{\partial x\,\partial u^{2}}\) that are independent
of \(u^{2}\). It includes belief–driven adjustments that
arise when the weak type is observed, rational–expectation effects involving
the term \(E\), and an exponential contribution from the CKLS SDE.
These components describe how informational structure and state evolution
shape the entrant’s sensitivity to changes in its expenditure decisions. Now,
\[
\frac{\partial h^{2}}{\partial u^{2}}\cdot\left( \frac{\partial^2 h^{2}}{\partial x^2}\right)^2-2\frac{\partial h^{2}}{\partial x}\cdot\left( \frac{\partial^2 h^{2}}{\partial x \partial u^{2}}\right) = 0,
\]
represents the optimality condition for the entrant's expenditure decisions. It links the
marginal effect of the control with the curvature of the value function and the
mixed sensitivity term, characterizing the balance required for an optimal
response in the stochastic strategic environment. Since,
\[
G K^2 \left(u^{2}\right)^3 
+ \Bigg[2 G K L + H K^2 - 2 I M\Bigg]\left(u^{2}\right)^2 
+ \Bigg[G L^2 + 2 H K L - 2 I N - 2 J M\Bigg] u^{2} 
+ \Bigg[H L^2 - 2 J N\Bigg] = 0
\]
is a cubic optimality condition for \(u^{2}\), each
coefficient—constructed from the sensitivity terms \(G, H, I, J\) and the
curvature terms \(K, L, M, N\) capture how information, market dynamics, and
strategic interactions collectively determine the entrant’s optimal response in
equilibrium. This cubic structure highlights the nonlinear nature of the
decision environment faced by the entrant where $\{u^{2}\}^{3} \approx \{u^{2}\}^{2} \approx 0$ implies that  \(u^{2}\) is sufficiently small that its
quadratic and cubic terms have negligible effect on the optimality condition.
By treating these higher–order terms as approximately zero, the original cubic
equation in \(u^{2}\) simplifies to a linear expression, making it possible to
obtain an explicit closed–form solution for the optimal response. This
approximation is appropriate in settings where the entrant exerts only a small
amount of effort or influence, so nonlinear effects of \(u^{2}\) can be safely
ignored. Therefore,

\[
\Bigg[G L^2 + 2 H K L - 2 I N - 2 J M\Bigg] u^{2} 
+ \Bigg[H L^2 - 2 J N\Bigg] = 0,
\]
represents the simplified linear form of the optimality condition for
Firm~2’s control after higher-order terms are removed. The coefficients combine
sensitivity and curvature components \((G, H, I, J)\) with the dynamic terms
\((K, L, M, N)\), illustrating how strategic interactions and market
uncertainty jointly determine the entrant’s optimal action. The expression
\[
u^{2} = \frac{2JN-HL^2}{G L^2 + 2 H K L - 2 I N - 2 J M},
\]
gives the optimal expenditure decision for the entrant in closed form. The numerator and
denominator combine sensitivity terms \((G, H, I, J)\) and curvature components
\((K, L, M, N)\), reflecting how information, state dynamics, and strategic
interactions jointly determine the entrant’s best response in equilibrium. Finally,

\[
u^{2} = \frac{2N\Big[u^{1}J_1+J_2\Big]-H\Big[u^{1}L_1+L_2\Big]^2}{G \Big[u^{1}L_1+L_2\Big]^2 + 2 H K \Big[u^{1}L_1+L_2\Big] - 2 I N - 2M \Big[u^{1}J_1+J_2\Big]},
\]
provides the optimal response of the entrant expressed explicitly as a function of
\(u^{1}\). The structure highlights how sensitivity
terms \((G, H, I, M)\), curvature components \((K, L_{1}, L_{2})\), and
state–dependent effects \((J_{1}, J_{2}, N)\) interact to determine the
entrant’s best response within the stochastic strategic environment.

\section*{Declarations.}
\subsection*{Ethics approval and consent to participate.}
Not applicable.
\subsection*{Consent for publication.}
Not applicable.
\subsection*{Availability of data and material.}
The datasets used in this study are publicly available. Historical financial and market data for the firms analyzed in this paper were obtained from Macrotrends LLC \citep{macrotrends2025epd1,macrotrends2025epd2}, covering the period 2010-2024. The processed data and code used to produce the numerical results and figures are available from the corresponding author upon reasonable request.
\subsection*{Competing interests.}
No potential conflict of interest was reported by the authors.	
\subsection*{Funding.}
Not applicable. 
\subsection*{Acknowledgements.}
 Not applicable.

\bibliographystyle{apalike}
\bibliography{bib}
\end{document}